\PassOptionsToPackage{final,colorlinks,linkcolor={blue},citecolor={blue},urlcolor={red},breaklinks=true}{hyperref}
\documentclass[sigconf,UKenglish,nonacm,final]{acmart}
\usepackage{ifthen}
\usepackage{ifdraft}
\usepackage{lineno}
\usepackage{microtype}

\usepackage[capitalize,nameinlink]{cleveref}

\newboolean{proofsinappendix}
\ifoptionfinal{
  \setboolean{proofsinappendix}{true}
}{
\setboolean{proofsinappendix}{false}
}

\newcommand{\mycomment}[1]{}  %

\title[Nominal Automata with Name Deallocation]{Nominal Automata with Name Deallocation}

\author{Simon Prucker}
\orcid{0009-0000-2317-5565}
\affiliation{
  \institution{Friedrich-Alexander-Universität Erlangen-Nürnberg}
  \city{Erlangen}
  \country{Germany}
}
\email{simon.prucker@fau.de}

\author{Stefan Milius}
\orcid{0000-0002-2021-1644}
\affiliation{
  \institution{Friedrich-Alexander-Universität Erlangen-Nürnberg}
  \city{Erlangen}
  \country{Germany}
}
\email{stefan.milius@fau.de}

\author{Lutz Schr\"oder}
\orcid{0000-0002-3146-5906}
\affiliation{
  \institution{Friedrich-Alexander-Universität Erlangen-Nürnberg}
  \city{Erlangen}
  \country{Germany}
}
\email{lutz.schroeder@fau.de}

\usepackage[inline]{enumitem}
\setlist[enumerate,1]{label=(\arabic*),font=\normalfont,align=left,leftmargin=0pt,labelindent=0pt,
listparindent=\parindent,labelwidth=0pt,itemindent=!,topsep=3pt,parsep=0pt,itemsep=3pt,start=1}
\setlist[enumerate,2]{label=(\alph*),font=\normalfont,labelindent=*,leftmargin=*,topsep=3pt,start=1}
\setlist[itemize]{labelindent=*,leftmargin=*,topsep=5pt,itemsep=3pt}
\setlist[description]{labelindent=*,leftmargin=*,itemindent=-1 em}

\usepackage[utf8]{inputenc}
\usepackage{graphicx}
\usepackage{adjustbox}
\usepackage{tikz}
\usetikzlibrary{decorations.pathreplacing,calc,graphs,shapes,arrows,positioning,fit,automata}
\usepackage{tikz-cd}

\usepackage{amsmath,amsthm}
\usepackage{mathtools}
\usepackage{dsfont}

\usepackage{thmtools}
\usepackage{bussproofs}
\usepackage{fancyvrb}
\usepackage{scalefnt}
\usepackage{ifthen}
\usepackage{xspace}
\usepackage{float}
\usepackage{verbatim}
\usepackage{minted} 
\usepackage{latexsym}
\usepackage{csquotes}
\usepackage{tasks}
\usepackage{wrapfig}
\usepackage{relsize}
\usepackage{iflang}
\usepackage{dirtytalk}
\usepackage{pdfpages}
\usepackage{environ}
\usepackage{newfile}
\usepackage[printonlyused]{acronym}
\usepackage{lplfitch}
\usepackage{newunicodechar}
\usepackage{booktabs}
\usepackage{etoolbox}
\usepackage[layout=footnote,marginclue]{fixme}
\FXRegisterAuthor{sp}{asp}{SP}
\FXRegisterAuthor{sm}{apn}{SM}
\FXRegisterAuthor{ls}{als}{LS}
\FXRegisterAuthor{tw}{atw}{TW}
\usepackage[type=CC, modifier=by-sa,version=4.0]{doclicense}
\usepackage{embrac}
\usepackage{babel}

\usepackage[notref,notcite]{showkeys}

\usepackage{seqsplit}
\usepackage{xstring}
\newcommand{\defaultshowkeysformat}[1]{%
\StrSubstitute{#1}{ }{\textvisiblespace}[\TEMP]%
\parbox[t]{\marginparwidth}{\raggedright\normalfont\small\ttfamily\(\{\){\color{red!50!black}\expandafter\seqsplit\expandafter{\TEMP}}\(\}\)}%
}

\renewcommand*\showkeyslabelformat[1]{%
\noexpandarg%
\defaultshowkeysformat{#1}%
}

\addto\extrasenglish{%

}

\makeatletter
\ifthenelse{\boolean{proofsinappendix}}{
  \newoutputstream{proofstream}
  \openoutputfile{\jobname-proofs.out}{proofstream}
}{
}

\newcommand{\proofappendixbegin}[2]{%
  \phantomsection%
  \subsection*{\textbf{#1~\autoref{#2}}}%
  \addcontentsline{toc}{subsection}{#1~\autoref{#2}}%
  \label{#2:proof}%
  \def\proofappendix@qedsymbolmissing{\qed}
}
\newcommand{\proofappendixend}{%
     \proofappendix@qedsymbolmissing%
}
\let\oldqedhere\qedhere
\def\qedhere{%
  \ifproofsinappendix%
    \global\def\proofappendix@qedsymbolmissing{}\qed%
  \else\oldqedhere\fi%
  }

\newenvironment{proofhere}[2][Proof of]{%
  \begin{proof}
}{
  \end{proof}
}

\ifthenelse{\boolean{proofsinappendix}}{%
  \NewEnviron{proofappendix}[2][Proof of]{%
    \addtostream{proofstream}{\noexpand\proofappendixbegin{#1}{#2}}%
    \addtostream{proofstream}{\noexpand\takeout{
      WARNING: this file is auto generated from the main tex file!
      Do not edit manually!
    }}
    \expandafter\addtostream{proofstream}{\expandonce{\BODY}}%
    \addtostream{proofstream}{\noexpand\proofappendixend}%
    \addtostream{proofstream}{}%
  }%
  \NewEnviron{inappendix}{%
    \addtostream{proofstream}{\noexpand\takeout{
      WARNING: this file is auto generated from the main tex file!
      Do not edit manually!
    }}
    \expandafter\addtostream{proofstream}{\expandonce{\BODY}}%
    \addtostream{proofstream}{}%
  }%
}{%
  \newenvironment{proofappendix}[2][Proof of]{%
      \begin{proofhere}[#1]{#2}
  }{
      \end{proofhere}
  }
  \newenvironment{inappendix}{%
    \marginpar{\color{blue}In Appendix!}
    \color{blue}
  }{
    \color{black}
  }
}
\makeatother

\newcommand{\takeout}[1]{} %

\newcounter{inlineequation}
\def\resettheorembrackets{
  \def\theorembracketopen{(}
  \def\theorembracketclose{)}
}
\resettheorembrackets
\makeatletter
\def\@spopargbegintheorem#1#2#3#4#5{\trivlist
  \item[\hskip\labelsep{#4#1\ #2}]{#4{\theorembracketopen}#3{\theorembracketclose}\@thmcounterend\ }#5}
\makeatother

\newcommand{\resetCurThmBraces}{%
  \gdef\curThmBraceOpen{(}%
  \gdef\curThmBraceClose{)}}
\resetCurThmBraces
\newcommand{\removeThmBraces}{%
  \gdef\curThmBraceOpen{}%
  \gdef\curThmBraceClose{}}

\patchcmd{\thmhead}{(#3)}{\curThmBraceOpen #3\curThmBraceClose }{}{}

\theoremstyle{plain}
\newtheorem{theorem}{Theorem}[section]
\newtheorem{lemma}[theorem]{Lemma}
\newtheorem{cor}[theorem]{Corollary}
\newtheorem{proposition}[theorem]{Proposition}

\theoremstyle{definition}
\newtheorem{defn}[theorem]{Definition}
\newtheorem{construction}[theorem]{Construction}
\newtheorem{example}[theorem]{Example}
\newtheorem{rem}[theorem]{Remark}
\newtheorem{notn}[theorem]{Notation}

\newtheorem{fact}[theorem]{Fact}

\newenvironment{proofsketch}{\begin{proof}[Proof sketch]}{\end{proof}}

\newcommand{\restr}[2]{\left.\kern-\nulldelimiterspace #1 \vphantom{\big|}\right|_{#2}}

\newunicodechar{∀}{\ensuremath{\forall}}
\newunicodechar{∃}{\ensuremath{\exists}}
\newunicodechar{×}{\ensuremath{\times}}
\newunicodechar{ø}{\ensuremath{\emptyset}}
\newunicodechar{≤}{\ensuremath{\le}}
\newunicodechar{∈}{\ensuremath{\in}}
\newunicodechar{→}{\ensuremath{\to}}
\newunicodechar{⊆}{\ensuremath{\subseteq}}
\newunicodechar{⊗}{\ensuremath{\otimes}}
\newunicodechar{∧}{\ensuremath{\wedge}}

\newcommand{\by}[1]{(\text{#1})}

\newcommand{\names}{\mathbb{A}}
\newcommand{\N}{\mathds{N}}
\newcommand{\A}{\mathbb{A}}

\newcommand{\aeq}{\equiv_\alpha}

\newcommand{\ahs}{\hat{\names}^*}
\newcommand{\db}{\mathsf{db}}
\newcommand{\pow}{\mathcal{P}}
\newcommand{\pfin}{\mathcal{P}_\omega}
\newcommand{\trans}[1]{\overset{#1}\longrightarrow}
\DeclareMathOperator{\supp}{\mathsf{supp}}  

\newcommand{\ah}{\hat\names}

\DeclareMathOperator{\cleq}{\uparrow_{\mathsf{eq}}}

\newcommand{\la}{L_\alpha}

\newcommand{\perm}{\mathsf{Perm}}

\newcommand{\fresh}{\mathrel{\#}}
\newcommand{\parto}{\rightharpoonup}
\newcommand{\dg}{\mathsf{deg}}
\newcommand{\drop}[1]{\mathord{\downarrow}_{#1}}
\newcommand{\disc}{\mathsf{disc}}
\newcommand{\D}{\mathcal{D}}
\newcommand{\shs}{\hat{S}^*}
\DeclareMathOperator{\clal}{\uparrow_\alpha}

\newcommand{\orb}{\mathcal{O}}
\newcommand{\permG}{\perm(\names)}
\newcommand{\RNS}{\mathsf{RNS}}

\newcommand{\id}{\mathrm{id}}

\newcommand{\set}[1]{\{#1\}}
\newcommand{\setw}[2]{\{#1 \mid #2\}}

\newcommand{\epito}{\twoheadrightarrow}

\newcommand{\midmid}{\hspace{0.2ex}{\rule[-0.1ex]{0.6pt}{1.65ex}}\hspace{0.2ex}}

\newcommand{\newletter}[1]{{\midmid}#1}

\newcommand{\mybar}[3]{%
  \mathrlap{\hspace{#2}\overline{\scalebox{#1}[1]{\phantom{\ensuremath{#3}}}}}\ensuremath{#3}}
\newcommand{\barNames}{{\mybar{0.7}{1.25pt}{\names}}}

\newcommand{\zz}{\mathrm{Z\kern-.6em\raise-0.5ex\hbox{Z}}}

\newcommand{\lb}{\lbrack\!\!\langle}
\newcommand{\rb}{\rangle\!\!\rbrack}
\newcommand{\braket}[1]{\langle #1 \rangle}
\newcommand{\q}[1]{\lb#1\rb}

\DeclareMathOperator{\rc}{\mathsf{RC}}
\DeclareMathOperator{\lo}{\mathsf{LO}}
\DeclareMathOperator{\lc}{\mathsf{LC}}

\begin{document}

\begin{abstract}
  Data words with binders formalize concurrently allocated memory. Most name-binding mechanisms
  in formal languages, such as the $\lambda$-calculus, adhere to properly nested
  scoping. 
  In contrast, stateful programming languages with explicit memory allocation and
  deallocation, such as~C, commonly interleave the scopes of allocated memory regions. This
  phenomenon is captured in dedicated formalisms such as dynamic sequences and bracket
  algebra, which similarly feature explicit allocation and deallocation of letters. One of
  the classical formalisms for data languages are register automata, which have been shown
  to be equivalent to automata models over nominal sets. In the present work, we introduce a
  nominal automaton model for languages of data words with explicit allocation and
  deallocation that strongly resemble dynamic sequences, extending existing nominal automata
  models by adding deallocating transitions. Using a finite NFA-type representation of the
  model, we establish a Kleene theorem that shows equivalence with a natural expression
  language. Moreover, we show that our non-deterministic model allows for determinization, a
  quite unusual phenomenon in the realm of nominal and register automata.
\end{abstract}


\maketitle

\section{Introduction}\label{sec:intro}

Formal languages over infinite alphabets are known as \emph{data
  languages}, and their elements as \emph{data words}, in a view where
the infinite alphabet is seen as abstracting data from infinite
(or very large)
domains such as UUIDs, idealized unbounded heap addresses,
or variable names in
programming languages. Numerous formalisms for data languages have
been developed, including both logics
(e.g.~\cite{DemriLazic09,GrumbergEA12,ColcombetManuel14,KlinLelyk19,HausmannEA21})
and automata models
(e.g.~\cite{KaminskiFrancez94,GrumbergEA12,ManuelEA16,BojanczykEA11}). For
our present purposes, the most relevant family of formalisms are
automata following the register paradigm~\cite{KaminskiFrancez94},
which are equipped with a fixed finite number of registers in which
data values encountered in the input can be stored. Such automata are
equivalent to various flavours of \emph{nominal
  automata}~\cite{BojanczykEA14}, which employ the formalism of
\emph{nominal sets}~\cite{Pitts13}. In the nominal framework, data
values are often identified with a fixed set of \emph{names}, and
elements of nominal sets may be thought of as containing finitely many
names.

Nominal automata and transition systems often involve explicit
\emph{name allocation} (or \emph{name binding})
(e.g.~\cite{SchroderEA17,ParrowEA21}), which may be thought of as
abstracting the allocation of a resource such as a memory location or
a channel name. Explicit deallocation, featured in many concrete
programming languages as illustrated by the following code fragment in C
\begin{verbatim}
    int *num1 = malloc(sizeof(int));
    ...
    free(num1);
\end{verbatim}
is rarely featured in abstract
formalisms. 
A notable exception are \emph{dynamic sequences}~\cite{GabbayEA15}, in
which allocated names can be deallocated in any order (that is,
without adhering to any nesting discipline); see the related work
paragraph for additional references.

We introduce a slight variant of dynamic sequences, along with a dedicated automata model
for languages of them, \emph{non-deterministic deallocation automata (NDA)}. The latter form
an extension of the name-allocating automata model of regular non-deterministic nominal
automata (RNNA)~\cite{SchroderEA17}, from which they are set apart by additionally featuring
\emph{deallocating transitions}. We equip NDA with an \emph{alphatic semantics} given in
terms of $\alpha$-equivalence classes of accepted words, and with a
\emph{local freshness semantics} given in terms of a data language obtained from the alphatic language by
erasing explicit binders. Under the latter semantics, NDA are expressively equivalent to
RNNA. Our main results on NDA are, on the one hand, a Kleene theorem stating equivalence,
w.r.t.~alphatic semantics, with what we term \emph{regular deallocation expressions}, and,
on the other hand, a determinization construction w.r.t.~local freshness semantics. The latter is a
quite unusual phenomenon in the world of regular and nominal automata, where
non-deterministic models typically have strictly higher expressivity than their
deterministic restrictions~\cite{BojanczykEA14,SchroderEA17}. In a nutshell, our
determinization result says that all that needs to be added to RNNA to allow their
determinization is explicit deallocation.

Apart from the definition of NDA, the key contributions of this work
are~(1) a construction closing the language of an NDA under
$\alpha$-equivalence (name-dropping modification, \cref{sec:ndm}), (2)
a complete algebraic specification of the accepted languages of NDA
via a form of regular expressions (Kleene-theorem, \cref{sec:kleene}),
and~(3) the determinization construction preserving local freshness semantics
(symbolic determinization, \cref{sec:determinization}). These results
are accompanied by algorithmic tractability via equiexpressiveness with RNNA,
entailing for instance decidability of language inclusion.

\paragraph*{Related work}

As mentioned above, we build on the model of regular non-deterministic
nominal automata (RNNA)~\cite{SchroderEA17}, which, on the one hand, are
related to automata models with name allocation featuring a partial
Kleene results for nominal Kleene
algebra~\cite{GabbayCiancia11,KozenEA15,KozenEA15b} and, on the other
hand, themselves satisfy a (full) Kleene theorem w.r.t.~\emph{regular
  bar expressions}~\cite{SchroderEA17}. In sharp contrast to the
computational hardness or undecidability that is otherwise typical for
non-deterministic register models
(e.g.~\cite{KaminskiFrancez94,DemriLazic09}), RNNA (which can be
identified with a subclass of register automata characterized by a
lossiness property) allow for inclusion checking in parametrized
polynomial space, with the number of registers as the parameter. Our
determinization result may eventually provide a perspective to
establishing similar complexity bounds for nominal deallocation
automata.

Dynamic sequences have been introduced by Gabbay et al.~\cite{GabbayEA15} and are equipped
with a notion of $\alpha$-equivalence and a relational semantics. A related way of denoting
data languages using name creation and destruction is presented by Brunet and
Silva~\cite{brunet_et_al:LIPIcs:2019:10683} who prove a Kleene Theorem relating the data
semantics of their \emph{memory-finite rational expressions} with non-deterministic
orbit-finite automata (NOFA)~\cite{BojanczykEA14}, equivalently register automata with
non-deterministic reassignment~\cite{KaminskiZeitlin10}. This implies that memory-finite
rational expressions inherit the computational hardness of the full register automaton
model; under local freshness semantics, NDA are equiexpressive with RNNA, i.e.~like the latter trade
expressiveness for computational tractability (as possibly indicated by our determinization result).

\section{Preliminaries}\label{sec:prel}
\begin{inappendix}
\end{inappendix}

We assume basic familiarity with automata theory and nominal sets;
we briefly review requisite concepts and notation.

\paragraph*{Nominal sets} Fix a countably infinite set $\names$ of \emph{names}. We
denote the group of finite permutations over $\names$ by $\permG$.  We write
$(a\,b)\in \permG$ for the permutation swapping~$a$ and $b$.  An \emph{action} on a set $X$ is
an operation $\permG \times X \trans{\cdot} X$.  A subset $S \subseteq \names$
\emph{supports} an element $x\in X$ if every permutation that fixes each element of $S$ also
fixes $x$:
$(\forall \pi \in \permG.\,(\forall a\in S.\, \pi(a) = a) \Rightarrow \pi \cdot x = x)$. If
every $x \in X$ has a finite support, then~$X$ is a \emph{nominal set}.  In this case there
is a least finite support $\supp(x)$ (referred to as \emph{the support} in the following).
A name $c$ is \emph{fresh} for~$x$, denoted
$c\fresh x$, if $c \notin \supp(x)$.  The set $\names$ with the action defined by
$\pi \cdot a = \pi(a)$ is a nominal set. Products $X\times Y$ of nominal sets are nominal
sets under the coordinate-wise action. The finitely supported subsets of a nominal set~$X$
form a nominal set with the expected group action:
$\pi \cdot Y = \setw{\pi \cdot y}{y \in Y}$ for every finitely supported $Y \subseteq X$. A
subset~$A$ of a nominal set~$X$ is \emph{uniformly finitely supported} if its support is the
union of the supports of its elements, i.e.~$\supp(A) = \bigcup_{x \in A} \supp(x)$.  We write $\pfin(X)$ for the \emph{finite powerset} of~$X$, i.e.~the set of finite subsets of~$X$; every finite subset of~$X$ is uniformly finitely supported. A map
$f\colon X \to Y$ between nominal sets is \emph{equivariant} if
\mbox{$f(\pi \cdot x) = \pi \cdot f(x)$} for all $x \in X$ and $\pi \in \permG$, which implies
$\supp f(x) \subseteq \supp (x)$ for all $x \in X$. In combination with uniform finite supportedness, we obtain the following simple property:
\begin{lemma}\label{lem:supp-pfin}
  Let $f\colon X\to\pfin(Y)$ be equivariant. Then
  $\supp(y)\subseteq\supp(x)$ for all $x\in X$, $y\in f(x)$.
\end{lemma}
\noindent 
A subset $S$ of a nominal set~$X$ is equivariant if
$\pi \cdot S \subseteq S$. We write $\cleq S$ for the closure of $S$
under equivariance: $\cleq S = \bigcup_{\pi \in \permG} \pi \cdot
S$. For $x \in X$, the set $\orb_x = \cleq \set{x} \subseteq X$ is the
\emph{orbit} of $x$. For two elements $x_1,x_2 \in X$, the two orbits
$\orb_{x_1}$ and~$\orb_{x_2}$ are either equal or disjoint.  A nominal
set $X$ is \emph{orbit-finite} if it has finitely many orbits.
A map $f\colon X \to Y$ between nominal sets is \emph{supported} by
$S \subseteq \names$ if $f(\pi \cdot x) = \pi \cdot f(x)$ holds for
all $x\in X$ and all $\pi \in \permG$ satisfying $\pi(s) = s$ for all
$s \in S$. The map $f$ is \emph{finitely supported} if it is supported
by a finite set $S$. (Note that equivariant maps are thus precisely
those maps which are supported by the empty set.)  On the nominal set
$\names\times X$ we have the equivalence relation defined by
$(a,x) \sim (b,y)$ iff $(a\,c)\cdot x = (b\,c) \cdot y$ for some, or
equivalently all, names $c$ such that $c\fresh a,b,x,y$; we
occasionally refer to this relation, somewhat informally, as
\emph{$\alpha$-equivalence}. The \emph{abstraction set}
$[\names](\cdot)$ is defined as the quotient set
$\names \times X /\mathord{\sim}$.  The $\sim$-equivalence class of
$(a,x) \in \names \times X$ is denoted by $\braket{a}x$. If
$b\in\supp(x)$, then $\braket{a}x$ does not equal $\braket{b}x'$ in
$[\names]X$ for any~$x'\in X$; in this case, we say that the renaming
of~$a$ into $b$ in $\braket{a}x$ is \emph{blocked}.
For further details on nominal sets, see Pitts' book~\cite{Pitts13}.

\paragraph*{Bar languages, RNNA and bar NFA} A \emph{data word} is a finite word over $\names$; as
usual, we write $\names^*$ for the set of all data words. A \emph{data language} is a subset
of $\names^*$. We use a bar symbol (\enquote{$\midmid$}) preceding names to indicate that
the next letter is bound until the end of the word. A \emph{bar string} is a finite word
over the alphabet
$\barNames = \names \cup \setw{\newletter{a}}{a\in \names}$~\cite{SchroderEA17}.  We have a
notion of $\alpha$-equivalence generated by $w\newletter{a}v \aeq w\newletter{b}u$ if
$\braket{a} v = \braket{b} u$ in $[\names]\bar\names^*$. We denote the $\alpha$-equivalence
class of a bar string $w$ by $[w]_\alpha$.  Hence, a bar string $w$ represents a set of
data words obtained by removing all bars from all words in $[w]_\alpha$.  Bar languages can
be recognized by either \emph{bar NFA} or \emph{RNNA}: the former are just ordinary
non-deterministic finite automata (NFA) with input alphabet (a necessarily finite subset of)
$\bar\names$, and the latter are infinite but finitely representable automata
$A = (Q,\Delta,i,F)$ with an orbit-finite nominal set of states~$Q$, an equivariant
transition relation $\Delta$ invariant under $\alpha$-renaming, an initial state $i$ and an
equivariant subset $F \subseteq Q$ of final states.

\section{Explicit Deallocation}\label{sec:alpha}


We next introduce our notion of words with explicit deallocation
and
an associated notion of $\alpha$-equivalence. As mentioned above,
these notions are related to those found in dynamic
sequences~\cite{GabbayEA15}.
Elements of~$\names$ can be equipped with a
\emph{binder}~$\lb\cdot$ and a \emph{marker}~$\cdot\rb$,
yielding $\hat{S} = \set{\lb a, a, a\rb, \q{a}\mid a\in S}$
for a set $S\subseteq \names$ of names.
We form a new alphabet
$\ah$, 
%
which is a nominal set under the action of~$\permG$ extending the action
on~$\names$ in the expected manner: $\pi\cdot\lb a=\lb\pi(a)$, $\pi\cdot\lb
a\rb=\lb\pi(a)\rb$, $\pi\cdot a\rb=\pi(a)\rb$. The set~$\ahs$ of words
over~$\ah$ is then a nominal set under the action
of~$\permG$ acting on the letters of a word as just described. We view names as an abstract
representation of resources. In this view, we understand a letter $\lb
a$ as \emph{opening} or \emph{allocating} the name~$a$; a letter
$a$ as using an allocated or preexisting name~$a$; a
letter~$a\rb$ as \emph{closing} or \emph{deallocating} the name~$a$; and a letter $\lb
a\rb$ as allocating and immediately deallocating the name~$a$.

Based on this view, one could give definitions of multisets of names
allocated or deallocated in a given word. It turns out to be
technically more convenient to just distinguish names according to the
type of their first occurrence when traversing the word from the left
or from the right, respectively. In particular, we will use these
concepts in defining a notion of (non-)shadowing to the right.
Specifically, we introduce the notions of \emph{right-closed},
\emph{left-open}, and \emph{left-closed} names of a word. A
name~$a$ is \emph{right-closed} in a
word~$w$ if the last appearance of~$a$
in~$w$ (read left to right) has the marker attached (i.e.~is
either~$a\rb$
or~$\q{a}$). Similarly, a name is \emph{left-closed} in a word
$w$ if the first appearance of~$a$
in~$w$ (read left to right) has the binder attached (i.e.~is
either~$\lb a$ or~$\q{a}$). Moreover, a name
$a$ is \emph{left-open}
in~$w$ if its first appearance does not have the binder attached
(i.e.~is either~$a$
or~$a\rb$). Recursive definitions of the sets $\rc(w)$,
$\lo(w)$, and
$\lc(w)$ of all right-closed, left-open, and left-closed names of the
word
$w$, respectively, are given in \cref{tab:namedefs}. As indicated
above, we use these notions to formalize a notion of \emph{right
  shadowing}: If a name~$a$ is right-closed
in~$v$ and left-open in~$u$, then the composite word
$vu$ is \emph{right-shadowing}, i.e.~some occurrence
of~$a$ with the marker shadows an occurrence
of~$a$ without the binder further to the right. Correspondingly. a
word~$w$ is \emph{right-non-shadowing} if $\rc(u) \cap \lo(v) =
\emptyset$ holds for each decomposition $w = uv$. We write
$\RNS(\ah)$ for the set of right non-shadowing words over the alphabet
with binders $\ah$.

All these notions are illustrated in \cref{ex:boundAndClean}.
\begin{table}%
  \begin{center}
    \resizebox{\columnwidth}{!}{%
      \(
      \def\arraystretch{1.2}
      \begin{array}{@{}l|@{\hspace*{12pt}}l|@{\,}l|@{\,}l@{\,}}
        \bottomrule
        w'
        & \rc(w')
        & \lo(w')
        & \lc(w')
        \\
        \midrule
        \epsilon
        
        & \emptyset
        & \emptyset
        & \emptyset
        \\\hline
        a\,w
        & \rc(w)
        & \lo(w) \cup \set{a}
        &\lc(w) \setminus \set{a}
        \\\hline
        \lb a w
        & \rc(w)
        & \lo(w)\setminus\set{a}
        & \lc(w) \cup \set{a}
        \\\hline
        a\rb w
        & \begin{array}{ll}
            \rc(w) \\
            \hspace*{12pt}\text{if } a\in \lc(w)\cup\lo(w)\\
            \rc(w)\cup\set{a}\\
            \hspace*{12pt}\text{otherwise}
          \end{array}
        & \lo(w) \cup \set{a}
        &\lc(w) \setminus \set{a}
        \\\hline
        \q{a}w
        & \begin{array}{ll}
            \rc(w) \\
            \hspace*{12pt}\text{if } a\in \lc(w)\cup \lo(w)\\
            \rc(w)\cup\set{a}\\
            \hspace*{12pt}\text{otherwise}
          \end{array}
        &\lo(w)\setminus\set{a}
        &\lc(w) \cup \set{a}
        \\
        \bottomrule
      \end{array}
      \)%
  }%
\end{center}

\caption{Definitions of the subsets $\rc(w'), \lo(w'), \lc(w')\subseteq \A$ by case distinction
  on $w'$.}
\label{tab:namedefs}
\end{table}

\begin{lemma}[label=lem:cnConcat]
  For compound words $w = uv$, we have that
  \begin{enumerate}
    \item $\rc(w) = (\rc(u)\setminus (\lc(v)\cup \lo(v))) \cup \rc(v)$,
    \item $\lc(w) = \lc(u) \cup (\lc(v)\setminus \lo(u))$, and
    \item $\lo(w) = \lo(u) \cup (\lo(v) \setminus \lc(u))$.
  \end{enumerate}
\end{lemma}

\begin{proofappendix}{lem:cnConcat}
  We prove the three statements simultaneously by induction on the length of
  a word $w$. The proof involves straightforward but extensive technicalities.
  In particular, we show that
  for every decomposition $w = uv$ we have
  \begin{align*}
    \rc(uv) &= (\rc(u)\setminus (\lc(v)\cup \lo(v))) \cup \rc(v),\\
    \lc(uv) &= \lc(u) \cup (\lc(v)\setminus \lo(u)),\\
    \lo(uv) &= \lo(u) \cup (\lo(v) \setminus \lc(u)).
  \end{align*}

  For the base case $w = \epsilon$, we have the trivial decomposition
  $\epsilon = \epsilon\epsilon$.
  and by definition we have
  \[
    \rc(\epsilon) = \lc(\epsilon) = \lo(\epsilon) = \emptyset.
  \]
  thus the claim holds.

  For the inductive step, let $w = \gamma w'$ be a non-empty word
  with $\gamma\in\ah$ the first symbol and suffix $w'$.
  We assume the claim holds for every proper suffix
  of~$w$, in particular for $w'$.

  Let $w = uv$ be a decomposition of $w$.
  If $u = \epsilon$, then with
  $\rc(\epsilon) = \lc(\epsilon) = \lo(\epsilon) = \emptyset$, we get

 
  \begin{align*}
    \rc(w) &= (\emptyset \setminus \emptyset) \cup \rc(w)
           = (\rc(u) \setminus (\lc(v)\cup\lo(v))) \cup \rc(v)\text{,}\\
    \lc(w) &= \emptyset \cup (\lc(w)\setminus \emptyset)
           = \lc(u) \cup (\lc(v)\setminus \lo(u))\text{,}\\
    \lo(w) &= \emptyset \cup (\lo(w)\setminus \emptyset)
           = \lo(u) \cup (\lo(v) \setminus \lc(u))\text{.}
  \end{align*}
  If instead $v = \epsilon$, then we have
  \begin{align*}
    \rc(w) &= (\rc(w) \setminus \emptyset) \cup \emptyset
           = (\rc(u) \setminus (\lc(v)\cup\lo(v))) \cup \rc(v)\text{,}\\
    \lc(w) &= \lc(w) \cup (\emptyset \setminus \lo(u))
           = \lc(u) \cup (\lc(v)\setminus \lo(u))\text{,}\\
    \lo(w) &= \lo(w) \cup (\emptyset\setminus \lc(u))
           = \lo(u) \cup (\lo(v) \setminus \lc(u))\text{.}
  \end{align*}
  Thus, assume both $u$ and $v$ are non-empty.
  Since $w = \gamma w'$, the split between $u$ and $v$
  occurs in $w'$, hence we have $w' = u'v$
  such that $u = \gamma u'$.
  By the induction hypothesis we have:
  \begin{align*}
    \rc(w')
      &= (\rc(u')\setminus (\lc(v)\cup\lo(v))) \cup \rc(v),\\
    \lc(w')
      &= \lc(u') \cup (\lc(v)\setminus \lo(u')),\\
    \lo(w')
      &= \lo(u') \cup (\lo(v)\setminus \lc(u')).
  \end{align*}

  We proceed via case distinction on $\gamma$.
  \begin{itemize}
  \item $\gamma = a$:
  From \cref{tab:namedefs}, we have
  \begin{align*}
    \rc(w) &= \rc(aw') = \rc(w'),\\
    \rc(u) &= \rc(au') = \rc(u'),\\
    \lc(w) &= \lc(aw') = \lc(w')\setminus\{a\},\\
    \lc(u) &= \lc(au') = \lc(u')\setminus\{a\},\\
    \lo(w) &= \lo(aw') = \lo(w')\cup\{a\},\\
    \lo(u) &= \lo(au') = \lo(u')\cup\{a\}.
  \end{align*}
  For $\rc$, we compute
  \begin{align*}
    \rc(w)
      &= \rc(w')\\
      &= (\rc(u')\setminus (\lc(v)\cup\lo(v))) \cup \rc(v)\\
      &= (\rc(u)\setminus (\lc(v)\cup\lo(v))) \cup \rc(v),
  \end{align*}
  for $\lc$, we obtain
  \begin{align*}
    \lc(w)
      &= \lc(w')\setminus\{a\}\\
      &= (\lc(u') \cup (\lc(v)\setminus \lo(u')))\setminus\{a\}\\
      &= (\lc(u')\setminus\{a\}) \cup ((\lc(v)\setminus \lo(u'))\setminus\{a\})\\
      &= \lc(u) \cup (\lc(v)\setminus (\lo(u')\cup\{a\}))\\
      &= \lc(u) \cup (\lc(v)\setminus \lo(u)),
  \end{align*}
  and for $\lo$, we find
  \begin{align*}
    \lo(w)
      &= \lo(w')\cup\{a\}\\
      &= (\lo(u') \cup (\lo(v)\setminus \lc(u')))\cup\{a\}\\
      &= \lo(u')\cup\{a\} \cup (\lo(v)\setminus (\lc(u')\setminus \{a\}))\\
      &= \lo(u) \cup (\lo(v)\setminus \lc(u)).
  \end{align*}
  \item $\gamma = \lb a$:
  From \cref{tab:namedefs}, we obtain
  \begin{align*}
    \rc(w) &= \rc(\lb a w') = \rc(w'),\\
    \rc(u) &= \rc(\lb a u') = \rc(u'),\\
    \lc(w) &= \lc(\lb a w') = \lc(w')\cup\{a\},\\
    \lc(u) &= \lc(\lb a u') = \lc(u')\cup\{a\},\\
    \lo(w) &= \lo(\lb a w') = \lo(w')\setminus\{a\},\\
    \lo(u) &= \lo(\lb a u') = \lo(u')\setminus\{a\}.
  \end{align*}
  For $\rc$, we calculate
  \begin{align*}
    \rc(w)
      &= \rc(w')\\
      &= (\rc(u')\setminus (\lc(v)\cup\lo(v))) \cup \rc(v)\\
      &= (\rc(u)\setminus (\lc(v)\cup\lo(v))) \cup \rc(v).
  \end{align*}
  For $\lc$, we have
  \begin{align*}
    \lc(w)
      &= \lc(w')\cup\{a\}\\
      &= \bigl(\lc(u') \cup (\lc(v)\setminus \lo(u'))\bigr)\cup\{a\}\\
      &= (\lc(u')\cup\{a\}) \cup (\lc(v)\setminus \lo(u'))\\
      &= \lc(u) \cup (\lc(v)\setminus \lo(u'))\\
      &= \lc(u) \cup (\lc(v)\setminus \lo(u)),
  \end{align*}
  where in the last step we
  observe that the only possible difference between $\lo(u')$ and
  $\lo(u)$ is the element $a$, which is in any case contained in
  $\lc(u)$.
  For $\lo$, we compute
  \begin{align*}
    \lo(w)
      &= \lo(w')\setminus\{a\}\\
      &= (\lo(u') \cup (\lo(v)\setminus \lc(u')))\setminus\{a\}\\
      &= (\lo(u')\setminus\{a\}) \cup\bigl((\lo(v)\setminus \lc(u'))\setminus\{a\}\bigr)\\
      &= \lo(u) \cup (\lo(v)\setminus (\lc(u')\cup\{a\}))\\
      &= \lo(u) \cup (\lo(v)\setminus \lc(u)).
  \end{align*}

  \item $\gamma = a\rb$:
  From \cref{tab:namedefs}, we have
  \begin{align*}
    \lc(w) &= \lc(a\rb w') = \lc(w')\setminus\{a\},\\
    \lc(u) &= \lc(a\rb u') = \lc(u')\setminus\{a\},\\
    \lo(w) &= \lo(a\rb w') = \lo(w')\cup\{a\},\\
    \lo(u) &= \lo(a\rb u') = \lo(u')\cup\{a\}.
  \end{align*}
  Also, we have
  \begin{align*}
    a\in\rc(w) &\iff a\in\rc(w') \text{ or } a\notin(\lc(w')\cup\lo(w')),\\
    a\in\rc(u) &\iff a\in\rc(u') \text{ or } a\notin(\lc(u')\cup\lo(u')).
  \end{align*}
  We distinguish cases:
  \begin{enumerate}
  \item If $a \in (\lc(w')\cup\lo(w'))$, then
    $\rc(w) = \rc(w') = (\rc(u')\setminus (\lc(v)\cup\lo(v))) \cup \rc(v)$.
    Recall that $\lc(w') = \lc(u')\cup(\lc(v)\setminus\lo(u'))$
    and $\lo(w') = \lo(u')\cup(\lo(v)\setminus\lc(u'))$ by
    induction, so $\lc(w')\cup\lo(w') = \lc(u')\cup\lc(v)\cup\lo(u')\cup\lo(v)$, and 
    we further distinguish cases on~$a$:
    \begin{enumerate}
    \item If $a \in \lc(u')$ or $a \in \lo(u')$, then $\rc(u) = \rc(u')$ and hence
      $(\rc(u')\setminus(\lc(v)\cup\lo(v))) \cup \rc(v) = (\rc(u) \setminus
      (\lc(v)\cup\lo(v))) \cup \rc(v)$.
    \item If $a \in \lc(v)$ or $a \in \lo(v)$, then
      $\rc(u')\setminus(\lc(v)\cup\lo(v)) = (\rc(u')\cup\{a\})\setminus(\lc(v)\cup\lo(v))$ and
      thus
      $(\rc(u')\setminus(\lc(v)\cup\lo(v))) \cup \rc(v) = (\rc(u)\setminus(\lc(v)\cup\lo(v)))
      \cup \rc(v)$.
    \end{enumerate}
  \item If $a \notin (\lc(w')\cup\lo(w'))$, then we immediately
    have
    $\rc(w) = \rc(w') \cup \set{a} = (\rc(u')\setminus(\lc(v)\cup\lo(v))) \cup
    \rc(v) \cup \set{a} = (\rc(u)\setminus(\lc(v)\cup\lo(v))) \cup \rc(v)$.
  \end{enumerate}
  For $\lc$, we obtain
  \begin{align*}
    \lc(w)
      &= \lc(w')\setminus\{a\}\\
      &= \bigl(\lc(u') \cup (\lc(v)\setminus \lo(u'))\bigr)\setminus\{a\}\\
      &= (\lc(u')\setminus\{a\}) \cup ((\lc(v)\setminus \lo(u'))\setminus\{a\})\\
      &= \lc(u) \cup \bigl(\lc(v)\setminus (\lo(u')\cup\{a\})\bigr)\\
      &= \lc(u) \cup (\lc(v)\setminus \lo(u)),
  \end{align*}
  and for $\lo$, we have
  \begin{align*}
    \lo(w)
      &= \lo(w')\cup\{a\}\\
      &= \lo(u') \cup (\lo(v)\setminus \lc(u')) \cup\{a\}\\
      &= \lo(u')\cup\{a\} \cup (\lo(v)\setminus (\lc(u')\setminus \{a\}))\\
      &= \lo(u) \cup (\lo(v)\setminus \lc(u)),
  \end{align*}
 
  \item $\gamma = \q{a}$: 
  Finally, for $\gamma = \q{a}$ we combine the above patterns.
  From \cref{tab:namedefs} for $\lc$ and $\lo$, we have
  \begin{align*}
    \lc(w) &= \lc(\q{a}w') = \lc(w')\cup\{a\},\\
    \lc(u) &= \lc(\q{a}u') = \lc(u')\cup\{a\},\\
    \lo(w) &= \lo(\q{a}w') = \lo(w')\setminus\{a\},\\
    \lo(u) &= \lo(\q{a}u') = \lo(u')\setminus\{a\},
  \end{align*}
  thus the proof is identical to the case for $\gamma = \lb a$ above.
  For $\rc$, $\q{a}$ behaves 
  the same as $a\rb$, and we can repeat the argument
  given in that case. \qedhere
  \end{itemize}
\end{proofappendix}
\noindent As an immediate consequence, we have
\begin{lemma}\label{cor:rcRNS}
  For a decomposition $w = uv$ of a right non-shadowing word $w$, we have
  $\rc(w) = (\rc(u)\setminus\lc(v)) \cup \rc(v)$.
\end{lemma}
\noindent 
We next use these concatenation identities to establish a closure
property of right non-shadowing words under concatenation. From the
definition of right non-shadowing words, we see that concatenating two
words $w$ and $u$ can only introduce shadowing to the right if a name
that is right-closed in $w$ clashes with a name that is left-open in
$u$; the condition $\rc(w)\cap \lo(u)=\emptyset$ in the next lemma
rules this~out:

\begin{lemma}\label{lem:rnsConcat}
  Let $w,u \in \RNS(\ah)$ be right non-shadowing words such that
  $\rc(w)\cap \lo(u) = \emptyset$. Then the concatenation $wu$ is also
  right non-shadowing.
\end{lemma}

\begin{proofappendix}{lem:rnsConcat}
  Let $wu = xy$ be a decomposition. We have to show that
  $\rc(x) \cap \lo(y) = \emptyset$.  If $x = w$ and $y = u$, then we
  are done by hypothesis. Otherwise, the decomposition splits
  either $w$ or $u$. We distinguish cases:

  First, assume that the split lies inside $w$. Then there is a word~$w'$ such that
  $w = xw'$ and hence $y = w'u$. Since $w$ is right-non-shadowing, the decomposition
  $w = xw'$ yields $\rc(x) \cap \lo(w') = \emptyset$. Moreover, by \cref{lem:cnConcat} we
  obtain
  \[
    \lo(y) = \lo(w'u) = \lo(w') \cup (\lo(u) \setminus \lc(w')),
  \]
  and it remains to exclude elements
  in $\rc(x) \cap (\lo(u) \setminus \lc(w'))$. Suppose
  for contradiction that $a\in\rc(x) \cap (\lo(u) \setminus \lc(w'))$.
  Then $a \in \rc(x)$, $a \in \lo(u)$ and $a \notin \lc(w')$,
  and again by \cref{lem:cnConcat}, we have
  \[
    \rc(w) = (\rc(x) \setminus \lc(w')) \cup \rc(w').
  \]
  Thus, we have $a \in \rc(w)$, and since $a \in \lo(u)$ holds true,
  we obtain
  $a \in \rc(w) \cap \lo(u)$, contradicting our assumption
  $\rc(w) \cap \lo(u) = \emptyset$. Consequently,
  $\rc(x) \cap \lo(y) = \emptyset$ in this case.

  Now assume that the split lies inside $u$. Then there exists
  a word~$u'$ such that $u = u'y$ and hence $x = wu'$.
  Since $u$ is right-non-shadowing, the decomposition $u = u'y$
  yields $\rc(u') \cap \lo(y) = \emptyset$. By
  \cref{lem:cnConcat}, we have
  \begin{align*}
    \rc(x) &= \rc(wu')
           = (\rc(w) \setminus \lc(u')) \cup \rc(u'),\\
    \lo(u) &= \lo(u'y)
           = \lo(u') \cup (\lo(y) \setminus \lc(u')).
  \end{align*}
  Let $a \in \rc(x) \cap \lo(y) = \rc(wu') \cap \lo(y)$. If we have $a \in \rc(u')$, then
  $a \in \rc(u') \cap \lo(y)$, contradicting that $u$ is right-non-shadowing. Hence,
  $a \notin \rc(u')$, and therefore $a \in \rc(w) \setminus \lc(u')$ holds. In particular,
  we have $a \notin \lc(u')$, so $a \in \lo(y) \setminus \lc(u')$ and thus, by the above
  expression for $\lo(u)$, we obtain $a \in \lo(u)$.  We conclude that
  $a \in \rc(w) \cap \lo(u)$, contradicting our assumption. Hence, also in this case
  we have $\rc(x) \cap \lo(y) = \emptyset$, whence $wu$ is right non-shadowing.
\end{proofappendix}

\begin{example}[label=ex:boundAndClean]
  We present some examples of words, giving the relevant sets of names
  as defined above and noting whether they are right non-shadowing
  (rns):

  \smallskip
  \centerline{
  \begin{tabular}{c||c|c|c||c}
    $w$ & $\lo(w)$ & $\rc(w)$ & $\lc(w)$ & rns\\
    \hline
    $\lb ab\rb a\rb$ & $\set{b}$ & $\set{a}$ & $\set{a}$ & yes\\
    $\q{a}\lb b$ & $\emptyset$ & $\set{a}$ & $\set{a,b}$ & yes\\
    $a\q{a}a$ & $\set{a}$ & $\emptyset$ & $\emptyset$ & no\\
    $bab\rb a\rb$ & $\set{a,b}$ & $\set{a,b}$ & $\emptyset$ & yes\\
  \end{tabular}}
\end{example}

\noindent As usual, renaming of bound names in a word leads to a notion of
$\alpha$-equivalence:

\begin{defn}[name=$\alpha$-equivalence,label=def:Alpha]
  The relation $\aeq$ on $\RNS(\ah)$ is
  defined inductively by the following rules 1 through 5 and
  closure under transitivity:
  \noindent
  \begin{tasks}[style=enumerate](4)
    \task\begin{minipage}[H]{0pt}
      \begin{prooftree}\label{rule:epsilon}
        \AxiomC{} \UnaryInfC{$\epsilon \aeq \epsilon$}
      \end{prooftree}
    \end{minipage}
    \task
    \begin{minipage}[H]{0pt}
      \begin{prooftree}\label{rule:alpha}
        \AxiomC{$\braket{a}\,w = \braket{b}\,w'$}
        \UnaryInfC{$\lb a\,w \aeq \lb b\,w'$}
      \end{prooftree}
    \end{minipage}
  \end{tasks}
  
  \begin{tasks}[style=enumerate,start=3](4)
  \task
    \begin{minipage}[H]{0pt}
      \begin{prooftree}\label{rule:cong-free-open}
        \AxiomC{$w\aeq w'$}
        \AxiomC{$\gamma \in \set{a, \lb a} $}
        \BinaryInfC{$\gamma w \aeq
          \gamma w'$}
      \end{prooftree}
    \end{minipage}
  \end{tasks}
  \begin{tasks}[style=enumerate,start=4](4)
    \task 
    \begin{minipage}[H]{0pt}
      \begin{prooftree}\label{rule:cong-close}
        \AxiomC{$w\aeq w', a \notin \lo(w)$}
        \UnaryInfC{$a\rb w \aeq a\rb w'$}
      \end{prooftree}
    \end{minipage}
  \end{tasks}
  \begin{tasks}[style=enumerate,start=5](4)
    \task 
    \begin{minipage}[H]{0pt}
    \begin{prooftree}\label{rule:cong-unknown}
    \AxiomC{$w\aeq w', \set{a,b} \cap (\lo(w)) = \emptyset$}
    \UnaryInfC{$\q{a} w \aeq \q{b} w'$}
  \end{prooftree}
    \end{minipage}
  \end{tasks}\medskip

  \noindent (Rule~4.~is phrased separately from Rule~3.~only to ensure
  that the word $a\rb w$ is right-non-shadowing). A \emph{literal language} is a
  language $L_0 \in \pow(\RNS(\ah))$.  We write $[w]_\alpha$ for the
  $\alpha$-equivalence class of~$w$. An \emph{alphatic language} is a
  set $L_\alpha \subseteq \RNS(\ah)/\mathord{\aeq}$ of
  $\alpha$-equivalence classes of right-non-shadowing words. A \emph{data language} is
  a set $L_D \in \pow(\names^*)$ of data words.
\end{defn}
\noindent
We write $\clal(L_0)$ for the \emph{closure under $\alpha$-equivalence} of a literal language~$L_0$,
i.e.~the set $\set{w\in \RNS(\ah) \mid w \aeq w' \text{ for some } w' \in L_0}$ of words $\alpha$-equivalent to some word in~$L_0$.

\begin{fact}\label{fact:supp-lo}
  For $w\in \RNS(\ah)$, the support of $[w]_\alpha$ is
  $\supp([w]_\alpha) = \lo(w) $; in particular, if $w \aeq w'$, then
  $\lo(w) = \lo(w')$.
\end{fact}

\mycomment{ 
\begin{lemma}\label{alphaRightShadow}
  Renaming to an $\alpha$-equivalent word does not introduce shadowing to the right: for all
  $w, w' \in \RNS(\ah)$ such that $w \aeq w'$, $w$ is right non-shadowing iff so is $w'$.
\end{lemma}

\begin{proofappendix}{alphaRightShadow}
For renaming $\lb a w$ to $\lb b w'$, we have $\braket{a}w = \braket{b}w'$ by definition.
Hence, we have $b \notin \lo(w)$ and $a \notin \lo(w')$, and renaming
$a$ to $b$ does not introduce shadowing to the right.
\end{proofappendix}
}

\begin{lemma}\label{alphaRuleAdmissible}
\sloppy
The following rule rule is admissible for $\alpha$-equivalence:
  \begin{prooftree}
    \AxiomC{$\braket{a}[w]_\alpha = \braket{b}[w']_\alpha$}
    \UnaryInfC{$\lb aw \aeq \lb bw'$}
  \end{prooftree}

\end{lemma}

\begin{proofappendix}{alphaRuleAdmissible}
Let $d$ be fresh for $a,b,w,w'$. Then the antecedence spells out to
$(ad)\cdot[w]_\alpha=(bd)\cdot[w']_\alpha$. By equivariance, $[(ad)\cdot w]_\alpha=[(bd)\cdot w']_\alpha$,
so $(ad)\cdot w \aeq (bd)\cdot w'$. Prefixing by $\lb d$ using rule (2) of \cref{def:Alpha} yields
$\lb d\,(ad)\cdot w \aeq \lb d\,(bd)\cdot w'$. Since $d$ is fresh for $w,w'$, we have
$\braket{d}(ad)\cdot w=\braket{a}w$ and $\braket{d}(bd)\cdot w'=\braket{b}w'$, and by rule (3)
of \cref{def:Alpha} and transitivity, we get $\lb a\,w \aeq \lb b\,w'$.\qedhere
\end{proofappendix}

\noindent From the definition of $\alpha$-equivalence for a letter
$\q{a}$ bound from both sides, it becomes apparent that the identity
of the bound letter is immaterial: In a word $\q{a}w$, the letter~$a$
can be substituted by any other name that is not left-open in~$w$. We
say that the letter is \emph{unknown} and use the following
alternative notation:
\begin{notn}
  We write $u?v$ in lieu of $u\q{a}v$ whenever~$u\q{a}v$ is right
  non-shadowing.
\end{notn}

\begin{rem}\label{rem:rns}
  Our actual interest in the following is exclusively in right
  non-shadowing words; notice, for instance, that we have
  defined $\alpha$-equivalence only on right-non-shadowing words. Technically, the
  property is crucial for maintaining the interpretability of words by
  our nominal automata model introduced in \cref{def:nda};
  essentially, it will be seen that any word having a run in one of
  our automata must be right-non-shadowing. Allowing shadowing on the left (as
  in~$\lb a\lb a$, which is right-non-shadowing) but not on the right is in tune with
  our intuitive view of $\lb a$ and $a\rb$ as representing allocation
  and deallocation of resources: Allocating the same resource twice in
  a row without an intervening deallocation will be regarded as bad
  style (see also the discussion of disciplined words in
  \Cref{sec:semantics}) but conceivably will often be technically
  possible; this corresponds to the fact that words like $\lb a\lb a$
  are right-non-shadowing. Contrastingly, deallocating the same resource twice, as
  in $a\rb a\rb$, or using a resource again after deallocation, as in
  $a\rb a$, would reasonably be regarded as technically impossible, as
  the resource should no longer be at hand after deallocation; this
  corresponds to the fact that $a\rb a\rb$ and $a\rb a$ are
  right shadowing.
\end{rem}

\section{Semantics}\label{sec:semantics}


In order to express data languages using languages with binders, we
provide an interpre\-ta\-tion of words in $\RNS(\ah)$ as words in
$\names^*$.  This is done using a \emph{debracket} function $\db$
turning a literal word into the data word it represents by removing
delimiters. We define the function
$\db\colon \RNS(\ah) \to \names^*$ as the extension to words of the
letter substitution which replaces each of the letters
$a, \lb a, a\rb, \q{a}$ with $a$. Following previous
work~\cite{PruckerSchroder24arXiv,SchroderEA17,UrbatEA21,frank2025alternatingnominalautomataallocation},
we define the \emph{local freshness semantics} $\D(L_\alpha)$ of an alphatic
language~$L_\alpha$ by taking all right-non-shadowing
representatives of $\alpha$-equivalence classes in the languages and
then debracketing; formally:
\begin{equation*}
  \D(L_\alpha) = \set{\db(w)\mid [w]_\alpha \in L_\alpha}.
\end{equation*}
%

\noindent The intuition of viewing $a\rb$ as deallocating a
pre-existing or previously allocated resource~$a$ suggests a discipline of
memory safety in which we insist that every resource is deallocated as
soon as it is no longer needed, i.e.~at the time of its last use. For
instance, the word $\lb a\lb bb\rb$ would leak the resource~$a$, which
is allocated but never used again; in a memory-safe discipline, one
would instead insist that~$a$ is immediately deallocated, i.e.~the
word should be $\lb a\rb\lb bb\rb$. Similarly, $\lb a\lb b a b$ leaks
both~$a$ and~$b$, and has $\lb a\lb b a\rb b\rb$ as a memory-safe
modification. We phrase associated formal definitions as follows.

\begin{defn}
  A word $w$ over $\ah$ is \emph{disciplined} if for every
  decomposition $w = uv$, we have
  $a \in LC(u)\cup LO(u) \Rightarrow a \in RC(u) \vee (a \in RC(v)
  \wedge a \notin LC(v))$).  A language $L$ over $\ah$ is
  \emph{disciplined} if every word in $L$ is so.  Given a word $w$
  over $\ah$, we obtain a disciplined word $\disc(w)$ recursively as
  follows: $\disc(\epsilon) = \epsilon$, $\disc(aw) = a\disc(w)$ if
  $a\in\lo(w)$ and $\disc(aw) = a\rb\disc(w)$ otherwise,
  $\disc(\q{a}w) = \q{a}\disc(w)$, $\disc(\lb aw) = \lb a\disc(w)$ if
  $a\in\lo(w)$ and $\disc(\lb aw) = \q{a}\disc(w)$ otherwise,
  $\disc(a\rb w) = a\rb \disc(w)$.
\end{defn}

\begin{lemma}\label{lem:disc}
  For every word $w\in\ah^*$, the following hold.
  \begin{enumerate}
  \item The word $\disc(w)$ is disciplined.
  \item $\lo(\disc(w))=\lo(w)$.
  \end{enumerate}
\end{lemma}
\begin{example}
  The words $w_1=\lb a\lb bb\rb$ and $w_2=\lb a\lb b a b$ indeed fail
  to be disciplined, as witnessed, for instance, by the decompositions
  $w_1=(\lb a)(\lb bb\rb)$ and $w_2=(\lb a\lb b)(ab)$. The function
  $\disc$ acts on these words as indicated above,
  i.e.~$\disc(w_1)=\lb a\rb\lb bb\rb$ and
  $\disc(w_2)=\lb a\lb b a\rb b\rb$.
\end{example}
\noindent Crucially, the transformation to disciplined words does not
affect the local freshness semantics:
\begin{lemma}[label=lem:discSameData]
  For $w\in \RNS(\ah)$, we have
  \begin{equation*}
    \D(\set{[w]_\alpha}) = \D(\set{[\disc(w)]_\alpha)}).
  \end{equation*}
\end{lemma}
\begin{proofappendix}{lem:discSameData}
  This is immediate by induction on the definition of $\disc$: The base case is trivial,
  and the only non-trivial case in the inductive step is where $w = \lb aw'$ with
  $a \notin \lo(w')$, then we have $\disc(w) = \q{a}\disc(w')$. Since $a \notin
  \lo(w')$, we have for each $b \notin \lo(w')$ that $\lb aw'
  \aeq \lb bw'$ holds, as well as $\q{a}\disc(w') \aeq \q{a}\disc(w')$ by the definition of
  $\alpha$-equivalence. Then we have $\D(\set{[\lb aw]_\alpha}) =
  \D(\set{[\q{a}\disc(w)]_\alpha})$.
\end{proofappendix}

\section{Deallocation Automata}\label{sec:nda}


We now proceed to introduce our automata model with deallocating
transitions. The model internalizes allocation, explicit deallocation,
and their immediate combination (i.e.~$\q{a}$) as first-class
transitions.

\begin{defn}[label=def:nda]
  A \emph{non-deterministic deallocation automaton (NDA)} is a tuple
  $A = (Q,\Delta,i,F)$ where $Q$ is an orbit-finite set of
  \emph{states}, $\Delta\subseteq Q \times \ah \times Q$ is an
  equivariant \emph{transition relation},~$i\in Q$ is the
  \emph{initial} state, and~$F\subseteq Q$ is an equivariant set of
  \emph{final} states. We write $q \trans{\gamma} q'$ both for a
  corresponding element of $\Delta \subseteq Q \times \ah \times Q$
  and for the statement $(q \trans{\gamma} q') \in \Delta$, which can
  be true or false; we refer to $q \trans{\gamma} q'$ as a
  \emph{transition}. In correspondence to their labelling letters, we
  refer to transitions as \emph{free} ($a$), \emph{allocating}
  ($\lb a$), \emph{deallocating} ($a \rb$), \emph{unknown} ($\q{a}$).
  We impose the following conditions on~$\Delta$:
    \begin{itemize}
    \item \emph{Left $\alpha$-invariance}: Let $q,q',q''\in Q$ and
      $a,b\in\names$ such that $\braket{a}\, q' = \braket{b}\,
      q''$. Then the following hold:
      \begin{itemize}
        \item If $q \trans{\lb a} q' $, then
        $q\trans{\lb b} q''$ .
        
      \item If $q \trans{\q{a}} q' $, then $q\trans{\q{b}} q'' $.
      \end{itemize}
    \item \emph{Name erasure}: If $q\trans{a \rb}q'$ or
      $q\trans{\q{a}} q'$, then $a\fresh q'$ .

    \item \emph{Finite branching}: Up to left $\alpha$-invariance,
      every state has only finitely many outgoing transitions: For
      each $q \in Q$, the sets
      $\set{(a, q')\mid (q,a,q') \in \Delta}$,
      $\set{\braket{a} q'\mid (q,\lb a,q') \in \Delta}$,
      $\set{(a,q')\mid (q,a\rb,q')\in \Delta}$, and
      $\set{\braket{a}q'\mid (q,\q{a},q')\in \Delta}$ are finite.
    \end{itemize}
    A state $q$ accepts a word $w$ if there is an \emph{accepting run}
    on~$w$ in the usual sense, i.e.~a sequence of successive
    transitions with sequence~$w$ of labels that starts in~$q$ and
    ends in a final state. The NDA~$A$ accepts the word~$w$ if the
    initial state $i$ accepts $w$.  The \emph{literal language of an
      NDA $A$} is the set $L_0(A)$ of words accepted by $A$, and its
    \emph{alphatic language} $L_\alpha(A)$ is the set of
    $\alpha$-equivalence classes of words accepted by $A$:
    $L_\alpha(A) = \set{[w]_\alpha\mid w \in L_0(A)}$.  The \emph{data
      language} of an NDA $A$ is the set
    $L_D(A) = \db[L_0(A)] \subseteq \names^*$.
    The local freshness semantics (of the alphatic language)
    of $A$ is the set $\D(\la(A))$.
    The
    \emph{degree} of an orbit-finite nominal set $X$ is
    $\deg(X) = \max_{n\in X} |\supp(x)|$. The \emph{degree} $\deg(A)$
    of an NDA $A = (Q,\Delta, i, F)$ is $\deg(Q)$.

\end{defn}

\noindent The notions of languages introduced above are
visualized in \cref{ex:ndaLang}.

\begin{example}\label{ex:ndaLang}
  Let $A = (\set{i,q(a),f(a,b)\mid a,b\in \names, a\neq b},\set{i\trans{\lb a}q(a),
  q(a)\trans{\lb b}f(a,b)\mid a,b\in \names, a\neq b},i,\set{f(a,b)\mid a,b\in \names,
  a\neq b})$ be the NDA visualized in the following scheme:
  
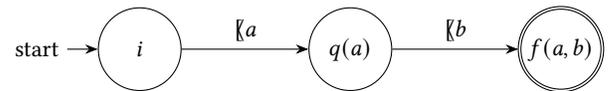
\begin{center}
    \begin{tikzpicture}[>=Stealth, node distance=2.8cm, auto,
      state/.style={circle, draw, minimum size=1.1cm, inner sep=0pt}]
      \node[state,initial] (i) {$i$};
      \node[state, right of=i] (q) {$q(a)$};
      \node[state,accepting, right of=q] (f) {$f(a,b)$};

      \path[->]
      (i) edge[above] node {$\lb a$} (q)
      (q) edge[above] node {$\lb b$} (f);
    \end{tikzpicture}
  \end{center}
  For $A$, we have the literal language $L_0(A) = \set{\lb a\lb b\mid a,b\in \names,
  a\neq b}$, the alphatic language $\la(A) = \set{[\lb a\lb b]_\alpha\mid a,b \in \names}$,
  the data language $L_D(A) = \set{ab\mid a,b\in \names, a\neq b}$, and the local freshness
  semantics (of the alphatic language) of $A$ given by
  $\D(\la(A)) = \set{ab\mid a,b \in \names}$.
  \emph{Nota bene}: While $a$ and $b$ are unequal for $L_0(A)$ and $L_D(A)$
  since there is no outgoing $\lb a$ transition from $q(a)$,
  they may coincide in $\la(A)$ and $\D(\la(A))$ because of
  $\alpha$-equivalent renaming.
\end{example}

\noindent We have informally recalled \emph{regular non-deterministic
  nominal automata (RNNA)}~\cite{SchroderEA17} in \Cref{sec:prel};
formally, we can now define an RNNA as an NDA without deallocating or
unknown transitions.\medskip

The right-handed analogue of left $\alpha$-invariance
follows from name erasure and equivariance:
\begin{lemma}[Right $\alpha$-invariance]\label{lem:rightR}
  Let $A = (Q,\Delta,i,F)$ be an NDA, and let $q,q',q''\in Q$,
  $a,b\in\names$ such that $\braket{a}\, q = \braket{b}\, q''$. Then
  the following holds:
  \begin{enumerate}
  \item If $q \trans{a\rb} q'$ in~$A$, then also $q'' \trans{b\rb} q'$.
  \item If $q \trans{\q{a}} q'$ in~$A$, then also
    $q'' \trans{\q{b}} q'$.
  \end{enumerate}
\end{lemma}
\begin{proofappendix}{lem:rightR}
  Assume w.l.o.g.~that $a\neq b$. We use \cref{lem:supptrans}. We
  prove both claims synchronously. So let
  $(q \trans{\gamma}q') \in \Delta$ where $\gamma = a \rb$ or
  $\gamma = \q{a}$.  Then $a\fresh q'$ by name erasure. Since
  $\braket{a} q = \braket{b} q''$, we have $(ab)\cdot q = q''$ and
  $b\fresh q$. 
  By \cref{lem:supptrans}, it follows that $b\fresh q'$.  By
  equivariance of~$\Delta$, we thus have
  $(ab)\cdot(q\trans{\gamma}q') = (((ab)\cdot
  q)\trans{(ab)\cdot\gamma}((ab)\cdot q')) = (q'' \trans{\beta} q')
  \in \Delta$ where $\beta = b \rb$ or $\beta = \q{b}$, respectively.
\end{proofappendix}
\noindent We illustrate the $\alpha$-invariance properties of NDAs in
\Cref{fig:NDATransitionProperties}.
\begin{rem}\label{rem:right-alpha}
  For unknown transitions, left $\alpha$-invariance means
  equivalently that whenever $q \trans{\q{a}} q' $, then
  $q \trans{\q{b}} q' $ for every $b \notin \supp(q')$ (while we
  explicitly allow $b\in\supp(q)$). Correspondingly, finite branching
  on unknown transitions means equivalently that for every state~$q$,
  the set
  $\set{q'\mid q\trans{\q{a}}q'\text{ for some~$a\in\names$}}$
  is finite.

  We note that also the converse of \cref{lem:rightR} holds; that is,
  under the remaining conditions imposed on NDA, right
  $\alpha$-invariance implies name erasure by \Cref{lem:supp-pfin};
  details are in the appendix. Name erasure, which confirms the
  intuition that deallocated names should be actually forgotten, thus
  is equivalent to the formally natural property of right
  $\alpha$-invariance.
\end{rem}
\begin{proofappendix}[Details for]{rem:right-alpha}
  We prove name erasure for $a\rb$; name erasure for $\q{a}$ is shown
  similarly. Note that by right $\alpha$-invariance, the equivariant
  map~$\names\times Q\to\pfin Q$,
  $(a,q)\mapsto\{q'\mid q\trans{a\rb}q'\}$ factors through the
  quotient map $\names\times Q\to [\names]Q$,
  $(a,q)\mapsto\braket{a}q$. Then apply \Cref{lem:supp-pfin} and use
  that $a\notin\supp(\braket{a}q)$.
\end{proofappendix}

Similarly as for words, we introduce uniform notation for unknown
transitions:
\begin{notn}
  We write $q \trans{?} q'$ to represent all transitions \mbox{$q \trans{\q{a}} q'$} in an NDA.
\end{notn}

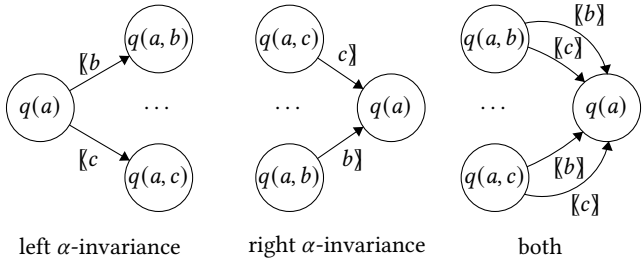
\begin{figure}
  \centering
  \begin{tikzpicture}[scale=0.15]
    \tikzstyle{every node}+=[inner sep=0pt]
    \draw [black] (5.7,-27.7) circle (3);
    \draw (5.7,-27.7) node {$q(a)$};
    \draw [black] (16.2,-21.6) circle (3);
    \draw (16.2,-21.6) node {$q(a,b)$};
    \draw [black] (16.2,-33.9) circle (3);
    \draw (11,-40) node {left $\alpha$-invariance};
    \draw (16.2,-33.9) node {$q(a,c)$};
    \draw (16.2,-27.75) node {$\dots$};
    \draw [black] (27.8,-21.6) circle (3);
    \draw (27.8,-21.6) node {$q(a,c)$};
    \draw [black] (27.8,-33.9) circle (3);
    \draw (27.8,-33.9) node {$q(a,b)$};
    \draw (27.8,-27.75) node {$\dots$};
    \draw (32,-40) node {right $\alpha$-invariance};
    \draw [black] (36.8,-27.7) circle (3);
    \draw (36.8,-27.7) node {$q(a)$};
    \draw [black] (46,-21.6) circle (3);
    \draw (46,-21.6) node {$q(a,b)$};
    \draw [black] (46,-33.9) circle (3);
    \draw (46,-33.9) node {$q(a,c)$};
    \draw (46,-27.75) node {$\dots$};
    \draw (50,-40) node {both};
    \draw [black] (55.9,-27.7) circle (3);
    \draw (55.9,-27.7) node {$q(a)$};
    \draw [black] (8.29,-26.19) -- (13.61,-23.11);
    \fill [black] (13.61,-23.11) -- (12.66,-23.08) -- (13.17,-23.94);
    \draw (10,-24.8) node [above] {$\lb b$};
    \draw [black] (8.28,-29.23) -- (13.62,-32.37);
    \fill [black] (13.62,-32.37) -- (13.18,-31.54) -- (12.67,-32.4);
    \draw (9.89,-31.3) node [below] {$\lb c$};
    \draw [black] (30.28,-23.28) -- (34.32,-26.02);
    \fill [black] (34.32,-26.02) -- (33.93,-25.15) -- (33.37,-25.98);
    \draw (33,-24) node [above] {$c\rb$};
    \draw [black] (30.27,-32.2) -- (34.33,-29.4);
    \fill [black] (34.33,-29.4) -- (33.39,-29.44) -- (33.95,-30.27);
    \draw (33.44,-31.3) node [below] {$b\rb$};
    \draw [black] (48.624,-20.211) arc (103.70508:13.01535:6.057);
    \fill [black] (55.96,-24.73) -- (56.27,-23.84) -- (55.29,-24.06);
    \draw (54.38,-20.44) node [above] {$\q{b}$};
    \draw [black] (56.087,-30.663) arc (-10.66072:-105.22453:6.024);
    \fill [black] (56.09,-30.66) -- (55.45,-31.36) -- (56.43,-31.54);
    \draw (54,-35.3) node [below] {$\q{c}$};
    \draw [black] (53.794,-29.832) arc (-49.14275:-66.7425:19.113);
    \fill [black] (53.79,-29.83) -- (52.86,-29.98) -- (53.52,-30.73);
    \draw (52.58,-32.07) node [below] {$\q{b}$};
    \draw [black] (48.908,-22.313) arc (70.25053:46.46991:14.386);
    \fill [black] (53.96,-25.42) -- (53.72,-24.51) -- (53.03,-25.23);
    \draw (52.5,-23.35) node [above] {$\q{c}$};
\end{tikzpicture}
  \caption{Invariance of NDA transitions under $\alpha$-equivalence}
  \label{fig:NDATransitionProperties}
\end{figure}

\begin{example}
  We model a small logging system using the NDA $A=(Q,\Delta, i,F)$ given by
  $Q=\set{s} \times \names^2 \cup \set{q_1} \times \names^3 \cup \set{q_2} \times \names^4$,
  $i = s(a,d)$, $F = \set{s}\times \names^2$, and $\Delta$ consisting of transitions as shown in
  \cref{fig:NDALog}, implicitly closed under all requirements on the transition relation of an NDA
  (only the fragment for $s(c,d)$, $q_1(a,c,d)$, $q_1(b,c,d)$ and $q_2(a,b,c,d)$ is displayed).

  The infinite alphabet $\names$ represents identifiers of users.  The
  support of a state in the automaton can then be interpreted as the
  currently known users.  Logging in and out is represented as binding
  and deallocating the corresponding identifier, respectively.  The
  example supports two new users with identifiers~$a$ and~$b$ logged
  in simultaneously (the $?$-transition on~$q_4$ represents failed
  login attempts of additional users that would exceed this capacity).
  IDs recognized globally are already contained in the support of the
  initial state.  In our example, this implements an allowlist and a
  blocklist: The user with ID~$c$ is an administrator; she does not
  log in or out and can perform actions (represented as looping free
  $c$-transitions) independently of the state of the machine.  On the
  other hand, the user with ID~$d$ is excluded from the logging
  system; she cannot perform actions (i.e.~free $d$-transitions) in
  any state and can log neither in nor out, formally because renaming
  transitions on $\lb a$ or $\lb b$ into $\lb d$ is blocked
  everywhere, cf.~\Cref{sec:prel}.  The words accepted by~$A$ are then
  precisely the valid logs where each user has properly logged in and
  out.

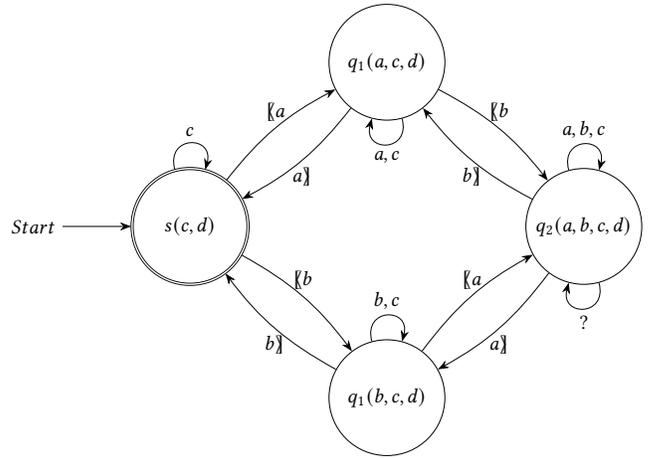
\begin{figure}
  \centering
  \resizebox{\columnwidth}{!}{%
    \begin{tikzpicture}[scale=0.22,>=Stealth]
      \tikzset{state/.style={draw,circle,minimum size=1.8cm}}
      \node[state,accepting] (s) at (22.4,-30.5) {$s(c,d)$};
      \node[state] (qonea) at (36.3,-18.9) {$q_1(a,c,d)$};
      \node[state] (qtwo) at (50.2,-30.5) {$q_2(a,b,c,d)$};
      \node[state] (qoneb) at (36.3,-42.6) {$q_1(b,c,d)$};

      \coordinate (start) at (13.4,-30.5);
      \draw[->] (start) -- (s.west);
      \node[left] at (start) {$Start$};

      \path[->]
        (s) edge[bend left=12] node[above] {$\lb a$} (qonea)
        (qonea) edge[bend left=12] node[below] {$a\rb$} (s)
        (s) edge[bend left=12] node[above] {$\lb b$} (qoneb)
        (qoneb) edge[bend left=12] node[below] {$b\rb$} (s)
        (qonea) edge[bend left=12] node[above] {$\lb b$} (qtwo)
        (qtwo) edge[bend left=12] node[below] {$b\rb$} (qonea)
        (qoneb) edge[bend left=12] node[above] {$\lb a$} (qtwo)
        (qtwo) edge[bend left=12] node[below] {$a\rb$} (qoneb)
        (s) edge[loop above, looseness=3] node[above] {$c$} (s)
        (qonea) edge[loop below, looseness=3] node[below] {$a,c$} (qonea)
        (qoneb) edge[loop above, looseness=3] node[above] {$b,c$} (qoneb)
        (qtwo) edge[loop above, looseness=3] node[above] {$a,b,c$} (qtwo);

      \path[->] (qtwo) edge[loop below,looseness=3] node[below] {$?$} (qtwo);
    \end{tikzpicture}%
  }
  \caption{An example NDA accepting valid logs of sessions with at most two participants (a and~b) and an admin (c).
  The NDA accepts all logs where all users except for the admin are logged out in the end.}\label{fig:NDALog}
\end{figure}

\end{example}



\noindent We record a few basic properties of NDA. First, by
\Cref{lem:supp-pfin}, equivariance and finite branching imply the
following estimates on supports of poststates of transitions:

\begin{lemma}[Support lemma]\label{lem:supptrans} For all NDA
  $A = (Q,\Delta,i,F)$, $q,q' \in Q$, and $a
  \in \names$, we have:
    \begin{enumerate}
        \item\label{it:free} If $q\trans{a}q'\in \Delta$, then $\supp(q') \cup \{a\} \subseteq \supp(q)$,
        \item\label{it:alloc} If $q\trans{\lb a}q' \in \Delta$, then $\supp(q') \subseteq \supp(q) \cup \{a\}$,
        \item \label{it:unknown} If $q\trans{\q{a}}q' \in \Delta$, then $\supp(q') \subseteq
        \supp(q)\setminus \set{a}$, $a \notin \supp(q)$,
        \item\label{it:dealloc} If $q\trans{a \rb}q' \in \Delta$, then $\supp(q') \subseteq \supp(q) \setminus \{a\}$.
        
    \end{enumerate}
\end{lemma}
\noindent That is, the support evolves along transitions in analogy to
the register paradigm: Transitions for $a$ or $a\rb$ can only be taken
if~$a$ is in memory (i.e.~in the support); names~$a$ are added to the
memory via transitions $\lb a$, i.e.~morally by reading~$a$ from the
input; and transitions under~$a\rb$ or~$\q{a}$ erase~$a$ from
memory. This entails that, as indicated in \Cref{rem:rns}, NFA operate
only on right-non-shadowing words:

\begin{proofappendix}{lem:supptrans}
  \sloppy Apply \Cref{lem:supp-pfin} to the following maps to prove
  each claim:
  \begin{enumerate}
  \item $Q\to\pfin(\names\times Q)$,
    $q\mapsto\{(a,q')\mid q\trans{a}q'\}$;
  \item $Q\to\pfin([\names]Q)$,
    $q\mapsto\{\braket{a}q'\mid q\trans{\lb a}q'\}$;
  \item $Q\to\pfin(Q)$,
    $q\mapsto\{\braket{a}q'\mid q\trans{\q{a}}q'\text{ for some~$a$}\}$;
  \item $Q\to\pfin(\names\times Q)$,
    $q\mapsto\{q'\mid q\trans{a\rb}q'\}$.
  \end{enumerate}
  For Claims~\ref{it:unknown} and~\ref{it:dealloc}, one additionally
  uses name erasure, and for Claim~\ref{it:alloc}, one uses
  additionally that
  $\supp(\braket{a}q')\subseteq\supp(q')\setminus a$.

\end{proofappendix}


\begin{proposition}\label{prop:ndaRns}
  Every word accepted by an NDA is right non-shadowing. 
\end{proposition}
\begin{proofappendix}{prop:ndaRns}
  Let $w$ be a word accepted by an NDA, and let $w = uv$ be any decomposition of $w$.
  We show that $\rc(u) \cap \lo(v) = \emptyset$.
  If $a \in \rc(u)$, then $a$ is deallocated (or unknown) in $u$,
  so $a$ is removed from the state support after processing $u$.
  By \cref{lem:supptrans}, once $a$ is removed from the support,
  no free or deallocating transition is possible for $a$, hence
  it cannot appear again in the remaining word $v$,
  and in particular cannot be left-open in $v$.
  Hence $\rc(u) \cap \lo(v) = \emptyset$.
\end{proofappendix}

For convenience, we introduce \emph{$\epsilon$-transitions} into
NDA. An \emph{NDA with $\epsilon$-transitions} is defined like an NDA
but can additionally have \emph{$\epsilon$-transitions},
i.e.~transitions of the form $q\trans{\epsilon}q'$; these are also
subject to equivariance and finite branching, i.e.~for every~$q$, the
set~$\{q'\mid q\trans{\epsilon}q'\}$ is required to be finite; this
implies an extension of the support lemma (\Cref{lem:supptrans})
stating that whenever $q\trans{\epsilon}q'$, then
$\supp(q')\subseteq\supp(q)$. Acceptance is then defined as usual in
automata models with $\epsilon$-transitions, i.e.~accepting runs can
contain $\epsilon$-transitions, and these transitions do not consume
letters of the input word. (Formally, an accepting run on a word~$w$
has the shape $i\trans{v_1}q_1\trans{v_2}q_2\dots\trans{v_n}q_n\in F$
where each $v_1$ is either a single letter from~$\ah$ or~$\epsilon$
and $v_1\dots v_n=w$.)

\begin{lemma}[$\epsilon$-Elimination]\label{lem:eps}
  For every NDA with $\epsilon$-transitions, there exists an NDA
  accepting the same literal language.
\end{lemma}

\begin{proofsketch}
  By the usual method of $\epsilon$-elimination as for NFA; one checks
  easily that applying this construction to an NDA with
  $\epsilon$-transitions does yield an NDA.
\end{proofsketch}

\begin{proofappendix}{lem:eps}
  We proceed by the standard method of
  $\epsilon$-elimination. Let~$A=(Q,\Delta,i,F)$ be an NDA with
  $\epsilon$-transitions; we construct an NDA $A'=(Q,\Delta',i,F)$
  without $\epsilon$-transitions that accepts the same language. We
  define $\Delta'$ by taking $(q\trans{\gamma}q')\in\Delta'$ for
  $q,q'\in q$ and $\gamma\in\ahs$ iff there exist $q_0,\dots,q_n\in Q$
  such that
  $q=q_0\trans{\epsilon}q_1\dots\trans{\epsilon}q_n\trans{\alpha}q'$
  in~$A$. It is immediate from the definition of accepting runs in~$A$
  that~$A'$ accepts the same literal language as~$A$ (every accepting
  run in~$A$ induces an accepting run in~$A'$ and vice versa). It
  remains only to show that~$A'$ is indeed an NDA. Equivariance is
  clear; name erasure is inherited from~$A$ by the support lemma for
  $\epsilon$-transitions. Left $\alpha$-invariance is immediately
  inherited from~$A$. Finally, finite branching follows from the fact
  that the relevant sets of successors as listed in \Cref{def:nda}
  are, by the support lemma, uniformly finitely supported, and hence
  finite because~$Q$ is orbit-finite.
\end{proofappendix}

By equivariance of the transition relation and the set of final
states, it is immediate that acceptance of words is equivariant:
\begin{lemma}\label{lem:acceptEquiv}
  If a state~$q$ of an NDA accepts a word $w$, then $\pi \cdot q$
  accepts $\pi \cdot w$.
\end{lemma}


%


We have seen above that RNNA can be regarded as NDA without
deallocating or unknown transitions. Under local freshness semantics,
we also have the converse inclusion:
\begin{proposition}\label{prop:NDA2RNNA}
  Under local freshness semantics, NDA and RNNA are equiexpressive.
\end{proposition}
\begin{proofsketch}
  In all transitions, replace $a\rb$ with~$a$ and $\q{a}$
  with~$\lb a$.
\end{proofsketch}
\takeout{ 
\begin{inappendix}
  \subsection*{Details for \cref{prop:NDA2RNNA}}
  
  We recall that an RNNA is a tuple $A = (Q,\Delta',i,F)$ where $Q$ is an orbit-finite set
  of states, $i\in Q$ is the initial state, $F\subseteq Q$ is an equivariant set of final
  states, and $\Delta'\subseteq Q\times \bar\names \times Q$, where
  $\bar\names = \set{a,|a\mid a\in \names}$, is an equivariant subset of
  $Q\times \bar\names \times Q$, the transition relation, such that $\Delta'$ is
  \emph{(left) $\alpha$-invariant} and \emph{finitely branching up to
    $\alpha$-invariance}. The two latter notions are completely analogous as for the
  transition relation of an NDA: We have $q \trans{|a} q' \in \Delta'$ whenever there is
  $q \trans{|b} q'' \in \Delta'$ with $\braket{a}q' = \braket{b}q''$, and the sets
  $ \set{(a, q')\mid q \trans{a} q' \in \Delta'}$ and
  $\set{\braket{a}q' \mid q \trans{|a} q' \in \Delta'}$ are finite.

  \begin{proof}[Proof of \cref{prop:NDA2RNNA}]
    Since RNNA are a proper subclass of NDA, it is clear that RNNA are at most equally
    expressive than NDA.  For the reverse inclusion, we show that given an NDA
    $A = (Q,\Delta,i,F)$ there is an RNNA $B = (Q,\Delta',i,F)$ such that $L_D(A) = L_D(B)$.
    We define the transition relation $\Delta'$ on the nominal state set $Q$ (leaving $i$
    and $F$ unchanged) by modifying the transitions of $\Delta$ as follows:
    \begin{itemize}
    \item For each $(q,a\rb,q')\in \Delta$ or $(q,a,q')\in \Delta$,
      put $(q,a,q') \in \Delta'$.
      
    \item For each $(q,\q{a},q') \in \Delta$ and $(q,\lb a,q') \in \Delta$, put
      $(q,|a,q') \in \Delta'$. 
      the set of new $|a$-transitions is finite up to $\alpha$-equivalence as well.
    \end{itemize}
    Since $\Delta$ is finitely branching up to $\alpha$-invariance, it follows that so is
    $\Delta'$.  \takeout{ Concretely, we define $\Delta'$ by:
      \begin{align*}
        \Delta' =&
        \set{(q,a,q') \mid a\in \names, (q,a,q') \in \Delta}\\
        &\cup \set{(q,\lb a,q') \mid (q,\lb a,q') \in \Delta}\\
        &\cup \set{(q,a,q') \mid (q,a\rb,q') \in \Delta}\\
        &\cup \set{(q,\lb a,q') \mid (q,?,q') \in \Delta, a\in \names\setminus \supp(q')}
      \end{align*}}
    Furthermore, $\Delta'$ is left $\alpha$-invariant, because so is $\Delta$.
    Then $B = (Q,\Delta',i,F)$ contains no $\q{a}$ and $a\rb$-transitions and thus,
    it is an RNNA. Furthermore, for each $q \trans{\gamma} q'$ in $\Delta$,
    we have $q \trans{\beta} q'$ in $\Delta'$ such that $\db(\gamma) = \db(\beta)$, where
    $\db$ is defined on bar letters analogously as on letters from $\ah$: $\db(a) =
    \db(\newletter a) = a$. It follows that every run on a word $w \in \ahs$ in $A$
    yields a run on some $w' \in \bar\names^*$ in $B$ such that $\db(w) = \db(w')$.
    Conversely, for every run of a word $w \in \bar\names^*$ in $A'$ there is a run for some
    word $w' \in ahs$ in $A$ such that $\db(w) = \db(w')$.
    Consequently, we obtain $L_D(A) = L_D(B)$.
  \end{proof}
\end{inappendix}
}

\begin{proofappendix}{prop:NDA2RNNA}
  It remains only to construct, given an NDA $A = (Q,\Delta,i,F)$, an
  RNNA $B = (Q,\Delta',i,F)$ that accepts the same data language
  as~$A$, i.e.~$L_D(A) = L_D(B)$.  We define the transition
  relation~$\Delta'$ as follows:
  \begin{itemize}
  \item Introduce a transition $(q,a,q') \in \Delta'$ whenever
    $(q,a\rb,q')\in \Delta$ or $(q,a,q')\in \Delta$; and
      
  \item Introduce a transition $(q,\lb a,q') \in \Delta'$ whenever
    $(q,\q{a},q') \in \Delta$ or $(q,\lb a,q') \in \Delta$.
    
  \end{itemize}
  The it is then clear that~$B$ is finitely branching up to
  $\alpha$-equivalence and left $\alpha$-invariant.  Furthermore, for
  every transition $q \trans{\gamma} q'$ in~$\Delta$, we have a
  transition $q \trans{\beta} q'$ in $\Delta'$ such that
  $\db(\gamma) = \db(\beta)$. It follows that every run on a word
  $w \in \ahs$ in $A$ yields a run on some word~$w'$ (without any
  letters $a\rb$ and $\q{a}$) in $B$ such that $\db(w) = \db(w')$.
  Conversely, for every run of a word~$w$ (without any letters $a \rb$
  and $\q{a}$) in $B$ there is a run for some word $w' \in \ahs$ in
  $A$ such that $\db(w) = \db(w')$. Consequently, we obtain
  $L_D(A) = L_D(B)$.
\end{proofappendix}

The above construction of a data-language equivalent RNNA from an NDA
is very simple, so under local freshness semantics, NDA inherit from
RNNA the algorithmic tractability of language
inclusion~\cite[Cor.~7.4]{SchroderEA17}:
\begin{theorem}
  Under local freshness semantics, language inclusion of NDAs is
  decidable in exponential space, in fact in parametrized polynomial
  space, with the degree as parameter.
\end{theorem}
\noindent We note here that in translations between nominal automata
models and register automata~\cite{BojanczykEA14,SchroderEA17}, the
degree corresponds to the number of registers.

\section{Name Dropping and Closure under $\alpha$-Equivalence}\label{sec:ndm}


Left and right $\alpha$-invariance of the transition relation of an
NDA may appear to introduce enough structural symmetry to ensure
closure of the accepted language under $\alpha$-equivalence. For some
NDA this is indeed the case, as the following lemma shows:

\begin{lemma}[label=lem:alphaClosedSupp]
  Let $A=(Q,\Delta,i,F)$ be an NDA. Suppose that each state $q\in Q$
  has as its support precisely the left-open names of each word
  accepted by $q$, in symbols: $\supp(q) = \lo(w)$ for all
  $w\in L_0(q)$. Then $L_0(A)$ is closed under $\alpha$-equivalence.
\end{lemma}

\begin{proofappendix}{lem:alphaClosedSupp}
  We generalize to languages $L_0(q)$ of states~$q$ of $A$, showing by
  induction on the derivation of $w'\aeq w$ as per \Cref{def:Alpha}
  that whenever $w\in L_0(q)$, then $w'\in L_0(q)$. The cases for
  transitivity and Rule~\ref{rule:epsilon} are trivial. In the
  remaining cases, we have $w=\gamma v$ for some $\gamma\in\ah$, and
  $q\trans{\gamma}q'$ where~$q'$ accepts~$v$.
  \begin{itemize}
    \sloppy
  
  \item Rule~\ref{rule:alpha}: In this case, $\gamma=\lb a$ and
    $w' = \lb bv' $ where $a,b\in \names$ and
    $\braket{b}v' = \braket{a}v$. Assuming w.l.o.g.~that $a\neq b$, we
    have $v' = (ab) \cdot v$ and $b \notin \supp(v)$, in particular
    $b \notin \lo(v)$. By hypothesis, it follows that
    $b\notin\supp q'$, so that $\braket{b}q'' = \braket{a}q'$ for
    $q'' = (ab) \cdot q'$.  By equivariance of acceptance
    (\Cref{lem:acceptEquiv}),~$q''$ accepts~$v'$, and by left
    $\alpha$-invariance, $q \trans{\lb b} q''$, so~$q$ accepts
    $\lb bv'$.
  \item Rule~\ref{rule:cong-free-open}: In this case, $\gamma$ is of
    the form $\gamma = a$ or $\gamma=\lb a$ where $a\in\names$, and
    $w'=\gamma v'$ where $v\aeq v'$. Then $q'$ accepts $v'$ by the
    induction hypothesis, hence $q$ accepts $w'$.
  
  \item Rule~\ref{rule:cong-close}: Analogous to
    Rule~\ref{rule:cong-free-open}. 
    
  \item Rule~\ref{rule:cong-unknown}: In this case, $\gamma=\q{a}$ and
    $w' = \q{b}v'$ where $v\aeq v'$ and
    $\set{a,b} \cap \lo(v) = \emptyset$. By the induction
    hypothesis,~$q'$ accepts~$v'$.  Since~$q'$ accepts~$v$, we have
    $\supp(q') = \lo(v)$ by hypothesis, hence $b \notin \supp(q')$.
    By left $\alpha$-invariance, we obtain $q \trans{\q{b}} q'$,
    so~$q$ accepts $w'$.\qedhere
  \end{itemize}
\end{proofappendix}
\noindent In general, however, closure under $\alpha$-equivalence may
fail. Consider the NDA represented pictorially as follows:
%
%
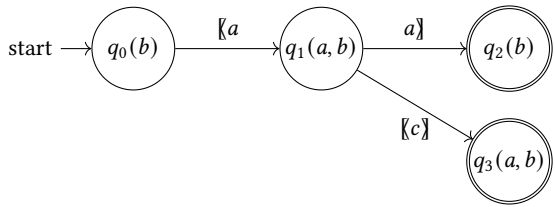
\begin{center}
  \begin{tikzpicture}[align=center, node distance = 2.5cm, state/.style={circle, draw,
      minimum size=1.1cm, inner sep=0pt}, row sep = 1cm]
    \node[state,initial] (q0) {$q_0(b)$};
    \node[state, right of=q0] (q1) {$q_1(a,b)$};
    \node[state,accepting, right of=q1] (q2) {$q_2(b)$};
    \node[state,accepting] at (5,-1.5) (q3) {$q_3(a,b)$};
    \draw[->]
    (q0) edge[above] node {$\lb a$} (q1)
    (q1) edge[above] node {$a \rb$} (q2)
    (q1) edge[below] node[inner sep=5pt] {$\q{c}$} (q3);
  \end{tikzpicture}
\end{center}
This NDA accepts the word $\lb a a \rb$ but not the
$\alpha$-equivalent word $\lb b b \rb$, since renaming the
$\lb a$-transition into $\lb b$ is blocked by $b\in\supp(q_1(a,b))$.
Similarly, the word $\lb a \q{c}$ is accepted but not the
$\alpha$-equivalent word $\lb a \q{b}$, since renaming~$\q{c}$ into
$\q{b}$ is blocked by $b\in\supp(q_1(a,b))$. In terms of the condition
of \Cref{lem:alphaClosedSupp}, notice that $a\rb\in L_0(q_1(a,b))$ and
$b\in\supp(q_1(a,b))$ but $b\notin\lo(a\rb)$, and similarly that
$\epsilon\in L_0(q_3(a,b)$ and $b\in\supp(q_3(a,b))$ but
$b\notin\lo(\epsilon)$.

We proceed to present a construction on NDA closing their language
under $\alpha$-equivalence that has been employed in a similar fashion
in related
work~\cite{frank2025alternatingnominalautomataallocation,SchroderEA17,PruckerSchroder24arXiv}. The
construction relies on different principles than
\Cref{lem:alphaClosedSupp}; we will return to the latter in
\Cref{sec:determinization}.

First, we establish that we may assume w.l.o.g.~that the nominal state
set of an NDA has a more explicit description. To this end, recall
that a nominal set $X$ is \emph{strong}~\cite{Tzevelekos07} if, for
all $x\in X$ and $\pi\in \permG$, one has $\pi\cdot x = x$ if and only
if $\pi$ fixes every element of $\supp(x)$.  (The `if' direction holds
in every nominal set.) Furthermore, an equivariant map
$f\colon X \to Y$ is \emph{support-reflecting} if
$\supp(f(x)) = \supp(x)$ for every $x \in X$ (recall that
$\supp(f(x)) \subseteq \supp(x)$ always holds). For each nominal set
$Q$, we have a support-reflecting, surjective and equivariant map
$e\colon P\epito Q$ with a strong nominal domain $P$, where~$P$ can be
taken to be orbit-finite in case~$Q$ is orbit-finite~\cite[full
version, Cor.~B.27.1]{MiliusUrbat19}. For example, the
nominal sets $\names^{\#n}$, $\names^n$ and $\names^*$ are strong,
where $\names^{\#n}$ denotes the $n$-fold fresh product of~$\names$
with itself, i.e.~the set of~$n$-tuples over~$\names$ with pairwise
distinct entries. Up to isomorphism, orbit-finite strong nominal sets
are precisely coproducts $\coprod_{j=1}^n \names^{\# n_j}$, where
$n_j\in \N$ and $n$ is the number of orbits of the nominal set
(e.g.~\cite[full version, Cor.~B.27]{MiliusUrbat19}).

For $n\in\N$, we write $\names^{\$n}$ for the nominal set of partial
injective maps $[n] \parto \names$, where $[n]$ denotes the set
$\set{1, \ldots, n}$. Here, the group action of~$\permG$ on
$\names^{\$n}$ is pointwise, as expected:
$(\pi \cdot f)(j) = \pi (f(j))$ if $f(j)$ is defined; otherwise,
$(\pi \cdot f)(j)$ is undefined. We may identify elements of
$\names^{\#n}$ with total injective maps $[n]\to\names$, making
$\names^{\#n}$ a nominal subset of $\names^{\$n}$. A total injective
map $\bar f \in \names^{\#n}$ is said to \emph{extend} an element
$f\in \names^{\$n}$ if for every $j = 1, \ldots, n$, whenever $f(j)$
is defined, then $\bar f (j) = f(j)$. Observe that $\names^{\$n}$ is
orbit-finite, but, unlike $\names^{\#n}$, it has more than one
orbit. We generally write elements of sums
$\sum_{j=1}^n \names^{\$n_j}$ or $\sum_{j=1}^n \names^{\#n_j}$ in the
form~$(j,f)$ where $f\colon [n_j]\parto\names$.

\begin{proposition}\label{prop:ndaStrong}
  For every NDA there is an NDA with a strong nominal state set accepting the same literal
  language (hence, also the same alphatic language).
\end{proposition}
\begin{proofappendix}{prop:ndaStrong}
  Given an NDA $A=(Q,\Delta,i,F)$, we pick a support-reflecting, surjective and equivariant
  map $e\colon P\epito Q$ with $P$ strong and orbit-finite.  We define
  $A'=(P,\Delta',i',F')$ by choosing some $i' \in P$ such that $e(i') = i$, putting
  $F' = \set{f \in P \mid e(f) \in F}$, and
  \[
    \Delta'
    =
    \set{p \trans{\gamma} q \in P\times\ah\times P \mid \text{$e(p) \trans{\gamma} e(q)$ in
        $\Delta$}}.
  \]
  (1)~We verify that this indeed defines an NDA.
  \begin{itemize}
  \item Equivariance: If $p \trans{\gamma} q$ in $\Delta'$, then $e(p)\trans{\gamma}e(q)$ in
    $\Delta$.  For every permutation $\pi$, equivariance of $e$ and $\Delta$ yields
    $e(\pi\cdot p)=\pi\cdot e(p) \trans{\pi\cdot\gamma} \pi\cdot e(q)=e(\pi\cdot q)$, hence
    $(\pi\cdot p)\trans{\pi\cdot\gamma}(\pi\cdot q)$ in $\Delta'$.

  \item Left $\alpha$-invariance: Suppose $p\trans{\gamma}q$ in $\Delta'$,
    $\gamma = \lb a$ or $\gamma = \q{a}$, and
    $\braket{a} q=\braket{b} q''$. Then
    $\braket{a} e(q)=\braket{b} e(q'')$ by equivariance of~$e$. From
    $e(p)\trans{\gamma} e(q)$ and left $\alpha$-invariance of~$\Delta$ we obtain
    $e(p)\trans{(ab)\cdot\gamma} e(q'')$, hence $p\trans{(ab)\cdot\gamma} q''$ in $\Delta'$.  
      
  \item Right $\alpha$-invariance: If
    $p\trans{\gamma}q$, $\gamma=a\rb$ or $ \q{a}$ and $\braket{a} p=\braket{b} p''$, then
    $\braket{a} e(p)=\braket{b} e(p'')$.  From $e(p)\trans{\gamma}e(q)$ and
    right $\alpha$-invariance of $\Delta$ we obtain $e(p'')\trans{(ab)\cdot\gamma}e(q)$, hence
    $p''\trans{(ab)\cdot\gamma}q$ in $\Delta'$.

  \item Finite branching: Let $p\in P$. By definition, the sets
    $\{(a,q')\mid e(p)\trans{a}q'\}$, $\{\braket{a} q'\mid e(p)\trans{\lb a}q'\}$,
    $\{(a,q')\mid e(p)\trans{a\rb}q'\}$ and $\{\braket{a}q'\mid e(p)\trans{\q{a}}q'\}$ are finite.
    For each such successor $q'$, the set $\set{q\mid e(q)=q'}$ is finite. Moreover, $e$ is
    support-reflecting, so for every $q\in e^{-1}(q')$ we have $\supp(q)=\supp(q')$. In an
    orbit-finite nominal set, there are only finitely many elements with the same support,
    hence, for each type of transition, the set of successors of $p$ in $A'$ is finite.
  \end{itemize}
  Finally, $F'=e^{-1}(F)$ is equivariant because both $F$ and $e$ are.

  \medskip\noindent
  (2)~We verify that $A'$ accepts the same literal language as $A$: $L_0(A')=L_0(A)$. 
  \begin{itemize}
  \item $L_0(A')\subseteq L_0(A)$: If $p\trans{\gamma}q$ in $A'$, then
    $e(p)\trans{\gamma}e(q)$ in $A$ by definition of $\Delta'$, and $e(i')=i$,
    $e(F')\subseteq F$. Hence $e$ maps accepting runs of $A'$ to accepting runs of $A$.
    
  \item $L_0(A)\subseteq L_0(A')$: Let
    $i=q_0\trans{\gamma_1}q_1\trans{\gamma_2}\cdots\trans{\gamma_n}q_n\in F$ be an accepting
    run in $A$.  Choose $p_0=i'$ with $e(p_0)=q_0$. Inductively, given $p_k$ with
    $e(p_k)=q_k$ and a transition $q_k\trans{\gamma_{k+1}}q_{k+1}$ in $A$, pick any
    $p_{k+1}\in e^{-1}(q_{k+1})$; then by definition of $\Delta'$ we have
    $p_k\trans{\gamma_{k+1}}p_{k+1}$ in $A'$. Since $q_n\in F$ implies $p_n\in F'$, this
    yields an accepting run in $A'$.
  \end{itemize}
  \sloppy
  Therefore, we have $L_0(A)=L_0(A')$, and consequently $L_\alpha(A)=L_\alpha(A')$.\qedhere
\end{proofappendix}

\noindent In the following, we tacitly assume that the state set of an
NDA is a strong nominal set.

\begin{construction}[Name-dropping modification]\label{def:ndm}
  Given an NDA $A = (Q,\Delta,i,F)$ with a strong nominal state set
  $Q=\sum_{j=1}^n \names^{\#n_j}$, its \emph{name-dropping modification} $A_\bot$ is
  the NDA $(Q_\bot, \Delta_\bot, i, F_\bot)$ defined by the following data:
  \begin{enumerate}
  \item the nominal set of states is 
    $Q_\bot = \sum_{j=1}^n \names^{\$n_j}$;
    
  \item $(j,f)$ is final in $A_\bot$ if $(j,\bar{f})$ is final in $A$
    for~$\bar{f}$ extending $f$ (by equivariance of finality, this is
    independent of the choice of the extension~$\bar f$);
    
  \item $(j,f)\trans{a} (k,g)$  in $\Delta_\bot$ if $\supp(f) \supseteq \{a\}
    \cup \supp(g)$ and $(j,\bar f)\trans{a} (k,\bar g)$ in $A$ for some
    $\bar f$ and $\bar g$ extending $f$ and $g$, respectively;
    
  \item $(j,f)\trans{\lb a} (k,g)$ in $\Delta_\bot$ if
    $\supp(f) \cup \{a\} \supseteq \supp(g)$ and $(j,\bar f)\trans{\lb b} (k,\bar g)$
    in~$A$ for some $b \in \names$, $g'\in Q_\bot$ such that
    $\braket{b} g' = \braket{a} g$, and total injective maps $\bar f$ and
    $\bar g$ extending $f$ and~$g'$, respectively;
    
  \item For $a\in\supp(f)$, $(j,f)\trans{a\rb} (k,g)$ in $\Delta_\bot$
    if $\supp(f) \supseteq\supp(g)$ and
    $(j,\bar f)\trans{a\rb} (k,\bar g)$ in $A$ for some $\bar f$ and
    $\bar g$ extending~$f$ and $g$ (thus, also $a \notin \supp(k,g)$).

  \item $(j,f)\trans{\q{a}} (k,g)$ in $\Delta_\bot$ if
    $\supp(f) \supseteq\supp(g)$ and
    $(j, \bar f)\trans{\q{b}} (k,\bar g)$ in $A$ for $b\in \names$,
    $f,g' \in Q_\bot$ such that
    $\braket{b}g' = \braket{a}g$, for some $\bar f$ and $\bar g$
    extending $f$ and $g'$, respectively (i.e. a transition
    $(j,f)\trans{\q{a}} (k,g)$ in $\Delta_\bot$ for each $a\notin \supp(k,g)$).\spnote{passt das so?}
  \end{enumerate}
\end{construction}
\takeout{
In most cases, it is more natural to drop a set of names instead of asserting their
presence.  Therefore, we introduce a new notation, $q\drop{S}$, read as `$q$ drop $S$',
defined as $q\drop{S}(n) = q(n)$ if $q(n) \not \in S$, and $q\drop{S}$ is undefined
otherwise. For a single name $a$, we write $\drop a$ for $\drop{\set{a}}$. 
}
\begin{lemma}\label{nameDropWellDefNDA}
  The name-dropping modification of an NDA is an NDA.
\end{lemma}

\begin{proofsketch}
  We verify the conditions of an NDA for~$A_\bot$. The nominal set $Q_\bot$ is orbit-finite,
  since it is a finite co\-pro\-duct of orbit-finite nominal sets $\names^{\$n}$.
  Equivariance of $\Delta_\bot$ follows utilizing equivariance of $\Delta$ and observing
  that the defining support side-conditions are preserved under the group action.  Left and
  right $\alpha$-invariance are obtained by combining $\alpha$-invariance of $\Delta$ with
  the abstraction in the definition of $\Delta_\bot$.  Finally, finite branching holds
  because, for a fixed $(j,f)$ in $Q_\bot$, there are only finitely many possible extensions
  $\bar f$, only finitely many successors in $A$ up to $\alpha$-invariance, and only
  finitely many ways of dropping names from the supports of these successors.
\end{proofsketch}
\begin{proofappendix}{nameDropWellDefNDA}
  Let $A = (Q, \Delta, i, F) $ be an NDA, and let
  $ A_\bot = (Q_\bot, \Delta_\bot, i, F_\bot) $ be its name-dropping modification as defined
  in \cref{def:ndm}.  We show that $A_\bot$ is an NDA.

  \begin{itemize}
    \sloppy
  \item The set of states is orbit-finite: By definition, we have
    $Q_\bot = \sum_{j = 1}^n \names^{\$n_j}$, which is a finite coproduct of orbit-finite
    nominal sets, whence orbit-finite.

  \item The transition relation $\Delta_\bot$ is equivariant: Suppose that 
    $(j,f) \trans{\gamma} (k,g) \in \Delta_\bot$ and $ \pi \in \permG$.
    \begin{itemize}
    \item $\gamma = a,a\rb$: By the definition of
    $\Delta_\bot$, there exist extensions $\bar{f}$ and $\bar{g}$ of $f$ and $g$,
    respectively, such that
    $(j, \bar{f}) \trans{\gamma} (k, \bar{g}) \in \Delta$.
    Since $\Delta$ is equivariant,
    we have
    $ ( j, \pi \cdot \bar{f}) \trans{\pi \cdot \gamma} (k, \pi \cdot \bar{g}) \in \Delta$.
    Since $\pi \cdot \bar{f}$ extends $\pi \cdot f$ and $\pi \cdot \bar{g}$ extends
    $\pi \cdot g$, by the definition of the name-dropping modification, we have
    $ (j, \pi \cdot f) \trans{\pi \cdot \gamma} (k, \pi \cdot g) \in \Delta_\bot$. Thus,
    $\Delta_\bot$ is equivariant.
    \item $\gamma = \lb a$: By definition of $\Delta_\bot$, we have
    $\supp(f) \cup \{a\} \supseteq \supp(g)$ and there exist $b \in \names$,
    $g' \in Q_\bot$ with $\braket{b} g' = \braket{a} g$ and total injective maps
    $\bar f,\bar g$ extending $f$ and $g'$, respectively, such that
    $(j,\bar f) \trans{\lb b} (k,\bar g) \in \Delta$.
    Since $\Delta$ is equivariant, we obtain
    $
      (j,\pi\cdot \bar f) \trans{\lb \pi(b)} (k,\pi\cdot \bar g) \in \Delta.
    $
    Moreover, $\pi\cdot \bar f$ and $\pi\cdot \bar g$ extend $\pi\cdot f$ and
    $\pi\cdot g'$, respectively, and by equivariance of name abstraction we have
    $
      \braket{\pi(b)} (\pi\cdot g') = \pi\cdot \braket{b} g'
      = \pi\cdot \braket{a} g
      = \braket{\pi(a)} (\pi\cdot g).
    $
    The support condition is preserved under permutations, since
    $
      \supp(\pi\cdot f) \cup \{\pi(a)\}
      = \pi(\supp(f)\cup\{a\}) \supseteq \pi(\supp(g)) = \supp(\pi\cdot g).
    $
    Thus, by the definition of $\Delta_\bot$, we obtain
    $
      (j,\pi\cdot f) \trans{\lb \pi(a)} (k,\pi\cdot g) \in \Delta_\bot,
    $
    i.e.\ $(j,\pi\cdot f) \trans{\pi\cdot \gamma} (k,\pi\cdot g) \in \Delta_\bot$.

  \item $\gamma = \q{a}$: By the definition of $\Delta_\bot$, we have
    $\supp(f) \supseteq \supp(g)$ and there exist $b \in \names$ and
    $f',g' \in Q_\bot$ with
    $\braket{b} f' = \braket{a} f$ and $\braket{b} g' = \braket{a} g$, and total
    injective maps $\bar f,\bar g$ extending $f'$ and $g'$, respectively, such that
    $(j,\bar f) \trans{\q{b}} (k,\bar g) \in \Delta$.
    By equivariance of $\Delta$ we obtain
    $
      (j,\pi\cdot \bar f) \trans{\q{\pi(b)}} (k,\pi\cdot \bar g) \in \Delta,
    $
    where $\pi\cdot \bar f$ and $\pi\cdot \bar g$ extend $\pi\cdot f'$ and
    $\pi\cdot g'$, respectively.  Equivariance of name abstraction yields
    $
      \braket{\pi(b)} (\pi\cdot f')
        = \pi\cdot \braket{b} f'
        = \pi\cdot \braket{a} f
        = \braket{\pi(a)} (\pi\cdot f),
    $
    and similarly
    $
      \braket{\pi(b)} (\pi\cdot g')
        = \pi\cdot \braket{b} g'
        = \pi\cdot \braket{a} g
        = \braket{\pi(a)} (\pi\cdot g).
    $
    The support condition is again preserved since
    $
      \supp(\pi\cdot f) = \pi(\supp(f)) \supseteq \pi(\supp(g)) = \supp(\pi\cdot g).
    $
    Hence, by the definition of $\Delta_\bot$, we obtain
    $
      (j,\pi\cdot f) \trans{\q{\pi(a)}} (k,\pi\cdot g) \in \Delta_\bot,
    $
    i.e.\ $(j,\pi\cdot f) \trans{\pi\cdot \gamma} (k,\pi\cdot g) \in \Delta_\bot$.
    \end{itemize}
  \item The transition relation is left $\alpha$-invariant: Suppose
    $(j,f) \trans{\gamma} (k,g) \in \Delta_\bot$, $\gamma = \lb a$ or $\gamma = \q{a}$, and
    $\braket{a} (k,g) = \braket{b} (k,g')$. By definition, there exist
    extensions $\bar{f}$ and $\bar{g}$ extending $f$ and $g''$ with
    $\braket{c} g'' = \braket{a} g$, respectively, such that
    $(j, \bar{f}) \trans{(ac)\cdot\gamma} (k, \bar{g}) \in \Delta$.  By left $\alpha$-invariance
    of $\Delta$, for $b \in \names$ and $(k, \bar{g}')$ such that
    $ \braket{a} (k, \bar{g}) = \braket{b} (k, \bar{g}')$, we have
    $ (j, \bar{f}) \trans{(ac)\cdot\gamma} (k, \bar{g}') \in \Delta$, where $\bar{g}'$ extends
    $g'$.  In particular, we have some $\pi \in \permG$ such that
    $\pi \cdot \bar{g} = \bar{g}'$ and $\pi \cdot g = g'$.
    Then by construction, $(j,f) \trans{(ab)\cdot\gamma} (k, g') \in \Delta_\bot$. Thus,
    $\Delta_\bot$ satisfies left $\alpha$-invariance.

  \item The transition relation is right $\alpha$-invariant: Suppose
    that $(j,f) \trans{\gamma} (k,g) \in \Delta_\bot$, $\gamma = a\rb$
    or $\gamma = \q{a}$, and
    $\braket{a} (j,f) = \braket{b} (j,f')$. By definition, there exist
    extensions $\bar{f}$ and $\bar{g}$ of $f''$ and $g$ with $\braket{c} f'' = \braket{a} f$, respectively, such that
    $(j, \bar{f}) \trans{\gamma} (k, \bar{g}) \in \Delta$.  By right $\alpha$-invariance
    of $\Delta$, for some $b \in \names$
    and $(j, \bar{f}')$, where $\bar{f}'$ extends $f'$, and such that
    $\braket{a} (j, \bar{f}) = \braket{b} (j, \bar{f}')$, we have
    $(j, \bar{f}') \trans{(ab)\cdot\gamma} (k, \bar{g}) \in \Delta$.
    Then by construction, $(j, f') \trans{(ab)\cdot\gamma} (k, g) \in \Delta_\bot$. Thus,
    $\Delta_\bot$ satisfies right $\alpha$-invariance.
    
  \item The name-dropping modification is finitely branching: Given a state
    $(j,f) \in Q_\bot$, we need to show that it has only finitely many outgoing transition
    $(j,f) \trans{\gamma} (k,g) \in \Delta_\bot$ up to left $\alpha$-invariance. Consider
    the following cases for $\gamma$:
    \begin{itemize}
    \item $\gamma = a, a\rb$: In the NDA $A$, there are only finitely many
      $(j,\bar{f}) \trans{\gamma} (k,\bar{g}) \in \Delta$, where $\bar{f}$ is
      any extension of $f$.  Furthermore, for each $(k,\bar{g}) \in Q$, there
      are only finitely many resulting states $(k,g) \in Q_\bot$ with names dropped from
      the support. Hence, for each $(j,f) \trans{\gamma} (k,g)$, there are only
      finitely many $(j,f) \trans{\gamma} (k,g) \in \Delta_\bot$ with $\bar g$ extending $g$.

    \item $\gamma = \lb a,\q{a}$: Here the number of transitions
      $(j,f) \trans{\gamma} (k,g)$ up to left $\alpha$-invariance is bounded above by
      the number of extension $\bar f$ of $f$ times the number of
      equivalence classes  $\set{ \braket{a} (j,\bar{g})\mid (k, \bar{f})
        \trans{\gamma} (k, \bar{g}) \in \Delta}$ times the number of states $(k,g)$ such that
      $\bar{g}$ extends $g$. All three of those numbers are finite, and so we are done.\qedhere
    \end{itemize}
  \end{itemize}
\end{proofappendix}
\begin{lemma}[label=lem:restrAccFNDN]
  Let $(j,r) \in Q_\bot$ be a state of the name-dropping modification
  $A_\bot = (Q_\bot,\Delta_\bot, i, F_\bot)$ of the NDA $A = (Q,\Delta,i,F)$.  Then $(j,r)$
  accepts every word $w$ that satisfies $\lo(w) \subseteq \supp(r)$ and is accepted by
  some $(j,\bar{r}) \in Q$ such that $\bar{r}$ extends~$r$; in symbols:
  \[
    \set{w\in L_0(j,\bar{r})\mid (j,\bar{r}) \in Q, \text{$\bar{r}$ extends $r$}, \lo(w)
    \subseteq \supp(r)}
    \subseteq
    L_0(j,r).
  \]
\end{lemma}
\begin{proofappendix}{lem:restrAccFNDN}
  We proceed via induction over the length of the word $w$. The base case is trivial. For
  the inductive step, let $S = \supp(j,r)$ and $(j,\bar{r})$ accept some word
  $w = \gamma w'$ such that $\lo(w) \subseteq S$.  We show that also $(j,r)$
  accepts $w$ by distinguishing cases:
  \begin{itemize}\sloppy
  \item $\gamma = a$: We show that $(j,r)$ accepts $w=aw'$. Since $(j,\bar{r})$ accepts $w$
    in $A$, there is a transition $(j,\bar{r})\trans{a}(k,\bar{r}') \in \Delta$ such that
    $(k,\bar{r}')$ accepts $w'$ (in $A$). Due to this fact and by the construction of the
    name-dropping modification, there exists $(k,r')\in Q_\bot$ with
    $\lo(w') \subseteq \supp(r') \subseteq S$ with $\bar{r}'$ extending $r'$.
    Thus, by the induction hypothesis, $(k,r')$ accepts $w'$ (in $A_\bot$). Since we have
    $\supp(r') \subseteq S$, again by the construction of the name-dropping modification, we
    have $(j,r)\trans{a} (k,r') \in \Delta_\bot$, hence $(j,r)$ accepts $w$.

  \item $\gamma = a\rb$: This case is analoguous to the previous case. For
    $w = a\rb w'$, we have $\lo(w) = \lo(w')\cup \set{a}$ and $a \notin \lo(w')$,
    which reflects in the condition on deallocating transitions with
    $a \in \supp(r)$, $\supp(r) \supseteq \supp(r')$ and $a \notin \supp(r')$, as stated in
    the construction of the name-dropping modification (i.e.~in $\Delta_\bot$).
      
  \item $\gamma = \lb a$: The state $(j,\bar{r})$ accepts $w$, hence we have a transition
    $(j,\bar{r})\trans{\lb a}(k,\bar{r}') \in \Delta$ such that $(k,\bar{r}')$ accepts $w'$
    (in $A$).  By the construction of $Q_\bot$, we may choose $(k,r') \in Q_\bot$ with
    $\bar{r}'$ extending $r'$ and $\supp(k,r')=\lo(w')$. By the
    induction hypothesis, $(k,r')$ accepts $w'$ (in $A_\bot$). Furthermore, we have
    that $\lo(w) \cup \set{a} \supseteq \lo(w')$. Hence, using 
    $\lo(w) \subseteq \supp(j,r)$, we obtain
    $\supp(j,r) \cup \set{a} \supseteq \supp(p,r')$. Again by the construction of
    the name-dropping modification, there exists a transition
    $(j,r) \trans{\lb a}(k,r') \in \Delta_\bot$, thus $(j,r)$ accepts
    $w$, as desired.
  
    \item $\gamma = \q{a}$: This case repeats in turn the patterns from 
    the cases of allocating and deallocating transitions.\qedhere
\end{itemize}
\end{proofappendix}

The next step is to show that the additional transitions
introduced by name-dropping is
precisely what is needed to make the literal language closed under $\alpha$-equivalence.

\begin{lemma}[label=lem:NDMClosedAlpha]
  The literal language of a name-dropping modification is closed under $\alpha$-equivalence.
\end{lemma}

\begin{proofsketch}
  Strengthen the claim to all states $(j,r)$ of $A_\bot$ and proceed
  by induction on the length of words.  Given
  $w=\gamma v\aeq \delta v'=w'$ and the first transition
  $(j,r)\trans{\gamma}(k,r_1)$, use the induction hypothesis to obtain
  acceptance of $v$ by acceptance of $v'$ in the successor state.  If
  $\gamma\in\{a,a\rb\}$, then $\delta=\gamma$ and no renaming is
  needed.  If $\gamma=\q{a}$, then $w'=\q{b}v'$ with $b$ fresh for
  $v$; use left $\alpha$-invariance of $\Delta_\bot$ to replace the
  $\q{a}$-transition by a $\q{b}$-transition to the same target.  If
  $\gamma=\lb a$, then $\delta=\lb b$ and $v'=(ab)\cdot v$; first drop
  the renamed name from the successor support using
  \cref{lem:restrAccFNDN}, then transport acceptance along
  equivariance, and finally apply left $\alpha$-invariance to switch
  the allocating letter from $\lb a$ to $\lb b$.
\end{proofsketch}

\begin{proofappendix}{lem:NDMClosedAlpha}
  Let $A_\bot=(Q_\bot,\Delta_\bot,i,F_\bot)$ be the name-dropping modification of
  an NDA $A = (Q,\Delta, i, F)$. We proceed via induction on the length of words and
  strengthen the induction hypothesis to state that the language $L_0(j,r)$ of every state
  $(j,r)\in Q_\bot$ is closed under $\alpha$-equivalence. The base case is trivial. For the
  induction step, let
  $w = \gamma v \aeq \delta v' = w'$ be a pair of $\alpha$-equivalent words, and suppose
  that $w \in L_0(j,r)$. We prove that $w' \in L_0(j,r)$ by distinguishing cases for $\gamma$:
  \begin{itemize}
  \item $\gamma = a,a\rb$: By the definition of $\alpha$-equivalence, we have
    $\delta = \gamma$ and $v \aeq v'$ with $\lo(v) = \lo(v')$. Since $(j,r)$ accepts $w$ in $A_\bot$, there is a
    transition $(j,r) \trans{\gamma} (k,r') \in \Delta_\bot$ with $(j,r')$ accepting $v$ (in
    $A_\bot$). By the induction hypothesis, $(k,r')$ also accepts $v'$, hence $(j,r)$
    accepts $\gamma v' = \delta v' = w'$ in $A_\bot$.

  \item $\gamma = \q{a}$: By the definition of $\alpha$-equivalence,
    we have $w' = \q{b} v' = \delta v'$ for some $b \in \names$ with
    $v \aeq v'$ and $\set{a,b} \cap \lo(v) = \emptyset$. Since $(j,r)$ accepts $w$ in
    $A_\bot$, there exists a transition $(j,r) \trans{\q{a}} (k,r') \in \Delta_\bot$ such that
    $(k,r')$ accepts $v$ (in $A_\bot$). By the induction hypothesis, $(k,r')$ also accepts
    $v'$. By \cref{nameDropWellDefNDA}, $\Delta_\bot$ is left $\alpha$-invariant,
    thus from $(j,r) \trans{\q{a}} (k,r')$ we obtain a transition
    $(j,r) \trans{\q{b}} (k,r'') \in \Delta_\bot$ for some state $(k,r'')$ with
    $\braket{a}(k,r') = \braket{b}(k,r'')$. By \cref{lem:supptrans}, we have
    $a \notin \supp(k,r')$, hence we obtain $(k,r') = (k,r'')$
     and therefore $(j,r)$ accepts $w' = \q{b} v'$ in $A_\bot$.
        
  \item $\gamma = \lb a$: By the definition of $\alpha$-equivalence, we have
    $\delta = \lb b$ and $\braket{a} v = \braket{b} v'$, i.e.
    $v' = (a\, b)\cdot v$, for some $b \in \names$. Since $(j,r)$ accepts $w$ in $A_\bot$,
    there exists a transition $(j,r) \trans{\lb a} (k,r') \in \Delta_\bot$ such that
    $(k,r')$ accepts $v$ (in $A_\bot$). By construction, let $(k,r'') \in Q_\bot$ such that
    $r'$ extends $r''$ and $\supp(r'') = \supp(r')\setminus\set{b}$. Since $(k,r')$ accepts
    $v$, we have that $\lo(v) \subseteq \supp(k,r')$. Furthermore, since
    $\braket{a} v = \braket{b} v'$, we have $b \notin \lo(v)$.
    Hence, $\lo(v) \subseteq \supp(k,r'')$ so that $(k,r'')$ accepts $v$ by
    \cref{lem:restrAccFNDN}.  By equivariance of acceptance, $(a\,b)\cdot (k,r'')$ accepts
    $(a\,b) \cdot v = v'$. Since $(j,r) \trans{\lb a} (k, r')\in \Delta_\bot$, we have a
    transition $(j,r) \trans{\lb a} (k,r'')\in \Delta_\bot$ by construction. By
    $\alpha$-invariance of $\Delta_\bot$, there exists a transition
    $(j,r)\trans{\lb b} (a\,b)\cdot (k,r'')\in \Delta_\bot$. Thus, $(j,r)$ accepts
    $\lb b v' = w'$.\qedhere
  \end{itemize}
\end{proofappendix}

The following lemma provides an alternative characterization of clauses
(3) through (6) of \cref{def:ndm}: According to the definition, each transition
of an NDA yields possibly several transitions in the name-dropping modification for any state
with dropped names. Yet, from a state in the name-dropping modification,
we can also trace back the construction and thence obtain a transition for each extension
of the state in the original NDA.

\begin{lemma}[label=lem:everyExt]
  Let $A = (Q,\Delta, i, F)$ be an NDA and $A_\bot = (Q_\bot,\Delta_\bot,i,F_\bot)$ its name-dropping modification.
  \begin{enumerate}
  \item If $(j,r)\trans{\gamma}(k,r') \in \Delta_\bot$ for $\gamma = a,a\rb,\q{a}$, then for
    every $\bar{r}$ extending $r$, there is some transition
    $(j,\bar{r})\trans{\gamma}(k,\bar{r}')\in \Delta$ with $\bar{r}'$ extending $r'$.
    
  \item If $(j,r)\trans{\lb a}(k,r')\in \Delta_\bot$, then for each $\bar{r}$ extending
    $r$, there is some name $b\in \names$ fresh for~$r'$ and some transition
    $(j,\bar{r})\trans{\lb b}(k,\bar{r}'')\in \Delta$ with $\bar{r}''$ extending
    $(a\, b)\cdot r'$.
  \end{enumerate}
\end{lemma}
\begin{proofappendix}{lem:everyExt}
\begin{enumerate}
\item Let $(j,r) \trans{\gamma} (k,r') \in \Delta_\bot$ for $\gamma = a,a\rb,\q{a}$ and
  $(j,\bar{r}) \in Q$ with $\bar{r}$ extending $r$.  By the construction of the
  name-dropping modification, there is some $\bar{r}''$ that extends $r$ such that
  $(j,\bar{r}'') \trans{\gamma}(k,\bar{r}') \in \Delta$ and $\bar{r}'$ extends $r'$. Let
  $\pi \in \permG$ such that $\pi \cdot \bar{r}'' = \bar{r}$; this $\pi$ exists, since
  $\bar{r}''$ and $\bar{r}$ are both total functions extending $r$ in the same single-orbit
  component $\names^{\#n_j}$ of the nominal state set $Q$. Then by equivariance, from
  $(j, \pi\cdot \bar{r}'') = (j,\bar{r})$ we obtain
  $(j,\bar{r})\trans{\pi\cdot \gamma} (k,\pi\cdot\bar{r}')\in \Delta$. We distinguish cases
  for~$\gamma$. If $\gamma=a$ or $\gamma=a\rb$ and $(j,r)\trans{\gamma}(k,r')\in \Delta_\bot$,
  then $a \in \supp(r)$, hence $\pi$ fixes $a$ and therefore $\pi\cdot \gamma = \gamma$.
  If $\gamma=\q{a}$, then by left $\alpha$-invariance of $\Delta$ for $\q{a}$-transitions,
  from $(j,\bar{r})\trans{\pi\cdot\gamma}(k,\pi\cdot\bar{r}')$ we obtain a transition
  $(j,\bar{r})\trans{\gamma}(k,\tilde{r})\in\Delta$ for some $\tilde{r}$ extending $r'$.
  In all cases, there exists a transition $(j,\bar{r})\trans{\gamma}(k,\tilde{r})\in\Delta$
  with $\tilde{r}$ extending $r'$.

\item Let $(j,r) \trans{\lb a} (j,r')\in \Delta_\bot$ and $(j,\bar{r}) \in Q$ with
  $\bar{r}$ extending $r$. By the construction of the name-dropping modification, we have
  some $(j,\bar{r}') \in Q$ with $\bar{r}'$ extending $r$ and
  $(j,\bar{r}')\trans{\lb c}(k,\bar{r}''')\in \Delta$ with $\bar{r}'''$ extending some
  $r'''\in Q_\bot$ such that $\braket{a} r' = \braket{c} r'''$.  As in the
  previous case, let $\pi \in \permG$ such that $\bar{r} = \pi \cdot \bar{r}'$ and put
  $b = \pi(c)$. By equivariance of $\Delta$ we have
  $\pi\cdot((o,\bar{r}')\trans{\lb c} (p,\bar{r}'''))= (o,\bar{r})\trans{\lb
    b}(p,\pi\cdot\bar{r}''') \in \Delta$. We put $r'' = \pi \cdot r'''$. Since $\bar r'''$
  extends $r'''$, we have that $\bar r'' = \pi\cdot \bar r'''$ extends
  $\pi \cdot r''' = r''$. Moreover, since $\pi \cdot \bar{r}' = \bar{r}$ holds for
  $\bar{r}'$ and $\bar{r}$ both extending $r$, $\pi$ fixes
  $\supp(r)\supseteq\supp(r')\setminus\set{a}$.  From
  $\braket{a} r' = \braket{c} r'''$ we conclude that
  $\braket{b} r'' = \braket{b} (\pi\cdot r''') = \braket{a} r'$.
  This implies that $r'' = (a \, b)\cdot r'$ and that $b$ is fresh for $r'$.
  \qedhere
\end{enumerate}
\end{proofappendix}
\begin{lemma}[label=lem:NDMSameAlpha]
  An NDA accepts the same alphatic language as its name-dropping modification.
\end{lemma}

\begin{proofsketch}
  The inclusion $L_\alpha(A)\subseteq L_\alpha(A_\bot)$ is immediate since $A_\bot$ subsumes
  $A$.
  For the converse, we show that if $A_\bot$ accepts a word $w$, the original automaton $A$
  accepts some $w'\aeq w$.
  Proceed by induction on the length of $w=\gamma v$ for total extensions $\bar r$
  of states $(j,r)\in Q_\bot$ in $A_\bot$.
  Use \cref{lem:everyExt} to lift the first transition of an accepting run in $A_\bot$ to a
  corresponding transition in $A$ and apply the induction hypothesis.
  The allocating case requires changing the binder name, for which we choose a fresh
  name $b$ provided by \cref{lem:everyExt}. 
\end{proofsketch}
\begin{proofappendix}{lem:NDMSameAlpha}
  \sloppy Let $A = (Q,\Delta, i, F)$ be an NDA and $A_\bot = (Q_\bot, \Delta_\bot, i, F_\bot)$ be
  its name-dropping modification. The inclusion $L_\alpha(A)\subseteq L_\alpha(A_\bot)$ is
  trivial, since the NDA $A$ is a subautomaton of $A_\bot$.

  We proceed to show that $L_\alpha(A_\bot)\subseteq L_\alpha(A)$: we prove that for each
  word $w$ accepted by $A_\bot$ there is an $\alpha$-equivalent word $w'\aeq w$ accepted by
  $A$. We proceed via induction on the length of words. The base case is trivial (since
  final states remain final when names are dropped). For the induction step, let
  $w=\gamma v$ be accepted by some $(j,r)\in Q_\bot$. We strengthen the induction hypothesis
  to that for some state $(j,r)\in Q_\bot$ accepting $w$, each $(j,\bar{r})\in Q$ with $\bar{r}$ extending $r$
  accepts some $w' \aeq w$. We distinguish cases for $\gamma$:
  \begin{itemize}
  \item $\gamma = a,a\rb$: Let $(j,\bar{r})\in Q$ with $\bar{r}$ extending $r$. We
    show that $(j,\bar{r})$ accepts some $w'\aeq w$.  We have a transition
    $(j,r)\trans{\gamma}(k,r') \in \Delta_\bot$ with $(k,r')$ accepting $v$ (in $A_\bot$).
    By \cref{lem:everyExt}, we also have a transition $(j,\bar{r})\trans{\gamma}(k,\bar{r}') \in \Delta$ with
    $\bar{r}'$ extending~$r'$. By the induction hypothesis,
    $(k,\bar{r}')$ accepts some $v' \aeq v$. Therefore, $(j,\bar{r})$ accepts
    $w' = \gamma v' \aeq \gamma v = w$.

  \item $\gamma = \q{a}$: Let $(j,\bar{r})\in Q$ with $\bar{r}$ extending $r$. Since
    $(j,r)$ accepts $w = \q{a} v$ in $A_\bot$, there is a transition
    $(j,r)\trans{\q{a}}(k,r') \in \Delta_\bot$ with $(k,r')$ accepting $v$ (in $A_\bot$).
    By \cref{lem:everyExt}, we obtain a transition
    $(j,\bar{r})\trans{\q{a}}(k,\bar{r}') \in \Delta$ with $\bar{r}'$ extending $r'$.
    By the induction hypothesis, $(k,\bar{r}')$ accepts some $v'\aeq v$,
    hence we eventually have $w' = \q{a} v' \aeq \q{a} v = w$, and $(j,\bar{r})$
    accepts $w'$.

  \item $\gamma = \lb a$: Let $(j,\bar{r})\in Q$ with $\bar{r}$ extending $r$, we show
    that $(j,\bar{r})$ accepts some $w'\aeq w$.  Since $(j,r)$ accepts $w$, there is a
    transition $(j,r)\trans{\lb a}(k,r') \in \Delta_\bot$ with $(k,r')$
    accepting~$v$. By \cref{lem:everyExt}, we have a name $b \in \names$ and a transition
    $(j,\bar{r})\trans{\lb b}(k,\bar{r}'') \in \Delta$ with~$\bar{r}''$ extending
    $(a\, b)\cdot r'$. Choose some $\bar{r}'$ extending $r'$ such that
    $\braket{a}\bar{r}' = \braket{b} \bar{r}''$.  By the induction hypothesis,
    $(k,\bar{r}')$ accepts some $v''\aeq v$. We have that
    $(a\, b)\cdot \bar{r}' = \bar{r}''$.  By equivariance of acceptance, we have that
    $\bar{r}''$ accepts $(a\, b)\cdot v''$. Consequently, $(j,\bar{r})$ accepts the word $w'
    = \lb b (a\, b)\cdot v''$. Finally, since
    $\braket{a}\bar{r}' = \braket{b}\bar{r}''$ and $v''\in L(\bar{r}')$ imply
    that $b\notin\lo(v'')$ and
    $a\notin\lo((ab)\cdot v'')$, we have that
    $\braket{a} v'' = \braket{b}((ab) v'')$. Since $v\aeq v''$ holds, we also
    have
    $\braket{a}\lbrack v\rbrack_\alpha = \braket{b} \lbrack(a\ b)\cdot
    v''\rbrack_\alpha$, thus eventually
    $w = \lb a v \aeq \lb b(a\, b)\cdot v'' = w'$.\qedhere
  \end{itemize}
\end{proofappendix}
\begin{theorem}\label{th:nameDropModAlphaClosed}
  The name-dropping modification closes the language of an NDA under $\alpha$-equivalence.
\end{theorem}
\begin{proof}
  The literal language of the name-dropping modification~$A_\bot$ is closed under
  $\alpha$-equivalence by \cref{lem:NDMClosedAlpha}, and its alphatic language coincides
  with that of $A$ by \cref{lem:NDMSameAlpha}.  Hence, $A_\bot$ accepts exactly the
  $\alpha$-renamings of words that $A$ accepts literally.
\end{proof}
\section{A Kleene Theorem for NDA}\label{sec:kleene}


We present a Kleene theorem establishing the precise correspondence between NDA and regular
deallocation expressions over~$\ah$, i.e.~an algebraic description for languages expressible
by NDA. For this purpose, we are going to introduce a representation of an NDA as a
\emph{D-NFA}, a (classical) non-deterministic finite automaton (NFA) with an additional
constraint on the transitions. Our results then allow translating between NDA, regular
deallocation expressions and D-NFA as shown in \cref{fig:kleene}.

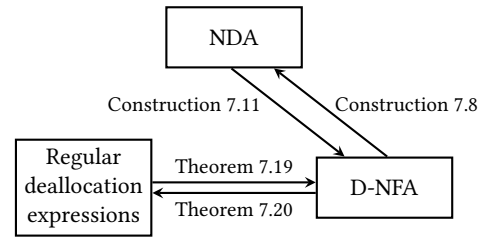
\begin{figure}[h]
\centering
\begin{tikzpicture}[
    auto,
    box/.style={rectangle, draw=black, thick, minimum width=1.8cm, minimum height=0.8cm, align=center},
    arrow/.style={->,>=stealth,thick},
    lemma/.style={font=\small}
]
\begin{scope}[x = 4cm, y = 2cm]
\node[box] (regdex) {Regular\\deallocation\\expressions};
\node[box] (dnfa) at (1,0) {D-NFA};
\node[box] (nda) at (0.5,1) {NDA};

\draw[arrow] ([yshift=2pt]regdex.east) -- node[above, lemma] {\cref{th:regdexToDnfa}} ([yshift=2pt]dnfa.west);
\draw[arrow] ([yshift=-2pt]dnfa.west) -- node[below, lemma] {\cref{th:dnfaToRegdex}} ([yshift=-2pt]regdex.east);
\draw[arrow] ([xshift=1pt]dnfa.north) -- node[right, lemma, outer sep = 3pt, pos = 0.6] {\cref{con:nom}} (nda);
\draw[arrow] ([xshift=-1pt]nda.south) -- node[left, lemma, outer sep = 3pt, pos = 0.4] {\cref{con:sRes}} (dnfa);
\end{scope}
\end{tikzpicture}
\caption{Correspondence between regular deallocation expressions, D-NFA, and NDA.
Together, this yields a Kleene-theorem for NDA.}
\label{fig:kleene}
\end{figure}

We recall the definition of regular expressions: A (classical) regular expression $r$ over
$\ah$ is generated by the following grammar
\[
  r ::= \emptyset \mid \epsilon \mid \gamma \mid r\cdot r \mid r + r \mid r^*,
  \qquad\text{where $\gamma \in \ah$}.
\]
The \emph{literal language} $L(r)\subseteq \ahs$ is defined in the standard way:
\begin{itemize}
  \item $L(\emptyset)=\emptyset$, $L(\epsilon)=\{\epsilon\}$, and $L(\gamma)=\{\gamma\}$ for letters $\gamma$;
  \item $L(r+r')=L(r)\cup L(r')$, $L(r\cdot r')=L(r)\cdot L(r')$, and $L(r^*)=L(r)^*$.
\end{itemize}

Given a regular expression $r$, the set $\rc(r)$ of all \emph{right-closed} names of~$r$
consists of all those names $a$ such that there is a word $w\in L(r)$ in which $a$ is
right-closed.

Similarly, the set $\lo(r)$ of all \emph{left-open} names of $r$
consists of all those names $a$ such that there is a word $w \in L(r)$
in which~$a$ is left-open.  Finally, the set $\lc(r)$ of all
\emph{left-closed} names of $r$ consists of all those names that are
left-closed in \emph{every} word of $L(r)$. In symbols, we have
\begin{align}
  \rc(r) &= \textstyle\bigcup_{w \in L(r)} \rc(w),\quad 
  \lo(r) = \bigcup_{w \in L(r)} \lo(w), \label{eq:rc-lo} \\
  \lc(r) &= \textstyle\bigcap_{w \in L(r)} \lc(w). \nonumber
\end{align}
  
%
%
%

We are interested in regular expressions that generate only right non-shadowing words, and
we shall prove in \cref{lem:RegNonShadow} that they have the following syntactic
characterization.

\begin{defn}[label=def:regdex]
  A \emph{regular deallocation expression} is a regular expression $r$ over $\ah$ which is
  generated by the following grammar
  \[
    r,r_1,r_2,r_3 := \emptyset\mid\epsilon\mid\lb a\mid a\rb\mid a
    \mid {?} \mid r_1\cdot r_2 \mid r + r \mid (r_3)^*,
  \]
  where $\rc(r_1)\cap \lo(r_2) = \emptyset$ and
  $\rc(r_3) \cap \lo(r_3) = \emptyset$. 
\end{defn}
Additionally, we write $L_\alpha(r)=\{[w]_\alpha\mid w\in L(r)\}$ for the corresponding alphatic
language.

\begin{lemma}[label=lem:gramDec]
  It is decidable whether a given regular expression over~$\ah$ is a
  regular deallocation expression as per \cref{def:regdex}.
\end{lemma}

\begin{proofsketch}
  We give recursive clauses for $\rc(r)$ (recursive clauses for
  $\lo(r)$ and $\lc(r)$ are in the appendix):
  \begin{align*}
    \rc(\emptyset) &= \rc(\epsilon) = \rc(\lb a) = \rc(a) = \rc(?) = \emptyset \\
    \rc(a\rb) & = \set{a} \\
    \rc(r_1\cdot r_2) &= \rc(r_2)\cup(\rc(r_1)\setminus \lc(r_2) ) \\
    \rc(r_1+r_2) &= \rc(r_1)\cup\rc(r_2) \\
    \rc(r_3^*) &= \rc(r_3)
  \end{align*}
  Since $r$ contains only finitely many name occurrences, all intermediate sets are finite
  (or equal to $\names$ as for the base case $\lc(\emptyset)=\names$).
  The side conditions in the grammar of \cref{def:regdex} are then decidable by checking the
  finite intersections $\rc(r_1)\cap\lo(r_2)$ and $\rc(r_3)\cap\lo(r_3)$ at each occurrence of
  concatenation and star.
\end{proofsketch}

\begin{proofappendix}{lem:gramDec}
  Towards a decision procedure we first we give recursive definitions of $\rc(r)$, $\lo(r)$
  and $\lc(r)$:
  \begin{align*}
    \rc(\emptyset) &= \rc(\epsilon) = \rc(\lb a) = \rc(a) = \rc(?) = \emptyset \\
    \rc(a\rb) & = \set{a} \\
    \rc(r_1\cdot r_2) &= \rc(r_2)\cup(\rc(r_1)\setminus \lc(r_2) ) \\
    \rc(r_1+r_2) &= \rc(r_1)\cup\rc(r_2) \\
    \rc(r_3^*) &= \rc(r_3)
    \\[5pt]
    \lo(\emptyset) &= \lo(\epsilon) = \lo(\lb a) = \lo(?) = \emptyset \\
    \lo(a\rb) &= \lo(a) = \set{a} \\
    \lo(r_1\cdot r_2) &= \lo(r_1)\cup(\lo(r_2)\setminus \lc(r_1)) \\
    \lo(r_1+r_2) &= \lo(r_1)\cup\lo(r_2) \\
    \lo(r_3^*) &= \lo(r_3)
    \\[5pt]
    \lc(\emptyset) &= \names \\
    \lc(\epsilon) &= \lc(a\rb) = \lc(a) = \lc(?) = \emptyset \\
    \lc(\lb a) &= \set{a} \\
    \lc(r_1\cdot r_2) &= \lc(r_1)\cup(\lc(r_2)\setminus \lo(r_1) ) \\
    \lc(r_1+r_2) &= \lc(r_1)\cap\lc(r_2) \\
    \lc(r_3^*) &= \emptyset \\
  \end{align*}
  We claim that for every regular expression $r$ the sets obtained
  from this recursive definitions coincide with the previously defined
  ones.
  
  We prove the three equalities simultaneously by structural induction on the regular
  expression $r$. The base cases are: 
  \begin{itemize}
    \item $r = \emptyset$: Then $L(r) = \emptyset$, so $\bigcup_{w \in L(r)} \rc(w) = \emptyset = \rc(\emptyset)$ 
    by the definition of $\rc$. Similarly, $\bigcup_{w \in L(r)} \lo(w) = \emptyset = \lo(\emptyset)$ and $\bigcap_{w \in L(r)} 
    \lc(w) = \names = \lc(\emptyset)$.
    \item $r = \epsilon$: Then $L(r) = \{\epsilon\}$, hence $\bigcup_{w \in L(r)} \rc(w) = \rc(\epsilon) = \emptyset 
    = \rc(\epsilon)$ by definition. This holds true similarly for $\lo$ and $\lc$.
    \item $r = \gamma$ where $\gamma \in \ah$: Then $L(r) = \{\gamma\}$, so $\bigcup_{w \in L(r)} \rc(w) = 
    \rc(\gamma)$. Again, by the definition in \cref{lem:gramDec},
    $\rc(\gamma) = \emptyset$ if $\gamma = a$ or $\gamma =
    \lb a$, and $\rc(\gamma) = \{a\}$ if $\gamma = a\rb$. This matches $\rc(\gamma)$ directly.
    The same reasoning applies to $\lo$ and $\lc$.
  \end{itemize}
  
  For the inductive step, we showcase the case distinction for $\rc$;
  the same arguments hold true for $\lo$ and $\lc$ analogously.
  \begin{itemize}
  \item $r = r_1 + r_2$: By definition, $L(r) = L(r_1) \cup L(r_2)$. Therefore, we have
    \begin{align*}
      \textstyle \bigcup_{w \in L(r)} \rc(w)
      &= \textstyle\big(\bigcup_{w \in L(r_1)} \rc(w)\big)
      \cup \big(\bigcup_{w \in L(r_2)} \rc(w)\big)\\
      &= \rc(r_1) \cup \rc(r_2) = \rc(r),
    \end{align*}
    where the second equality follows from the induction hypothesis.
    
  \item $r = r_1 \cdot r_2$: By definition,
    $L(r) = L(r_1) \cdot L(r_2) = \{uv \mid u \in L(r_1), v \in L(r_2)\}$.  Using
    \cref{lem:cnConcat} and the induction hypothesis:
    \begin{align*}
      \textstyle\bigcup_{w \in L(r)} \rc(w)
      &= \textstyle\bigcup_{u \in L(r_1)} \bigcup_{v \in L(r_2)} \rc(uv) \\
      &= \textstyle\bigcup_{u \in L(r_1)} \bigcup_{v \in L(r_2)} \left((\rc(u) \setminus \lc(v)) \cup \rc(v)\right) \\
      &= \textstyle\big(\bigcup_{u \in L(r_1)} \rc(u) \setminus \bigcap_{v \in L(r_2)}
      \lc(v)\big) \mathrel{\cup} \\ 
      &\phantom{=\ \ }\textstyle\bigcup_{v \in L(r_2)} \rc(v) \\
      &= (\rc(r_1) \setminus \lc(r_2)) \cup \rc(r_2) = \rc(r).
    \end{align*}
    \item $r = r_1^*$: By definition, we have
    \[L(r) = L(r_1)^* = \{\epsilon\} \cup L(r_1) \cup L(r_1)L(r_1) \cup \cdots\text{.}\]
    Then we compute:
    \begin{align*}
      &\textstyle\bigcup_{w \in L(r)} \rc(w)\\ 
      &\textstyle= \rc(\epsilon) \cup \bigcup_{w \in L(r_1)} \rc(w) \cup \bigcup_{w_1,w_2 \in L(r_1)} \rc(w_1w_2) \cup \cdots\\
      &\textstyle\subseteq \rc(\epsilon) \cup \bigcup_{w \in L(r_1)} \rc(w) \mathrel{\cup} \\
      & \phantom{\subseteq\ \ }\textstyle\bigcup_{w_1,w_2 \in L(r_1)}
      (\rc(w_1)\cup\rc(w_2)) \cup \cdots\\
      &= \textstyle\bigcup_{w \in L(r_1)} \rc(w) = \rc(r_1) = \rc(r)
    \end{align*}
    The inclusion $\rc(r)= \rc(r_1) \subseteq \bigcup_{w \in L(r)} \rc(w)$ holds by induction since $L(r_1) \subseteq L(r)$.
  \end{itemize}

  Since regular expressions are finite and only finitely many
  names occur in a given $r$, the two sets $\rc(r)$ and $\lo(r)$ are finite, and the set
  $\lc(r)$ is either finite or equal to $\names$. It follows that the three sets can easily
  be computed recursively to check the side conditions for concatenation and Kleene star in
  the grammar in \cref{def:regdex}. Thus, we obtain a decision procedure for that grammar.
\end{proofappendix}

\begin{lemma}\label{lem:RegNonShadow}
  Let~$r$ be a classical regular expression~$r$ over $\ah$. Then the
  following are equivalent.
  \begin{enumerate}
  \item\label{item:r-rns} All words in the language over~$\ah$
    defined by~$r$ are right non-shadowing
  \item\label{item:r-regdex} The expression~$r$ is a regular deallocation expression
    in the sense of \cref{def:regdex}.
  \end{enumerate}
\end{lemma}

\begin{proofsketch}
  We proceed by structural induction on $r$.
  The base cases ($\emptyset$,$\epsilon$, and $\gamma\in\ah$) are immediate.
  Sums are handled by $L(r_1+r_2)=L(r_1)\cup L(r_2)$.
  For concatenation we employ \cref{lem:cnConcat} and a contradiction argument;
  the star case is analogous.
\end{proofsketch}

  \begin{proofappendix}{lem:RegNonShadow}
  We proceed via structural induction on the regular expression $r$.

  For the base cases $r \in \{\emptyset,\epsilon,a,\lb a,a\rb,\q{a}\}$ the claim is
  immediate: the grammar in \cref{def:regdex} admits all these expressions as
  regular deallocation expressions, and every word in $L(r)$ has length at most
  one, hence is right-non-shadowing.

  For sums $r = r_1 + r_2$ we have $L(r) = L(r_1) \cup L(r_2)$ and the grammar
  imposes no side conditions. Thus $r$ is a regular deallocation expression
  iff $r_1$ and $r_2$ are, which holds immediately by the induction hypothesis.

  For a concatenation $r = r_1 \cdot r_2$, we prove both implications.

  First assume that $r$ is a regular deallocation expression. Then~$r_1$ and~$r_2$ are
  regular deallocation expressions and $\rc(r_1) \cap \lo(r_2) = \emptyset$ by
  \cref{def:regdex}. By the induction hypothesis, for $i=1,2$, we have
  $L(r_i) \subseteq \RNS(\ah)$.  Let $w = w_1 w_2 \in L(r)$ with $w_i \in L(r_i)$. By the
  definitions of $\rc(r_1)$ and $\lo(r_2)$ in \cref{eq:rc-lo}, we have
  $\rc(w_1) \cap \lo(w_2) = \emptyset$. So by \cref{lem:rnsConcat}, $w$ is
  right-non-shadowing.  Therefore, $L(r) \subseteq \RNS(\ah)$ holds.

  Conversely, assume that $L(r) \subseteq \RNS(\ah)$. Suppose that there is some
  $a \in \rc(r_1) \cap \lo(r_2)$. By \cref{eq:rc-lo}, we have words $w_1 \in L(r_1)$ and
  $w_2 \in L(r_2)$ with $a \in \rc(w_1) \cap \lo(w_2)$. Then for the decomposition
  $w_1 w_2 = w_1 \cdot w_2$, we obtain $\rc(w_1) \cap \lo(w_2) \neq \emptyset$, so $w_1 w_2$
  is not right-non-shadowing, contradicting $L(r) \subseteq \RNS(\ah)$. Thus, we have
  $\rc(r_1) \cap \lo(r_2) = \emptyset$. By the induction hypothesis, $r_1$ and $r_2$ are
  regular deallocation expressions, whence so is $r$.

  The case $r = (r_3)^*$ works analogously: Using the induction hypothesis for
  $r_3$ and the above argument for a finite concatenation
  $w = w_1 \cdots w_n$ with each $w_i \in L(r_3)$ shows that the condition
  $\rc(r_3) \cap \lo(r_3) = \emptyset$ in \cref{def:regdex} holds iff
  $L(r_3^*) \subseteq \RNS(\ah)$.\qedhere
\end{proofappendix}

\begin{defn}[label=def:dnfa]
  Let $A = (Q,\Delta, i, F)$ be a classical NFA over the alphabet $\ah$.  The sets $\rc_A(q)$
  and $\lo_A(q)$ of right-closed and left-open names, respectively, of a state $q \in Q$ are
  inductively defined as the least family of sets such that
  \begin{itemize}
  \item if $q \trans{a\rb} q'$, then $\rc_A(q) \cup \set{a} \subseteq \rc_A(q')$
    and $\lo_A(q) \supseteq \set{a} \cup \lo_A(q')$,
    
  \item if $q \trans{a} q'$, then $\lo_A(q) \supseteq \lo_A(q') \cup \set{a}$ and
    $\rc_A(q)\setminus\set{a} \subseteq \rc_A(q')$,
    
  \item if $q \trans{\q{a}} q'$, then $\rc_A(q)\cup\set{a} \subseteq \rc_A(q')$,
    $\lo_A(q) \supseteq \lo_A(q')$, and
    
  \item if $q \trans{\lb a} q'$, then $\lo_A(q) \supseteq \lo_A(q') \setminus \set{a}$ and
    $\rc_A(q)\setminus \set{a} \subseteq \rc_A(q')$.
  \end{itemize}
  We ommit the index $A$ if it is clear from context. We say that $A$ is a \emph{deallocation NFA} (\emph{D-NFA}, for short) if for every state
  $q \in Q$, we have $\rc_A(q) \cap \lo_A(q) = \emptyset$. The \emph{literal language}~$L(A)$
  of~$A$ is its language as an NFA. The \emph{alphatic language}~$\la(A)$
  of~$A$ is the quotient set $\la(A) = \set{ [w]_\alpha \mid w \in L(A)}$.
  The \emph{data language}~$L_D(A)$ of~$A$ is the set $L_D(A) = \db[L(A)]$.
\end{defn}
\begin{rem}\label{rem:dnfa}
  This definition is intuitively understood as follows.  The sets
  $\rc(q)$ and $\lo(q)$ contain the names that must be absent or
  present, respectively, in the support of~$q$ in order for the given
  NFA to be realizable as an NDA. Clearly, the two sets have to be
  disjoint. The rules defining these sets reflect the properties of the support 
  stated in \cref{lem:supptrans}. At the same
  time, $\rc(q)$ contains the names that are right-closed in some word
  that has a run ending in~$q$, while $\lo(q)$ contains the names that
  are left-open in some word having a run beginning in~$q$, so that
  the disjointness property ensures that words having a run from any
  state in the automaton are right non-shadowing, a property
  guaranteed also by NDAs.
\end{rem}

For ease of the upcoming construction of a D-NFA from a regular
deallocation expression in \cref{th:regdexToDnfa}, we introduce
$\epsilon$-transitions into D-NFAs, building on the notion of
classical NFAs with $\epsilon$-transitions.  However, these do not
increase the expressive power of D-NFAs, as we show in the following.

\begin{defn}\label{def:dnfaEps}
  Let $A = (Q,\Delta,i,F)$ be an NFA with $\epsilon$-transitions
  ($(q,\epsilon,q')\in\Delta$). The sets $\rc(q)$ and $\lo(q)$ are defined similarly as
  before, but in addition to the inclusions given in \cref{def:dnfa}, we require that for
  every $(q\trans{\epsilon}q')\in\Delta$, we have $\lo(q) \supseteq \lo(q')$ and
  $\rc(q) \subseteq \rc(q')$. We say that $A$ is a \emph{D-NFA with $\epsilon$-transitions}
  if $\rc(q) \cap \lo(q) = \emptyset$ for all $q\in Q$.
\end{defn}

\begin{lemma}\label{lem:epsDnfa}
  For D-NFAs, $\epsilon$-transitions are admissible, i.e.~for each
  D-NFA with $\epsilon$-transitions $A=(Q,\Delta, i,F)$,
  there is a D-NFA $A'$ (without $\epsilon$-transitions) such that $L(A') = L(A)$.
\end{lemma}

\begin{proofsketch}
  Eliminate $\epsilon$-transitions on the underlying NFA via $\epsilon$-closures $E(q)$ and
  the standard construction of an equivalent $\epsilon$-free NFA.
  Define $\rc(E(q)):=\rc(q)$ and $\lo(E(q)):=\lo(q)$.
  The additional constraints for $\epsilon$-transitions on $\rc$ and $\lo$
  ensure that every visible transition in the quotient construction still
  satisfies the defining inclusions of \cref{def:dnfa}.
  Hence disjointness $\rc\cap\lo=\emptyset$ is preserved and the result is again a D-NFA.
\end{proofsketch}

\begin{proofappendix}{lem:epsDnfa}
  Let $A=(Q,\Delta,i,F)$ be a D-NFA with $\epsilon$-transitions such that for
  every $q\trans{\epsilon}q'\in\Delta$ we have $\lo(q) \supseteq \lo(q')$ and
  $\rc(q) \subseteq \rc(q')$.

  We first remove $\epsilon$-transitions on the underlying NFA by the standard
  $\epsilon$-elimination for NFA. Define the $\epsilon$-closure of a state
  \[
    E(q) = \{q'\in Q \mid q \trans{\epsilon^*} q'\},
  \]
  where $\trans{\epsilon^*}$ are arbitrarily many consecutive $\epsilon$-transitions.
  We put $A'=(Q',\Delta',E(i),F')$ with
  \begin{align*}
    &Q' = \set{E(q) \mid q \in Q},\\
    &E(q)\trans{\gamma}E(q')\in\Delta'\quad  \text{if 
      $\exists p \in E(q).\, p\trans{\gamma}q'\in\Delta$ and $\gamma\neq\epsilon$},\\
    &F' = \{E(q) \in Q' \mid E(q)\cap F\neq\emptyset\}.
  \end{align*}
  Then $A'$ is an NFA without $\epsilon$-transitions such that $L(A) = L(A')$.

  It remains to show that $A'$ is a D-NFA, i.e.~that the conditions on
  $\rc$ and $\lo$ from \cref{def:dnfa} are preserved. We claim that
  $\rc(E(q)) = \rc(q)$ and $\lo(E(q)) = \lo(q)$ satisfy the clauses of
  \cref{def:dnfa} when read over the transition relation $\Delta'$.

  Take a transition $E(q)\trans{\gamma}E(q')\in\Delta'$.  By construction, there is
  a state $p\in E(q)$ such that
  \[
    q \trans{\epsilon^*} p \trans{\gamma} q'
  \]
  is a path in $A$ with the same visible letter $\gamma$.  For each
  $\epsilon$-transition $r\trans{\epsilon}r'$ on this path we have
  $\lo(r) \supseteq \lo(r')$ and $\rc(r) \subseteq \rc(r')$ by assumption.  At
  step $p\trans{\gamma}q'$ the
  condition from \cref{def:dnfa} holds for $A$.  Thus, we easily
  obtain the required clause between $\rc(E(q)),\lo(E(q))$ and $\rc(E(q')),\lo(E(q'))$ for
  the transition $E(q)\trans{\gamma}E(q')$. For instance, we show the
  case for $\gamma = a\rb$, all other cases work analogously:
  \resizebox{\columnwidth}{!}{
    $\rc(E(q))\cup\set{a} = \rc(q)\cup\set{a} \subseteq \rc(p)\cup\{a\}
      \subseteq \rc(q') = \rc(E(q'))$}
  and
    \begin{align*}
      \lo(E(q)) = \lo(q) \supseteq \lo(p) \supseteq \{a\}\cup \lo(q') = \{a\}\cup \lo(E(q')).
    \end{align*}
  Hence, setting $\rc(E(q)) = \rc(q)$ and $\lo(E(q)) = \lo(q)$ for all $q\in Q$,
  the $\rc$ and $\lo$ of $A$ satisfy all the defining
  inequalities of \cref{def:dnfa} for the transition relation $\Delta'$ as
  well. In particular, we still
  have $\rc(q) \cap \lo(q) = \emptyset$ for all $q\in Q$,
  which implies $\rc(E(q)) \cap \lo(E(q)) = \emptyset$ for all $E(q)\in Q'$,
  so $A'$ is a D-NFA.
\end{proofappendix}

We proceed to give a translation between NDA
and D-NFA that preserves the alphatic language.

\begin{construction}[Nominalization]\label{con:nom} Let
  $A = (Q,\Delta,i,F)$ be a D-NFA. We fix a section
  $\iota\colon Q\to Q\times\names^*$ of the left projection
  $Q\times\names^*\to Q$ such that for $q\in Q$,  $\iota(q)=(q,(a_1,\dots,a_n))$
  where $\lo(q)=\{a_1,\dots,a_n\}$ (in other words, we fix an enumeration
  of the elements of $\lo(q)$ for each~$q$). Then the
  \emph{nominalization} of $A$ is the NDA
  $A' = (\cleq Q', \cleq \Delta', i', \cleq F')$ given by
  \begin{equation*}
    Q'  = \iota[Q]\subseteq Q\times\names^*,\qquad
    F'  = \iota[F],\qquad
    i' = \iota(i),
  \end{equation*}
  and $\Delta'$ consists of the following transitions:
  \begin{itemize}
  \item a transition $\iota(q) \trans{\alpha} \iota(q')$ whenever
    $(q \trans{\alpha} q') \in \Delta$ and~$\alpha\in\ah$ has the
    shape $\alpha=a$, $\alpha=\q{a}$, or $\alpha=a\rb$; and
  \item a transition $\iota(q)\trans{\lb a} (q',s)$ whenever
    $(q \trans{\lb c} q') \in \Delta$ and
    $\braket{a}(q',s) = \braket{c}\iota(q')$.
  \end{itemize}
  Here, equivariant closures~$\cleq$ and abstractions are understood
  w.r.t.~the product $Q\times\names^*$ of~$Q$ viewed as a discrete
  nominal set and the nominal set~$\names^*$.
\end{construction}

\begin{lemma}[label=lem:nomNda]
The nominalization of a D-NFA is an NDA.
\end{lemma}

\begin{proofsketch}
  Since $Q$ and $\Delta$ are finite and only finitely many names occur
  on transitions, there are only finitely many possible supports
  $\lo(q)$ and hence only finitely many tuples chosen by the section
  $\iota$; thus $Q'$ is finite and its equivariant closure is
  orbit-finite.  By construction, $\cleq\Delta'$ is equivariant.  Left
  $\alpha$-invariance holds because the construction of $\Delta'$
  mirrors precisely the requirements in \cref{def:nda}.  Name erasure
  follows from the clauses for $\rc$ and the required disjointness property
  for D-NFAs. Finite
  branching is inherited from finiteness of $\Delta$ and the fact that
  $Q'$ is finite.
\end{proofsketch}

\begin{proofappendix}{lem:nomNda}
  Let $A = (Q,\Delta,i,F)$ be a D-NFA and let
  $A' = (\cleq Q', \cleq \Delta', i', \cleq F')$ be
  its nominalization
  constructed as in \cref{con:nom}.
  We show that $A'$ is an NDA.
  Since $A=(Q,\Delta,i,F)$ is an NFA, the sets $Q$ and $\Delta$
  are finite, and only
  finitely many names occur in the labels of transitions in
  $\Delta$.
  Let $N \subseteq \names$ be this finite set of names. By the
  definition of $\lo(q)$ and $\rc(q)$ in \cref{def:dnfa},
  we have $\lo(q) \subseteq N$ for every $q\in Q$,
  and for each $q\in Q$ there are only finitely many tuples
  $r=(a_1,\dots,a_n)$
  such that $\lo(q)=\{a_1,\dots,a_n\}$, and hence $Q'$ is finite.
  Consequently, $\cleq Q'$ is an orbit-finite nominal set.

  By construction, we have $\Delta'\subseteq \cleq Q'\times \ah
  \times \cleq Q'$,
  hence $\cleq \Delta'$ is an equivariant subset of
  $(\cleq Q') \times \ah \times (\cleq Q')$.
  The conditions on the transition relation in \cref{def:nda}
  are equivariant,
  therefore it suffices to verify them for
  transitions in $\Delta'$,
  which is straightforward by construction:
  \begin{itemize}
    \item Left $\alpha$-invariance: Let
      $(q,r) \trans{\gamma} (q',s)$, $\gamma = \lb a$ or
      $\gamma = \q{a}$, be a transition in $\Delta'$ and suppose
      we have $(q',u) \in \cleq Q'$ with $\braket{a}s =
      \braket{b}u$.
      We have to show that there is $(q,r) \trans{(ab)\cdot\gamma}
      (q',u)\in\Delta'$.
      By definition of $\Delta'$, there are $c$ and $t$ such that
      $q \trans{(ac)\cdot\gamma} q' \in \Delta$, $(q,r),(q',t)
      \in Q'$ and
      $\braket{a}s = \braket{c}t$.
      Then we clearly also have
      $\braket{b}u = \braket{c}t$, hence we have
      $(q,r) \trans{(ab)\cdot\gamma} (q',u)\in\Delta'$
      by construction.

    \item Name erasure: Let $(q,r)\trans{\gamma}(q',s)$ be a transition
      in $\Delta'$ with $\gamma=a\rb$ or $\gamma=\q{a}$.  By
      \cref{con:nom}, we have $q\trans{\gamma}q'$ in $\Delta$ and
      $(q,r)=\iota(q)$, as well as $(q',s)=\iota(q')$.  In both cases,
      \cref{def:dnfa} yields $\rc(q)\cup\set{a}\subseteq\rc(q')$, hence
      $a\in\rc(q')$.  Since $A$ is a D-NFA,
      $\rc(q')\cap\lo(q')=\emptyset$, so $a\notin\lo(q')$.  By
      definition of $\iota$, the tuple $s$ enumerates $\lo(q')$, so $a$
      does not occur in $s$ and therefore $a\fresh(q',s)$.  For
      transitions in $\cleq\Delta'$, the claim follows from
      equivariance of freshness.

  \item Finite branching: Fix $(q,r)\in Q'$. For free and deallocating
    transitions, the number of outgoing transitions is given by the
    number of free and deallocating transitions of $q$ in $\Delta$,
    respectively, which is finite. This holds true similarly for
    allocating transitions: From each transition
    $q \trans{\lb c}q' \in \Delta$, we construct one equivalence class
    $\braket{c}(q',s)$ such that $(q,r) \trans{\lb c} (q',s)$ in
    $\Delta'$. Similarly, a transition $(q,r) \trans{\q{a}} (q',s)$ in
    $\Delta'$ is obtained from a transition
    $q \trans{\q{c}} q' \in \Delta$.  While there are infinitely many
    such transitions $q \trans{\q{c}} q'\in \cleq \Delta'$ (one for
    each fresh $c$), the number of states they lead to
    $\set{(q',s)\mid ((q,r),\q{a},(q',s))\in\Delta'}$ is finite, since
    $Q'$ and $\Delta$ are.
  \end{itemize}
  Finally, the initial state of $A'$ is the unique $i'=(i,r)\in Q'$.
  The set of final states of $A'$ is
  $F' = \{(f,r)\mid f\in F\}$, which is equivariant,
  making $A'$ an NDA.
\end{proofappendix}

\begin{lemma}[label=lem:nomSameAlpha]
For a D-NFA $A$ and its nominalization $A'$, we have
$\la(A') = \la(A)$.
\end{lemma}

\begin{proofsketch}
  One inclusion is immediate because every transition of $A$ induces a corresponding
  transition between the chosen representatives $\iota(q)$, so every accepting run of $A$
  lifts to an accepting run of $A'$ on the same word.
  For the reverse inclusion, we project a run of $A'$ back to its first component to obtain
  a run of $A$.
  The only mismatch can come from allocating and unknown transitions,
  where the nominalization may
  choose a different fresh name.
  Use equivariance of acceptance in NDAs (\cref{lem:acceptEquiv}) together with the defining
  rules of $\alpha$-equivalence to turn the word of the run in $A'$ into an
  $\alpha$-equivalent word accepted by $A$.
\end{proofsketch}

\begin{proofappendix}{lem:nomSameAlpha}
  Let $A = (Q,\Delta,i,F)$ be a D-NFA and let
  $A' = (\cleq Q', \cleq \Delta', i', \cleq F')$ be its
  nominalization
  with $Q'\subseteq Q\times\names^*$, $\Delta'\subseteq
  \cleq Q'\times\ah\times \cleq Q'$,
  $i' \in Q'$ and $F' \subseteq Q'$
  given as in \cref{con:nom}.
  We prove $\la(A') = \la(A)$ by establishing both inclusions.

  "$\supseteq$": For every accepting run of $A$, there is an
  accepting run of $A'$ on the same word.
  This is straightforward since $A$ is essentially a
  subautomaton of $A'$: 
  Let
  \[
    \sigma = i \trans{\gamma_1} q_1 \trans{\gamma_2} \cdots
    \trans{\gamma_n} f
  \]
  be an accepting run of $A$ on $w = \gamma_1\cdots\gamma_n$.
  Then,
  \[
    \rho = \iota(i) \trans{\gamma_1} \iota(q_1) \trans{\gamma_2}
    \cdots
    \trans{\gamma_n} \iota(f)
  \]
  is an accepting run of $A'$ on $w$ by construction.
  
  "$\subseteq$": We show more generally that for each state
  $(q,r) \in Q'$ accepting a word~$w$, the state $q \in Q$ accepts
  some $w'\aeq w$. We proceed
  by induction on~$w$. 

  The base case is simple: If $(f,r)\in Q'$ accepts $w = \epsilon$,
  then we have $(f,r) \in F'$, hence $f\in F$, so~$f$
  accepts~$\epsilon$.

  Inductive step: Let $(q,r) \in Q'$ accept $w = \gamma w'$. We
  distinguish cases on~$\gamma$:
  \begin{itemize}
    \item $\gamma = a$:
      Then there is $(q,r) \trans{a} (q',s) \in \Delta'$ with
      $(q',s)$ accepting $w'$.
      By \cref{con:nom}, this transition arises from a free
      transition
      $q \trans{a} q' \in \Delta$.  By the induction hypothesis,
      $q'$ accepts some
      $v \aeq w'$, hence $q$ accepts $av$, and by \cref{def:Alpha},
      we have
      $av \aeq aw' = w$.

    \item $\gamma = \lb a$: We have a transition
      $(q,r) \trans{\lb a} (q',s) \in \Delta'$ with $(q',s)$ accepting
      $w'$.  By \cref{con:nom}, there is some allocating transition
      $q \trans{\lb c} q' \in \Delta$ with $\braket{a}s = \braket{c}t$
      for $(q',t)\in Q'$, and by \cref{lem:acceptEquiv}, $(q',t)$
      accepts $w''$ with $w'' = (ac)\cdot w'$.  By the induction
      hypothesis, $q'$ accepts a word $v \aeq w''$, so $q$ accepts
      $\lb c v$. By the definition of $\alpha$-equivalence
      (\cref{def:Alpha}), we obtain
      $\lb c v \aeq \lb c w'' \aeq \lb a w' = w$.

    \item $\gamma = \q{a}$: Then
      $(q,r) \trans{\q{a}} (q',s) \in \Delta'$ with $(q',s)$ accepting
      $w'$.  By construction of $\Delta'$ in \cref{con:nom}, the
      transition $(q,r) \trans{\q{a}} (q',s) \in \Delta'$ arises by
      applying some permutation to a transition
      $\iota(q)\trans{\q{c}}\iota(q')$ in~$A$.  By the induction
      hypothesis, $q'$ accepts some $v \aeq w'$, so~$q$ accepts
      $\q{c}v$; it remains to show that $\q{c}v\aeq\q{a}w'$.  We have
      $c\in\rc(q')$ and hence $c \notin \lo(q')$, since~$A$ is a
      D-NFA. Moreover, we have already shown that~$A'$ is an NDA
      (\Cref{lem:nomNda}), so $a\notin\supp(s)$, and hence
      $a\notin\supp([w']_\alpha)=\lo(w')$ by the support lemma
      (\Cref{lem:supptrans}) and \Cref{fact:supp-lo}.
      Moreover, $c\notin\lo(w')$ since $c\in\rc(q')$ and $w'$ is read from $(q',s)$ in a
      right-non-shadowing run, so $\{a,c\}\cap\lo(w')=\emptyset$.
      Consequently $q$
      accepts $\q{c}v \aeq \q{a}w' = w$ by the definition of
      $\alpha$-equivalence.

    \item $\gamma = a \rb$:
      Finally, suppose $(q,r)$ accepts $w = a\rb w'.$
      Then there is $(q,r) \trans{a\rb} (q',s) \in \Delta'$
      with $(q',s)$ accepting
      $w'$.  By \cref{con:nom}, this step comes from a
      deallocating transition
      $q \trans{a\rb} q' \in \Delta$.  By the induction hypothesis,
      $q'$ accepts
      some $v \aeq w'$. As above, since $a \notin \supp(s)$,
      we have that
      $a \notin \lo(w')$.  Hence, by \cref{def:Alpha},
      we have $a\rb v \aeq a\rb w' = w$, and $q$ accepts $a\rb v$.
  \end{itemize}
  Finally, since for each state $(q,r)\in Q'$ accepting $w$, the state
  $q\in Q$ accepts some $w'\aeq w$, this holds in particular for $i'$ and $i$,
  yielding $\la(A') = \la(A)$.
\end{proofappendix}

\begin{construction}[$S$-restriction]\label{con:sRes}
    Let $A = (Q, \Delta, i,F)$ be an NDA.
    We fix a set $S \subseteq \names$ such that $\supp(i) \subseteq S$ and 
    that $|S| = \dg(A)+1$.
    The
    \emph{$S$-restriction} $A_S$ of $A$ is defined by 
    $A_S = (Q_S, \Delta_S, i,F_S)$, where
    \begin{align*}
      Q_S &= \{q\in Q\mid \supp(q) \subseteq S\}, \\
      \Delta_S &= \{q\trans{\gamma}q'\in \Delta \mid q,q' \in Q_S, \db(\gamma) \in S\}, \\
      F_S &= F  \cap Q_S.
    \end{align*}
  \end{construction}

\begin{lemma}\label{lem:sResDnfa}
    The $S$-restriction $A_S$ of an NDA $A$ is a D-NFA.
\end{lemma}

\begin{proofappendix}{lem:sResDnfa}
  Let $A = (Q,\Delta,i,F)$ be an NDA and let
  $A_S = (Q_S,\Delta_S,i,F_S)$ be its $S$-restriction.  To show that
  $A_S$ is a D-NFA, we need to show that for each state~$q$ in~$A_S$,
  the sets~$\lo(q)$ and $\rc(q)$ are disjoint. To this end, we show
  that $\lo(q) \subseteq \supp(q)$ and
  $\rc(q) \subseteq \names \setminus \supp(q)$. This follows once we
  show that the families $\lo'$, $\rc'$ of sets defined by
  $\lo'(q)=\supp(q)$, $\rc'(q)= \names \setminus \supp(q)$ satisfy the
  defining conditions of $\lo$, $\rc$ as per \cref{def:dnfa}. This,
  however, is immediate from the support lemma (\cref{lem:supptrans}).
\end{proofappendix}
\noindent The restriction of an NDA accepts the correspondingly
restricted literal language:

\begin{lemma}\label{lem:restrictNDAL0}
    The $S$-restriction $A_S$ of an NDA $A$ accepts
    $L_0(A)$ restricted to words over $\hat{S}$, that is, $L(A_S) = L_0(A)\cap \shs$.
\end{lemma}

\begin{proofsketch}
  We show this via induction on the length of accepting runs, and use that
  by \cref{lem:supptrans}, each run on a word $w\in\shs$
  remains in $A_S$. 
\end{proofsketch}

\begin{proofappendix}{lem:restrictNDAL0}
    Let $A_S$ restrict $A$ to $S$, we show that $L(A_S) =
    L_0(A)\cap \shs$. Since $A_S$ is a proper 
    subautomaton to $A$, one direction of the inclusion is trivial.
    So let $w \in \shs$ such that $A$
    accepts $w$. We reinforce the claim and show that each $q$
    in $A_S$ accepts all $w \in \shs$
    accepted by $q$ in $A$, and proceed via induction on the
    length of accepting runs.
    
    For the base case, let $q\in Q_S$ accept $w = \epsilon$ in $A$,
    then $q\in F$, hence $q\in F_S$, and $q$ accepts $w$ in
    $A_S$. For the inductive step from~$w$ to $\gamma w\in\shs$, let
    $q\in Q_S$, and let $\sigma$ be an accepting run of~$A$ on
    $\gamma w$ starting in~$q$. Then~$\sigma$ has the shape
    $\sigma = q \trans{\gamma} q' \trans{\sigma'} f$ in $A$. Since
    $\gamma \in \hat{S}$, we have $q \trans{\gamma} q' \in \Delta_S$,
    and then are done by induction, as soon as we show $q'\in Q_S$.
    But this is immediate from the support lemma
    (\cref{lem:supptrans}).
\end{proofappendix}
\noindent As an immediate consequence of \Cref{lem:restrictNDAL0}, we
obtain a similar property for data languages:
\begin{lemma}\label{cor:restrictNDALD}
  Restricting an NDA $A$ to $S$ restricts $L_D(A)$ to words over $S$, i.e.~$L_D(A_S) = S^* \cap L_D(A)
  = \set{\db(w)\mid w \in L_0(A_S)}$.
\end{lemma}
\noindent On the other hand, crucially, the alphatic language does not
change under restriction:

\begin{proposition}\label{prop:sResAlpha}
  The $S$-restriction $A_S$ of an NDA $A$ accepts the same alphatic language:
  $\la(A_S) = \la(A)$.
\end{proposition}

\begin{proofappendix}{prop:sResAlpha}
  The left-to-right inclusion is trivial, since~$A_S$ is a
  subautomaton of~$A$.  For the reverse inclusion, let $w \in L(A)$,
  and choose $w' \aeq w$ such that $w' \in \shs$. Since
  $\deg(A) \geq |\bigcup_{v \in L_0(q)}\lo(v)|$ and we have picked
  $|S| = \deg(A)+1$, such a~$w'$ exists. By \cref{lem:restrictNDAL0},
  $w'$ is accepted by $A_S$.
\end{proofappendix}

As motivated in \cref{rem:dnfa}, the following two lemmas characterize
the sets $\rc(q)$ and $\lo(q)$ for a state $q$ in a D-NFA
more precisely.

\begin{lemma}\label{lem:cnDnfaCoinc}
  Let $A = (Q,\Delta,i,F)$ be an NFA over $\ah$ and $q,q'\in Q$ a pair
  of states.  Then for the NFA $A' = (Q,\Delta,q,\set{q'})$ we have
  that $\bigcup_{w\in L(A')} \rc(w) \subseteq \rc(q')$ and
  $\bigcup_{w \in L(A')} \lo(w) \subseteq \lo(q)$.
\end{lemma}

\begin{proofappendix}{lem:cnDnfaCoinc}
  \takeout{Ex ante, we observe that the conditions of a D-NFA in \cref{def:dnfa}
  only constrain
  the transition relation~$\Delta$, hence changing the initial and final states
  does not affect whether an NFA is a D-NFA. Hence $A'$ is again a D-NFA.}
  We show the inclusions
  \[
    \bigcup_{w\in L(A')} \rc(w) \subseteq \rc(q')
    \,\text{and}\,
    \bigcup_{w\in L(A')} \lo(w) \subseteq \lo(q)
  \]
  simultaneously by induction on the length of an accepting run in~$A'$.
  Let
  \[
    \rho = q_0 \trans{\gamma_1} q_1 \trans{\gamma_2} \cdots
    \trans{\gamma_n} q_n
  \]
  be an accepting run of $A'$ on $w = \gamma_1\cdots\gamma_n$
  with $q_0 = q$ and $q_n = q'$.

  For $n=0$, we have $w = \epsilon$ and thus $q=q'$.  By
  \cref{tab:namedefs}, $\rc(\epsilon) = \lo(\epsilon) = \emptyset$,
  so the claim holds.

  For the induction step, let $w = \gamma u$ and
  \[
    \rho = q \trans{\gamma} p \trans{*} q'
  \]
  for some state $p$.  The suffix $p \trans{*} q'$ is an accepting run of
  length $n-1$ on the word $u$.  By the induction hypothesis,
  for $A'' = {Q,\Delta, p, \set{q'}}$, we obtain
  \[
    \rc(u) \subseteq \rc(q')\qquad\text{and}\qquad
    \lo(u) \subseteq \lo(p).
  \]
  We now distinguish cases on $\gamma$:

  \begin{itemize}
  \item $\gamma = a$:
    From \cref{tab:namedefs}, we have $\rc(w) = \rc(u)$ and
    $\lo(w) = \lo(u)\cup\{a\}$.  Since $q\trans{a}p\in \Delta$,
    \cref{def:dnfa} yields $\rc(q) \subseteq \rc(p)$ and
    $\lo(q) \supseteq \lo(p)\cup\{a\}$.  Hence
    $\rc(w) \subseteq \rc(q')$ by the induction hypothesis and
    $\lo(w) \subseteq \lo(q)$.

  \item $\gamma = \lb a$:
    Here, \cref{tab:namedefs} gives $\rc(w) = \rc(u)$ and
    $\lo(w) = \lo(u)\setminus\{a\}$.  From $q\trans{\lb a}p$ we obtain
    $\rc(q)\setminus\{a\} \subseteq \rc(p)$ and
    $\lo(q) \supseteq \lo(p)\setminus\{a\}$, so again
    $\rc(w) \subseteq \rc(q')$ and $\lo(w) \subseteq \lo(q)$.

  \item $\gamma = a\rb$:
    By \cref{tab:namedefs}, $\lo(w) = \lo(u)\cup\{a\}$ and
    $a\in\rc(w)$ iff $a\in\rc(u)$ or $a\notin\lc(u)$.  Since
    $q\trans{a\rb}p$, \cref{def:dnfa} implies
    $\rc(q)\cup\{a\} \subseteq \rc(p)$ and
    $\lo(q) \supseteq \lo(p)\cup\{a\}$.  Using the induction
    hypothesis and the relationships between $\rc$ and $\lc$ for
    words from \cref{lem:cnConcat}, we obtain
    $\rc(w) \subseteq \rc(q')$ and $\lo(w) \subseteq \lo(q)$.

  \item $\gamma = ?$:
    Finally, if $q\trans{?}p$, then \cref{def:dnfa} gives
    $\rc(q) \subseteq \rc(p)$ and $\lo(q) \supseteq \lo(p)$.
    From \cref{tab:namedefs}, we have $\rc(w) = \rc(u)$ and
    $\lo(w) = \lo(u)$, so the claim follows from the induction
    hypothesis. \qedhere
  \end{itemize}
\end{proofappendix}

Clearly, the reverse inclusions do not hold, since the values~$\rc(q)$ and $\lo(q)$ of a
state $q$ are properties of the transition structure and hence independent of the choice of
initial and final state.  For instance, the simple D-NFA
\[
  A = (\set{q,q'}, \set{q\trans{a}q'},q,\set{q})
\]
accepts the language
$L(A) = \set{\epsilon}$ with $\lo(\epsilon) = \emptyset$ while we have $\lo(q) = \set{a}$.
However, if we consider all possible incoming and outgoing runs, not exclusively those
ending in a particular final state or starting in the initial state, respectively, the
reverse inclusions do hold true:

\begin{lemma}\label{lem:cnDnfaCoincInv}
  Let $A = (Q,\Delta, i, F)$ be an NFA over $\ah$, and let $q \in Q$ be a state.
  \takeout{ 
  Let further
  \[
    O(q) = \set{w\mid 
      \text{there is a run for $w$ in $A$ starting in $q$}}
  \]
  be the set of words on which there is an outgoing, not necessarily
  accepting run in $A$ from $q$. Likewise, let
  \[
    I(q) = \set{w\mid
      \text{there is a run for $w$ in $A$ ending in $q$}}
    \]
    be the set of words on which there is an incoming run in $A$ from some state of $Q$.
  }
  Further, denote by $O(q)$ the set of all words having a run in $A$ that starts in $q$, and
  by $I(q)$ the set of all words having a run in $A$ that ends in $q$.
  Then we have 
  \[
    \textstyle
    \bigcup_{w\in I(q)} \rc(w) = \rc(q)
    \quad \text{and} \quad
    \bigcup_{w \in O(q)} \lo(w) = \lo(q).
  \]
\end{lemma}

\begin{proofappendix}{lem:cnDnfaCoincInv}
  Let $A = (Q,\Delta,i,F)$ be an NFA over $\ah$, $q \in Q$, and
  $O(q)$ and $I(q)$ as defined above.
  The inclusions
  \[
    \textstyle
    \bigcup_{w\in I(q)} \rc(w) \subseteq \rc(q)
    \qquad \text{and} \qquad
    \bigcup_{w \in O(q)} \lo(w) \subseteq \lo(q)
  \]
  hold by \cref{lem:cnDnfaCoinc}, so it remains to show
  \[
    \textstyle
    \bigcup_{w\in I(q)} \rc(w) \supseteq \rc(q)
    \quad \text{and} \quad
    \bigcup_{w \in O(q)} \lo(w) \supseteq \lo(q).
  \]

  We consider $\rc$. The argument for $\lo$ is dual and hence omitted.
  Let $R(q) = \bigcup_{w \in I(q)} \rc(w)$.
  We claim that $R(q)$ satisfies the clauses
  for $\rc(q)$ in \cref{def:dnfa}. Since $\rc(q)$ is defined as the
  least set fulfilling the requirements in \cref{def:dnfa},
  this implies $\rc(q) \subseteq R(q)$, and hence $\rc(q) = R(q)$.

  Let $q \trans{\gamma} q' \in \Delta$ be a transition in $A$.
  We proceed via case distinction on $\gamma$:
  \begin{itemize}
    \item $\gamma = a\rb$:
      We have to show $R(q) \cup \{a\} \subseteq R(q')$.
      Let $b \in R(q)$. Then there is a run $p \trans{*} q$ on some
      word $w$ with $b \in \rc(w)$. This yields a run
      $p \trans{*} q \trans{a\rb} q'$ on $wa\rb$, so $wa\rb \in I(q')$.
      By \cref{lem:cnConcat}, we have
      $\rc(wa\rb) = \rc(w) \cup \set{a}$, hence $b \in \rc(wa\rb)$ and thus
      $b \in R(q')$, and also $a \in R(q')$.
    \item $\gamma = a$:
      Here, we have to show $R(q)\setminus \set{a} \subseteq R(q')$.
      Let $b \in R(q)$, $a\neq b$, be witnessed by a run $p \trans{*} q$ on $w$
      with $b \in \rc(w)$. Then the run $p \trans{*} q \trans{a} q'$ on $wa$
      with $b \in \rc(wa)$ yields $b \in R(q')$.
    \item $\gamma = \lb a$:
      We need to show $R(q) \setminus \{a\} \subseteq R(q')$.
      Let $b \in R(q)$, $a \neq b$, so there is a run
      $p \trans{*} q$ on $w$ with $b \in \rc(w)$.
      This yields a run
      $p \trans{*} q \trans{\lb a} q'$ on $w \lb a$,
      so $w \lb a \in I(q')$.
      By \cref{tab:namedefs} and \cref{lem:cnConcat}, we have
      $\rc(w \lb a) = \rc(w) \setminus \{a\}$, hence
      $b \in \rc(w \lb a)$ and $b \in R(q')$.
    \item $\gamma = \q{a}$:
      We have to show $R(q) \cup \{a\} \subseteq R(q')$.
      Let $b \in R(q)$. Then there is a run $p \trans{*} q$ on some
      word $w$ with $b \in \rc(w)$. This yields a run
      $p \trans{*} q \trans{\q{a}} q'$ on $w\q{a}$, so $w\q{a} \in I(q')$.
      By \cref{lem:cnConcat}, we have
      $\rc(w\q{a}) = (\rc(w)\setminus (\lc(\q{a})\cup \lo(\q{a}))) \cup \rc(\q{a})$.
      Since $\lc(\q{a}) = \{a\}$, $\lo(\q{a}) = \emptyset$, and $\rc(\q{a}) = \{a\}$,
      we get $\rc(w\q{a}) = (\rc(w)\setminus \{a\}) \cup \{a\} \supseteq \rc(w)$,
      hence $b \in \rc(w\q{a})$ and thus $b \in R(q')$,
      as well as $a \in \rc(w\q{a})$ and $a \in R(q')$.\qedhere
  \end{itemize}  
\end{proofappendix}

\begin{rem}
  Note that the set of all left-closed names shared by all words $w$
  having a run from or to some state $q$ is empty: since the empty word $\epsilon$
  has a run starting and ending in every state, we have
  $\bigcap_{w\in I(q)} \lc(w) = \emptyset = \bigcap_{w\in O(q)} \lc(w)$.
\end{rem}

We are now ready to show the main results of this section: the mutual
translations between regular deallocation expressions and D-NFAs. 

\begin{theorem}\label{th:regdexToDnfa}
    For every regular deallocation expression $r$,
    there exists a D-NFA $A$ such that $L_\alpha(r) = L_\alpha(A)$.
\end{theorem}

\begin{proofsketch}
  We construct $A$ by structural induction on the grammar in \cref{def:regdex}.
  The base cases are straightforward. For the inductive cases, we use the standard Kleene constructions for NFA
  (concatenation, star, and union) with $\epsilon$-transitions and verify
  the conditions imposed on $\rc$ and $\lo$ in \cref{def:dnfa}.
\end{proofsketch}

\begin{proofappendix}{th:regdexToDnfa}
  We show that for every regular deallocation expression~$r$, there is
  a D-NFA with $\epsilon$-transitions $A = (Q,\Delta, i, \set{f})$
  with a single final state $f$ such that $L(r) = L(A)$.  From
  \cref{lem:epsDnfa}, we then obtain a D-NFA $A'$ without
  $\epsilon$-transitions such that $L(A') = L(A) = L(r)$.

  We proceed via structural induction on the derivation of $r$ in the grammar from \cref{def:regdex}.
  
  \noindent Base cases:
  \begin{itemize}
  \item $r=\emptyset$: $A=(\{q,f\},\emptyset,q,\set{f})$.
  \item $r=\epsilon$: $A=(\{q\},\emptyset,q,\{q\})$.
  \item $r=a$: $A=(\{q,f\},\{q\trans{a} f\},q,\{f\})$.
  \item $r=\lb a$: $A=(\{q,f\},\{q\trans{\lb a} f\},q,\{f\})$.
  \item $r=a\rb$: $A=(\{q,f\},\{q\trans{a\rb} f\},q,\{f\})$.
  \item $r=\q{a}$: $A=(\{q,f\},\{q\trans{\q{a}} f\},q,\{f\})$.
  \end{itemize}
  We check in each case that~$A$ is a D-NFA. In the cases $r=\emptyset$ and $r=\epsilon$,
  there are no transitions, hence $\rc(q)=\lo(q)=\emptyset
  =\rc(f)=\lo(f)$.
  For $r=a$, we have $\lo(f)=\rc(f)=\emptyset$ and $\lo(q)=\{a\}$ and
  $\rc(q)=\emptyset$.
  For $r=\lb a$, we have $\rc(q)=\lo(q)=\rc(f)=\lo(f)=\emptyset$.
  For $r=a\rb$, we have $\lo(f)=\emptyset$ and $\rc(f)=\{a\}$ and
  $\lo(q)=\{a\}$ and $\rc(q)=\emptyset$. For $r = \q{a}$, we have $\lo(q) =
  \emptyset = \lo(f) = \rc(q)$ and $\rc(f) = \set{a}$.
  In all cases, we have $\rc(q)\cap\lo(q)=\emptyset$
  for each state $q$ in $A$.

  \noindent Inductive step: We distinguish cases:
  \begin{itemize}
  \item $r\cdot r'$: By the induction hypothesis, let $A=(Q,\Delta,i,\set{f})$
    and $A'=(Q',\Delta',i',\set{f'})$ be D-NFAs with $\epsilon$-transitions
    such that $L(r)=L(A)$ and $L(r')=L(A')$.
    Assume w.l.o.g. that $Q$ and $Q'$ are disjoint.
    We build the NFA $A''=(Q\cup Q',\Delta'',i,\set{f'})$ where
    \[\Delta''=\Delta\cup\Delta'\cup\{f\trans{\epsilon}i'\}.\]
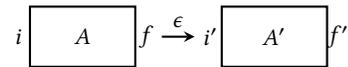
    \begin{center}
    \begin{tikzpicture}[
    node distance=0.85cm,
    auto,
    box/.style={rectangle, draw=black, thick, minimum width=1.4cm, minimum height=0.8cm, align=center},
    arrow/.style={->,>=stealth,thick},
    lemma/.style={font=\small\itshape}
    ]
    \node (i) {$i$};
    \node[box, right of=i] (A) {$A$};
    \node[right of=A] (f) {$f$};
    \node[right of=f] (i') {$i'$};
    \node[box, right of=i'] (A') {$A'$};
    \node[right of=A'] (f') {$f'$};
    \draw[arrow] (f) -- node[above] {$\epsilon$} (i');
    \end{tikzpicture}
    \end{center}
    Clearly, we have $L(A'') = L(r\cdot r')$ by the standard Kleene theorem for NFA.
    We establish the following upper bounds for the sets~$\rc(q)$ and $\lo(q)$
    for a state $q$ in $A''$ depending on which automaton it comes from. For $\rc_{A''}$, we claim:
    \[\rc_{A''}(q) \subseteq \begin{cases}
      \rc_A(q) & \text{if } q\in Q\\
                                \rc_{A'}(q)\cup(\rc_A(f)\setminus\lo_{A'}(q)) & \text{if } q\in Q'
                              \end{cases}\]
                            and for $\lo_{A''}$:
                            \[\lo_{A''}(q) \subseteq \begin{cases}
                                                        \lo_{A}(q) \cup (\lo_{A'}(i')\setminus\rc_{A}(q)) & \text{if } q\in Q\\
                                                        \lo_{A'}(q) & \text{if } q\in Q'
                                                      \end{cases}\]
    Let $A''' = (Q,\Delta,q,\set{f})$ be the D-NFA $A$ with $q$ as initial state.
       By \cref{lem:cnDnfaCoincInv}, we have $\lo_{A''}(q) = \bigcup_{w \in O(q)}\lo(w)$.
    Recall from \cref{lem:cnDnfaCoincInv} that $O_A(q)$ denotes the set of words admitting a run
    in~$A$ starting in~$q$.
    By construction, we have for each $w \in O_{A''}(q)$ that either $w \in O_A(q)$ or
    $w = vu$ where on $v$ there is a run~$\sigma_v$ from~$q$ to~$f$ in~$A$ (i.e.~$v \in L(A''')$) , and  on~$u$ 
    we have a run $\sigma_u$ from~$i'$ to some~$p\in Q'$ in $A'$. Using this observation, we calculate:
    \begin{align*}
    &\phantom{=} \lo_{A''}(q) = \bigcup_{w \in O(q)}\lo(w)\\
    &= \bigcup_{w \in O_A(q)}\lo(w) \cup \bigcup_{u \in L(A''')}\bigcup_{v \in O_{A'}(i')}\lo(uv)\\
    &= \lo_{A}(q) \cup \bigcup_{u \in L(A''')}\bigcup_{v \in O_{A'}(i')}\lo(u)\cup (\lo(v)\setminus\lc(u))\\
    &= \lo_{A}(q) \cup \bigcup_{u \in L(A''')}\bigcup_{v \in O_{A'}(i')}(\lo(v)\setminus\lc(u))\\
    &= \lo_{A}(q) \cup \lo_{A'}(i') \setminus(\bigcap_{u \in L(A''')}\lc(u))\\
    \end{align*}
    Here, we observe the following:
    If $a \in \lo(i') \subseteq \lo(f)$, then we have $a \in \lo(w)$ for some word $w$
    having a run starting in $i'$. Furthermore, if $a\notin (\bigcap_{u \in L(A''')}\lc(u))$
    then there is a word $u$ with a run from $q$ to $i'$ for which $a \notin \lc(u)$.
    Now, we cannot have $a \in \rc(q)$, since otherwise there would be a word $v$
    with a run ending in $q$ such that $a \in \rc(v)$, yielding a run on a
    right shadowing word $vuw$. Thus, we conclude
    $\lo(i') \setminus(\bigcap_{u \in L(A''')}\lc(u)) \subseteq
    \lo(i') \setminus(\bigcup_{u \in I(q)}\rc(u)) = \lo(i') \setminus \rc(q)$,
    and hence we have 
    \[\lo_{A''}(q) \subseteq \lo(q) \cup (\lo(i') \setminus \rc(q)).\]
    The reasoning for $\rc$ works entirely analogously.
  \item $r + r'$: By the induction hypothesis, we have
    $A = (Q,\Delta,i,\set{f})$ and $A' = (Q',\Delta',i',\set{f'})$
    such that $\la(r) = \la(A)$ and $\la(r') = \la(A')$.  Then,
    $A'' = (Q\cup
    Q'\cup\set{i_\mathsf{new},f_\mathsf{new}},\Delta\cup\Delta'\cup
    \set{i_\mathsf{new}\trans{\epsilon}i,
      i_\mathsf{new}\trans{\epsilon}i', f
      \trans{\epsilon}f_\mathsf{new}, f'
      \trans{\epsilon}f_0},i_0,\set{f_0})$,
    where $i_0,f_0\notin Q\cup Q'$ are new states,
    is a D-NFA with $\epsilon$-transitions such that
    $\la(r+r') = \la(A'')$.  We have
    $\lo(i_\mathsf{new}) \subseteq \lo(i) \cup \lo(i')$,
    $\rc(i_\mathsf{new}) = \emptyset$,
    $\lo(f_\mathsf{new}) = \emptyset$ and
    $\rc(f_\mathsf{new}) \subseteq \rc(f) \cup \rc(f')$, since the
    right-hand sides satisfy all inequalities from \cref{def:dnfa};
    moreover, the sets of right-closed names and left-open names of
    the original states remain unchanged. This implies the required
    disjointness, so $~A''$ is a D-NFA (with
    $\epsilon$-transitions). Moreover, $A''$ clearly accepts precisely
    the language $L(r+r')$.
  \item $r^*$: By the induction hypothesis, we have
    $A = (Q,\Delta,i,\set{f})$ such that $\la(r) = \la(A)$.
    Then, $A'' = (Q\cup\set{q_\mathsf{new}},\Delta\cup\set{q_\mathsf{new}\trans{\epsilon}i,
    f \trans{\epsilon}q_\mathsf{new}},i,\set{f})$
    is a D-NFA with $\epsilon$-transitions
    such that $\la(r^*) = \la(A'')$.
    Analogously to the case for concatenation, one shows that
    $\rc(q) \subseteq \rc(q) \cup (\rc(f)\setminus\lo(q))$ and
    $\lo(q) \subseteq \lo(i)\setminus\rc(q))$ for all $q \in Q$,
    as well as $\lo(q_\mathsf{new}) = \lo(i)$ and $\rc(q_\mathsf{new}) = \rc(f)$.\qedhere
  \end{itemize}
\end{proofappendix}

\begin{theorem}\label{th:dnfaToRegdex}
  For every D-NFA $A$, there exists a regular deallocation expression
  $r$ such that $L(r) = L(A)$.
\end{theorem}
\begin{proofsketch}
  We obtain a regular expression for $A$ by induction on the subsets of its transition relation:
  starting from the empty subset of transitions,
  at each step we add one missing
  transition $\delta = q\trans{\gamma}q'$, extend the subautomaton,
  and update the expression to
  $r'' = r' + r'_{(i,\set{q})} (\gamma\, r'_{(q',\set{q})})^{*}\gamma r'_{(q',F)}$,
  summarizing all paths that iterate $\delta$.
  Here, $r_{(x,Y)}$ denotes the regular expression obtained
  by the induction hypothesis
  from the previous step for initial state $x$ and final states $Y$.
  Iterating until all transitions are added yields the desired~$r$.
\end{proofsketch}

\begin{proofappendix}{th:dnfaToRegdex}
  \sloppy Let $A = (Q,\Delta, i, F)$ be a D-NFA. We proceed by
  induction on the (finite) cardinality of~$\Delta$, exploiting that
  the class of D-NFA is closed under subautomata and under
  changing~$i$ or~$F$.

  Induction base ($\Gamma = \emptyset$): In this case, $A$ accepts
  either the empty language, expressed by $r = \emptyset$, or the
  language consisting only of the empty word, expressed by
  $r = \epsilon$.

  Induction step: Pick $\delta = (q \trans{\gamma} q') \in \Delta$,
  and put $\Gamma=\Delta\setminus\{\delta\}$. As remarked above, we
  have D-NFAs $A' = (Q,\Gamma, i,F)$
  $A'_{i,q} = (Q,\Gamma, i,\set{q})$,
  $A'_{q',q} = (Q,\Gamma, q',\set{q})$ and
  $A'_{q',F} = (Q,\Gamma, q',F)$, for which by induction we have
  corresponding regular expressions $r'$, $r'_{i,q}$, $r'_{q',q}$ and
  $r'_{q',F}$ such that $L(r')=L(A')$,
  $L(r'_{i,q}) = L(A'_{i,q})$, $L(r'_{q',q}) = L(A'_{q',q})$
  and $L(r'_{q',F}) = L(A'_{q',F})$.  Now put
  \begin{equation*}
    r = r' + r'_{i,q} (\gamma r'_{q',q})^*\gamma r'_{q',F}.
  \end{equation*}
  This is a regular deallocation expression: Since $A$ is a D-NFA, we
  have $\rc(q) \cap \lo(q) = \emptyset$ for each $q \in Q$.  By the
  induction hypothesis and \cref{lem:cnDnfaCoinc}, this implies the
  corresponding disjointness conditions for the concatenations and
  Kleene-star in $r$ via a simple case distinction on $\gamma$.

  We verify $L(r) = L(A)$:
  
  \begin{enumerate}
  \item $L(r) \subseteq L(A)$: Let $w \in L(r)$. By the structure of $r$,
    either $w \in L(r')$ or $w$ is of the form
    $w = w_{(i,q)} (\gamma w_{(q',q,i)})_{i\leq n} \gamma w_{(q',F)}$, where
    $w_{(i,q)} \in L(r'_{i,q}) = L(A'_{i,q})$,
    $w_{(q',q,i)} \in L(r'_{q',q}) = L(A'_{q',q})$ for each $i \leq n$,
    and
    $w_{(q',F)} \in L(r'_{q',F}) = L(A'_{q',F})$.  In the first case, $w \in L(A') \subseteq L(A)$
    by the induction hypothesis.  In the second case, for each part of $w$,
    we have runs in the corresponding $A'$, and together with $\delta$ there is a run
    $i \trans{*} q \trans{\gamma} q' (\trans{*} q \trans{\gamma} q')^n \trans{*} f$
    for $w$ in $A$ (where the loop $q' \trans{*} q \trans{\gamma} q' $ is iterated $n$ times),
    so $w \in L(A)$.

  \item $L(r) \supseteq L(A)$: Let $w \in L(A)$, then there is an
    accepting run $\sigma$ for $w$ in $A$. Let
    run $\sigma$ pass through $\delta$
    (if it doesn't, then $w \in L(A') = L(r') \subseteq L(r)$ by the induction hypothesis).
    Then $w$ is of the form
    $w = w_1(\gamma w_{2,1}) \cdots (\gamma w_{2,n}) \gamma w_3$ with a run
    $\sigma = \sigma_1(\delta\sigma_{2,1}) \cdots (\delta\sigma_{2,n})\sigma_3$,
    where $\sigma_1$ is a run for $w_1$ in $A'_{(i,q)}$, for each $i \in \{1,\ldots,n\}$ we have
    $\sigma_{2,i}$ is a run for $w_{2,i}$ in $A'_{(q',q)}$, and $\sigma_3$ is a run for $w_3$ in
    $A'_{(q',F)}$. By the induction hypothesis, we have that
    $w_1 \in L(r'_{i,q})$, $w_{2,i} \in L(r'_{q',q})$ for each $i \in \{1,\ldots,n\}$, and
    $w_3 \in L(r'_{q',F})$, therefore $w \in L(r)$ by the definition of $L(r)$.\qedhere
  \end{enumerate}
\end{proofappendix}
\noindent 
In summary, regular deallocation expressions and D-NFAs are
equiexpressive.

\section{Determinization of NDA}\label{sec:determinization}


We show next that NDA can be determinized. Recall that NDA are
equivalent to RNNA under local freshness semantics
(\cref{prop:NDA2RNNA}), but RNNA cannot be determinized, as the next
example shows.

\begin{example}\label{E:nodetRNNA}
  Consider the data language of all words in which the last letter has
  occurred before. It is accepted by the following RNNA under local
  freshness semantics:
  \begin{center}
    \begin{tikzpicture}[scale=0.15]
    \tikzstyle{every node}+=[inner sep=0pt]
    \draw [black] (54.1,-28) circle (3);
    \draw (54.1,-28) node {$q_2$};
    \draw [black] (54.1,-28) circle (2.4);
    \draw [black] (23.9,-28) circle (3);
    \draw (23.9,-28) node {$q_1$};
    \draw [black] (39,-28) circle (3);
    \draw (39,-28) node {$q_2(a)$};
    \draw [black] (22.577,-25.32) arc (234:-54:2.25);
    \draw (23.9,-20.75) node [above] {$\lb a$};
    \fill [black] (25.22,-25.32) -- (26.1,-24.97) -- (25.29,-24.38);
    \draw [black] (37.677,-25.32) arc (234:-54:2.25);
    \draw (39,-20.75) node [above] {$\lb b$};
    \fill [black] (40.32,-25.32) -- (41.2,-24.97) -- (40.39,-24.38);
    \draw [black] (26.9,-28) -- (36,-28);
    \fill [black] (36,-28) -- (35.2,-27.5) -- (35.2,-28.5);
    \draw (31.45,-28.5) node [below] {$\lb a$};
    \draw [black] (42,-28) -- (51.1,-28);
    \fill [black] (51.1,-28) -- (50.3,-27.5) -- (50.3,-28.5);
    \draw (46.55,-28.5) node [below] {$a$};
    \end{tikzpicture}
  \end{center}
  Suppose that there were a deterministic RNNA~$A$ accepting the same
  language, and let~$k$ be the degree of~$A$. We first assume that the
  initial state of $A$ has empty support. Let $a_i \in \names$,
  $i = 0, \ldots k$, be pairwise distinct. Then every word
  $a_0a_1\cdots a_ka_i$\smnote{Ich meine, dass man hier eigentlich
    noch argumentieren muss, warum der letzte vom RNNA gelesene
    Buchstabe nicht $\newletter a_i$ sein kann. Nämlich weil man dann
    per $\alpha$-Äquivalenz Wörter der Länge $k+2$ mit allen
    Buchstaben verschieden in der Datensprache bekommt.}  is contained
  in the given data language, and so $A$ must have an accepting run of
  every word $\newletter a_0\newletter a_1 \cdots \newletter a_k a_i$,
  for $i = 0, \ldots, k$ (as the automaton clearly cannot accept
  $\newletter a_0\newletter a_1 \cdots \newletter a_k \newletter
  a_i$). Since $A$ is deterministic, it has only one run for the
  prefix $\newletter a_0\newletter a_1 \cdots \newletter a_k$ ending
  in the state $q$, say. By \cref{lem:supptrans}, the support of $q$
  must contain all letters $a_0, \ldots, a_k$ in order to be able to
  have an $a_i$-labelled transition to a final state for every
  $i = 0, \ldots, k$. This contradicts that $|\supp(q)| \leq k$.

  Now assume that the initial state $s$ of $A$ has non-empty support. Then the argument
  essentially remains the same, except that the above string could contain some letters
  without a bar (at most $|\supp(s)|$ many). 
\end{example}

However, there is a deterministic NDA accepting the data language from \cref{E:nodetRNNA};
just replace the labels at the loops of the RNNA with $\q{a}$ and $\q{b}$, respectively.

\takeout{
In order to clarify why RNNA are not determinizable,
we note that there are other mechanisms for data languages,
for instance bar-DFA, which have the favorable property of a
deterministic transition relation.
However, bar-DFA in particular suffer from not being closed
under $\alpha$-equivalent renamings.
This limitation arises from several factors: the finite state space and
transition relation and the lack of a notion of memory require that
$\alpha$-equivalent renaming must be applied a posteriori.
RNNA, in constrast, address this issue by weakening to orbit-finiteness
only, and by introducing new states and transitions in the name-dropping modification,
one closes the language of
the automata under all $\alpha$-equivalent renamings. However, this solution
requires proper non-determinism to achieve that closure. RNNA still possess
exculively two types of transitions, bound and free ones, where --- intuitively ---
the former allows to store a name
and the latter requires that the name on the transition is already present in memory.
During all transitions, RNNA can non-deterministically lose names from the memory,
just as NDA do, yet
they do not facilitate explicit deletions of names, and hence strictly require
this non-determinism
in order to maintain their expressivity.

However, our NDA subsume RNNA and hence admit the same non-determinism,
which materializes in
odd ways; for instance, they can exhibit allocating transitions $q\trans{\lb a} q'$
where the label $a$ is contained in $\supp(q)$, yet in $q'$ the name $a$ is
indeed dropped from the support.
Lest we have to deal with this unintuitive
behaviour later on, we bridge this gap by introducing a notion of discipline
to the behaviour of the automaton. In particular,
we require the automaton
to drop names only then, when the symbol on the transition indicates it.
This means that $\lb a$
always allocates $a$,
$ a\rb$ always deallocates $a$, a free $a$ transition always knows $a$ before and after,
and $\q{a}$ always allocates and deallocates $a$ immediately (for $a\rb$ and $\q{a}$, this
is the case by default).

In this regard, the following construction gives rise to a deterministic
automata model, and while this is already remarkable in the nominal setting generally,
it is particularly striking that deallocation is all one needs to achieve it.

As visualized in \cref{fig:DDASketch}, we establish a correspondence
between non-deterministic and
deterministic automata in terms of local freshness semantics:
we are going to provide a construction that converts an NDA $A$ into a
\emph{deterministic deallocation automaton} (DDA), i.e.\ an
NDA $A' = (Q,\Delta, i, F)$ where $\Delta$ is (partially) deterministic,
accepting the local freshness semantics of $A$ as data language: $\D(\la(A)) = L_D(A')$.
We recall that a transition relation
$\Delta$ is (partially) deterministic if for all $q \in Q$ and $\gamma \in \ah$,
there is at most one $q' \in Q$ such that $q \trans{\gamma} q'$.} 

\begin{figure*}[h]
  \centering
  \begin{tikzcd}[column sep=40, row sep=large]
    & & \D(\la(A)) \arrow[d,equal] &
    L_0(A) \arrow[d, "\text{\cref{th:nameDropModAlphaClosed}}"]\\
    L_0(A_{DDA}) \arrow[r, "\db"] &
    L_D(A_{DDA}) \arrow[r, equal, "\text{\cref{lem:l-intersection}}"'] \arrow[r, equal, "\cref{th:DDA}"]  \arrow[ru, equal] &
    L_D(A_{\bot}) \arrow[d, "\text{\cref{cor:restrictNDALD}}" right] &
    L_0(A_{\bot}) \arrow[l, "\db"'] \arrow[d, "\text{\cref{lem:restrictNDAL0}}" right]\\
    L_0(A_D)
      \arrow[r, "\db"]
      \arrow[u, "\text{\cref{lem:ADiscSameAlpha}}"']
      \arrow[u, "\clal"]
      & L_D(A_D) \arrow[r, equal, "\text{\cref{lem:discRestrictionSameData}}"]
      \arrow[u, "\text{\cref{lem:closingDDA}}"'] &
    L_D(A_S)
    & L_0(A_S)
      \arrow[l, "\db"']
  \end{tikzcd}
  \caption{Proof sketch of \cref{th:DDA}. }
  \label{fig:DDASketch}
\end{figure*}
\takeout{
In order to further clarify why RNNA are not determinizable, we note
that there are other mechanisms for data languages, for instance bar
DFA (the obvious deterministic variant of bar NFA mentioned in
\cref{sec:prel}), which have the favorable property of a deterministic
transition relation\lsnote{Is it really useful to talk about DFAs
  here? It seems to me that the issue discussed later, regarding
  non-deterministic name dropping, is more to the point}.  However,
the language of a bar DFA is not closed under $\alpha$-equivalent
renamings in general. This limitation arises from several factors: the
finite state space, the transition relation, and the lack of a notion
of memory require that $\alpha$-equivalent renaming must be applied a
posteriori.

In contrast, RNNA address this issue by admitting orbit-finite in lieu
of only finite state sets. Moreover, by introducing new states and
transitions in the name-dropping modification, one closes the language
of an RNNA under all $\alpha$-equivalent renamings. However, this
solution requires proper non-determinism to achieve that closure. RNNA
have only two types of transitions, bound and free ones. Intuitively,
the former allow storing a name, and the latter require that the name
on the transition is already present in memory.  During a transition,
an RNNA can non-deterministically drop names from its memory, and the
same holds for NDA (e.g.~there can be an allocating transition
$q\trans{\lb a} q'$, where $a$ is in $\supp(q)$ but it may or may not
be in $\supp(q')$).

However, unlike NDA, RNNA  do not allow for explicit deletions of
names from memory. Consequently, RNNA require this non-determinism in order to maintain their
expressivity.\smnote{Ich habe das Gefühl, dass diese Erklärungen viel früher gemacht werden
  sollten, bei einem Vergleich von NDA und RNNA, nahe der Einführung von NDA, weil sie
  nämlich NDA erst motivieren. Hier könnte dann kurz daran erinnern werden.}

To avoid the slightly unintuitive non-deterministic name dropping in
NDA, we introduce in \cref{con:disc} a notion of discipline to the
behaviour of the automaton. In particular, we require the automaton to
drop names only when the label of the transition indicates it.  This
means that $\lb a$ always allocates $a$, $ a\rb$ always deallocates
$a$, a free $a$ transition always keeps the known name~$a$ in memory,
and $\q{a}$ always allocates and deallocates $a$ immediately (for
$a\rb$ and $\q{a}$, this already holds in every NDA).

In this regard, \cref{con:disc} is the key to enabling determinization. As we have
already pointed out, having determinization is quite remarkable for any nominal automaton
model; it is particularly striking that deallocation is all one needs to achieve it.
}
Intuitively, RNNA cannot be determinized because
they rely on non-deterministic name dropping to maintain their expressivity:
since RNNA lack explicit deallocation,
names can only be removed from memory by non-deterministically dropping
them during transitions. NDA inherit this non-deterministic name
dropping (e.g., an allocating transition $q\trans{\lb a} q'$ may or may
not retain $a$ in $\supp(q')$), but crucially add the ability to drop
names explicitly via deallocating transitions.

To avoid the non-deterministic name dropping in NDA, we introduce in
\cref{con:disc} a notion of discipline on automaton behaviour.  In
particular, we require the automaton to drop names only when the label
of the transition indicates it: $\lb a$ always allocates $a$, $a\rb$
always deallocates $a$, a free $a$ transition always keeps the known
name~$a$ in memory, and $\q{a}$ always allocates and deallocates $a$
immediately (for $a\rb$ and $\q{a}$, this already holds in every NDA).

This discipline is the key to enabling determinization.  As we have
already pointed out, having determinization is quite remarkable for any
nominal automaton model; it is particularly striking that explicit
deallocation is all one needs to achieve it.\spnote{reworded and removed bar DFA}

We proceed to introduce our deterministic automaton model.

\begin{defn}
  A \emph{deterministic deallocation automaton} (DDA) is an NDA $A = (Q,\Delta, i, F)$ where
  $\Delta$ is \emph{(partially) deterministic}: for all $q \in Q$ and $\gamma \in \ah$,
  there is at most one $q' \in Q$ such that $q \trans{\gamma} q'$.
\end{defn}

We will establish a correspondence between non-deterministic and
deterministic automata in terms of local freshness semantics: we are
going to provide a construction that converts a given NDA~$A$ into a
DDA~$A'$ such that $\D(\la(A)) = L_D(A')$.
Here $A'$ is produced $\alpha$-closed (via name dropping in the overall construction),
so the statement applies uniformly to every NDA~$A$.


Since for the big majority of applications closed languages are of particular interest, 
at least for the purpose of
comparability to other automata models, we restrict our investigations to closed languages 
(like the one in \cref{E:nodetRNNA}).

\takeout{ 
This additionally carries the 
advantage of a construction that is widely standard (with a convincingly simple adjustment).
Eventually, this bears:
\begin{center}
  \emph{Closed regular data languages are deterministic.}
\end{center}
}

From this point on, we will only consider name-dropping modifications
of NDAs (e.g. an $S$-restriction is always one of a name-dropping
modification).  This ensures that all states required for applying our
disciplination policy, defined next, actually exist.

\begin{construction}[Disciplined $S$-Restriction]\label{con:disc}
  Given an NDA, fix a set $S \subseteq \names$ as in \cref{con:sRes}.
  Let $A_S = (Q,\Delta, i, F)$ be the $S$-restriction of the
  name-dropping modification of an NDA~$A$.  The \emph{disciplined
    $S$-restriction} of~$A$ is the NFA $A_D = (Q,\Delta_D,i,F_D)$ with
   transitions  given by
  \begin{align*}
      q \trans{\lb a} q'\text{ in~$A_D$}
      &\text{ if $q \trans{\lb a} q'$ in~$A_S$},& a \notin \lo_{A_S}(q),\quad&a \in \lo_{A_S}(q'),\\
      q \trans{a} q'\text{ in~$A_D$}
      &\text{ if $q \trans{a} q'$ in~$A_S$},& &a \in \lo_{A_S}(q'),\\
      q \trans{a\rb} q'\text{ in~$A_D$}
      &\text{ if $q \trans{a\rb} q'$ in~$A_S$, or}\\
      &\text{ if $q \trans{a} q'$ in~$A_S$},& &a \notin \lo_{A_S}(q'),\\
      q \trans{\q{a}} q'\text{ in~$A_D$}
      &\text{ if $q \trans{\q{a}} q'$ in~$A_S$, or }\\
      &\text{ if  $q \trans{\lb a} q'$ in~$A_S$,}& a\notin \lo_{A_S}(q),\quad &a \notin \lo_{A_S}(q').
    \end{align*}
    \takeout{
      \item $q \trans{\lb a} q'\text{ in~$A_D$}$ if $q \trans{\lb a} q'$ in~$A'$ and
              $a \in \lo(q')$,
      \item $q \trans{a} q'\text{ in~$A_D$}$ if $q \trans{a} q'$ in~$A'$, and
              $a \in \lo(q')$,
      \item $q \trans{a\rb} q'\text{ in~$A_D$}$ if $q \trans{a\rb} q'$ in~$A'$, 
              or if $q \trans{a} q'$ in~$A'$ and $a \notin \lo(q')$, and
            \item $q \trans{\q{a}} q'\text{ in~$A_D$}$ if
              $q \trans{\q{a}} q'$ in~$A'$, or if $q \trans{\lb a} q'$
              in~$A'$ and $a \notin \lo(q')$,} The set of final states
            is
            $F_D = \set{f \in F \mid \lo_{A_S}(f) =
              \emptyset}$.
\end{construction}

\begin{lemma}\label{lem:discSResDnfa}
  The disciplined $S$-restriction of an NDA is a D-NFA.
\end{lemma}

\begin{proofappendix}{lem:discSResDnfa}
  Fix $S\subset \names$. As of \cref{con:disc},
  let $A_S=(Q,\Delta,i,F)$ be
  the $S$-restriction of an NDA and $A_D = (Q, \Delta_D, i, F_D)$
  the disciplined $S$-restriction. We show that $A_D$ is a D-NFA.

  We have that $A_S$ is a D-NFA by \cref{lem:sResDnfa}. Let $q \in Q$, we
  show that $\lo_{A_D}(q) || \rc_{A_D}(q)$.

  We claim that:
  \[\lo_{A_D}(q) \subseteq \lo_{A_S}(q)\] and
  \[\rc_{A_D}(q) \subseteq \rc_{A_S}(q)
  \cup \names\setminus\lo_{A_S}(q) =: R(q).\]

  The sets $\lo_{A_S}(q)$ and $R(q)$ are trivially disjoint, so we are left to show that they meet the inequalities given in \cref{def:dnfa}.

  For $q \trans{\gamma}q'\in \Delta_D$, we distinguish cases on $\gamma$:

  \begin{itemize}
    \item $\gamma = \lb a$: We show that 
    $\lo_{A_S}(q)\supseteq \lo_{A_S}(q')\setminus \set{a}$ and
    $R(q)\setminus\set{a} \subseteq R(q')$. By \cref{con:disc}, we have $q \trans{\lb a}q'\in \Delta$.
    Since $A_S$ is a D-NFA, we immediately obtain $\lo_{A_S}(q) \supseteq \lo_{A_S}(q')\setminus\set{a}$.
    For the second inclusion, let $b\in R(q)\setminus\set{a}$.
    If $b\in \rc_{A_S}(q)\setminus\set{a}$, then $b\in \rc_{A_S}(q')\subseteq R(q')$ by \cref{def:dnfa}.
    Otherwise, we have $b\notin \lo_{A_S}(q)$ by \cref{con:disc} and $b\neq a$.
    Then we have $\lo_{A_S}(q')\subseteq \lo_{A_S}(q)\cup\set{a}$, hence $b\notin \lo_{A_S}(q')$, and
    thus $b\in \names\setminus \lo_{A_S}(q') \subseteq R(q')$.

    \item $\gamma = a$: We show that $\lo_{A_S}(q)\supseteq \lo_{A_S}(q')\cup\set{a}$ and
    $R(q)\setminus\set{a} \subseteq R(q')$. By \cref{con:disc}, we have $q \trans{a}q'\in \Delta$.
    Since $A_S$ is a D-NFA, the first inclusion is clear.
    Let $b\in R(q)\setminus\set{a}$. If $b\in \rc_{A_S}(q)\setminus\set{a}$, then $b\in \rc_{A_S}(q')\subseteq R(q')$ by \cref{def:dnfa}.
    Otherwise, we have $b \in \names\setminus\lo_{A_S}(q)$, i.e.~$b\notin \lo_{A_S}(q)$, and $b\neq a$. We have $\lo_{A_S}(q')\subseteq \lo_{A_S}(q)$, hence $b\notin \lo_{A_S}(q')$, thus we get $b\in \names\setminus \lo_{A_S}(q')\subseteq R(q')$.

    \item $\gamma = a\rb$: We show that $\lo_{A_S}(q)\supseteq \set{a}\cup \lo_{A_S}(q')$ and
    $R(q)\cup\set{a} \subseteq R(q')$.
    We distinguish cases on how such transitions arise in \cref{con:disc}:
    If $q\trans{a\rb}q'\in \Delta$, then both desired inclusions follow immediately from
    $A_S$ being a D-NFA.

    Otherwise, we have $q\trans{a}q'\in \Delta$ and $a\notin \lo_{A_S}(q')$.
    Since $A_S$ is a D-NFA, we have $\lo_{A_S}(q)\supseteq \lo_{A_S}(q')\cup\set{a}$.
    For $R(q)$, let $b\in R(q)\cup\set{a}$.
    If $b=a$, then $b\in \names\setminus \lo_{A_S}(q')\subseteq R(q')$ by $a\notin \lo_{A_S}(q')$.
    If $b\neq a$ and $b\in \rc_{A_S}(q)$, then $b\in \rc_{A_S}(q')\subseteq R(q')$ by
    the clause for $q\trans{a}q'$ in \cref{def:dnfa}.
    If $b\neq a$ and $b\notin \lo_{A_S}(q)$, then from $\lo_{A_S}(q)\supseteq \lo_{A_S}(q')\cup\set{a}$ we have $b\notin \lo_{A_S}(q')$.
    Hence $b\in \names\setminus \lo_{A_S}(q')\subseteq R(q')$.

    \item $\gamma = \q{a}$: We show that $\lo_{A_S}(q)\supseteq \lo_{A_S}(q')$ and
    $R(q)\cup\set{a} \subseteq R(q')$.
    Again, we distinguish the two cases from \cref{con:disc}.
    If $q\trans{\q{a}}q'\in \Delta$, then since $A_S$ is a D-NFA we have $\lo_{A_S}(q)\supseteq \lo_{A_S}(q')$ and $\rc_{A_S}(q)\cup\set{a}\subseteq \rc_{A_S}(q')$, and the desired inclusions follow.
    Otherwise, we have $q\trans{\lb a}q'\in \Delta$, $a\notin \lo_{A_S}(q)$, and $a\notin \lo_{A_S}(q')$.
    Since $A_S$ is a D-NFA and $q\trans{\lb a}q'$, we immediately have $\lo_{A_S}(q)\supseteq \lo_{A_S}(q')\setminus\set{a}=\lo_{A_S}(q')$.
    For $R(q)$, let $b\in R(q)\cup\set{a}$.
    If $b=a$, then $b\in \names\setminus \lo_{A_S}(q')\subseteq R(q')$ by $a\notin \lo_{A_S}(q')$.
    If $b\neq a$ and $b\in \rc_{A_S}(q)$, then $b\in \rc_{A_S}(q')\subseteq R(q')$ by \cref{def:dnfa}.
    If $b\neq a$ and $b\notin \lo_{A_S}(q)$, then from $\lo_{A_S}(q)\supseteq \lo_{A_S}(q')$ we get $b\notin \lo_{A_S}(q')$.
    Hence, we obtain $b\in \names\setminus \lo_{A_S}(q')\subseteq R(q')$.
  \end{itemize}
\end{proofappendix}

We say a D-NFA is disciplined, if its literal language is.\spnote{Dieses Statement oder das entsprechene Lemma verschieben}

\begin{lemma}\label{lem:discRestrictionSameData}
  The $S$-restriction of an NDA and its corresponding disciplined $S$-restriction accept the
  same data language. 
\end{lemma}

\begin{proofsketch}
  For $L_D(A_D)\subseteq L_D(A_S)$, note that every transition in $\Delta_D$ is obtained from a
  transition in $\Delta$ by possibly changing the binder and marker, but never the
  underlying name after applying $\db$; hence every run in $A_D$ corresponds to a run in $A_S$
  on a word with the same debracketing.
  For $L_D(A_S)\subseteq L_D(A_D)$, take an accepting run of $A_S$ on some word $w$ and transform
  it into a run of $A_D$ on $\disc(w)$ by inspecting each transition locally: whenever a name
  is not needed in the successor's support, we replace the step by a deallocation
  (or an unknown), otherwise the step is kept.
  Since $\db(\disc(w))=\db(w)$, this preserves the data word.
\end{proofsketch}

\begin{proofappendix}{lem:discRestrictionSameData}
  Let the data of \cref{con:disc} be given, with the same
  designators.
  
  "$L_D(A_D) \subseteq L_D(A_S)$": This inclusion is straightforward, since each transition
  $q\trans{\gamma}q' \in \Delta_D$ is induced by a transition $q\trans{\hat\gamma}q' \in \Delta$ such
  that $\db(\hat\gamma) = \db(\gamma)$, hence for each run on some $w$ in $A_D$ there is a
  run on some $w'$ in $A_S$ such that $\db(w') = \db(w)$.
    
  "$L_D(A_S) \subseteq L_D(A_D)$": For the opposite direction, we recall
  the representation of states in $Q$ inherited from the name-dropping
  modification as tuples $(o,r)$ where $r$ is a partial function
  $[n]\parto S$, with~$n$ being the size of the support of
  elements of the orbit~$o$ in the original NDA. We proceed via
  induction on words, and reinforce the
  claim: For each state $(o,r)$ in $Q$ and each smaller word $w$ accepted by
  $(o,r)$,
  there is a state $(o,r')\in Q$ such that $r'$ extends $r$,
  $\lo(r') \supseteq \lo(w)$ and $(o,r')$ accepts $\disc(w)$ in $A_D$
  (for which we have $\db(\disc(w)) = \db(w)$).

  Base case $w = \epsilon$: We have that $(o,r)$ is a final state in $A_S$ and $w = \epsilon$.
  Then $(o,r_\bot)$ where $r_\bot$ is the partial function undefined everywhere
  is a final state in $A_D$, $r$ extends $r_\bot$, $\supp(r_\bot) = \emptyset =
  \lo(w)$,
  and $(o,r_\bot)$ accepts $\disc(\epsilon) = \epsilon$.

  Induction step $w = \gamma u$: Let $(o,r)\in Q$ accept $w = \gamma u$. Then there is
  $(p,s)$ such
  that $(o,r)\trans{\gamma}(p,s) \in \Delta$ and $(p,s)$ accepts $u$ in $A_S$.
  By the induction hypothesis, there is $(p,s_\bot)\in Q$ accepting $\disc(u)$ with
  $\supp(s_\bot) = \lo(u)
  = \lo(\disc(u))$, and $s$ extends $s_\bot$. We distinguish cases:
  
  $\gamma = a$: By the definition of the name-dropping modification and $A_S$, there is
  $(o,r_\bot)\trans{a}(p,s_\bot)$ in $\Delta$ for some $r_\bot$ extended by $r$. By
  \cref{lem:supptrans},
  we have that $a \in \supp(r_\bot)$.
  If $a \in \supp(s_\bot) = \lo(\disc(u))$, then there is $(o,r_\bot)\trans{a}(p,s_\bot)$
  in $\Delta_D$ by the
  definition of $\Delta_D$; otherwise $(o,r_\bot)\trans{a\rb}(p,s_\bot) \in \Delta_D$.
  In the former case,
  $(o,r_\bot)$ accepts $a\disc(u) = \disc(w)$; in the latter, $(o,r_\bot)$ accepts
  $a\rb\disc(u) = \disc(w)$.

  $\gamma = \lb a$: By the definition of the name-dropping modification, there is
  $(o,r_\bot)\trans{\lb a}(p,s_\bot)$
  in $\Delta$ for some
  $r_\bot$ extended by $r$, and we choose $r_\bot$ with $a \notin \supp(r_\bot)$.
  If $a \in \supp(s_\bot) = \lo(\disc(u))$,
  then $(o,r_\bot)\trans{\lb a}(p,s_\bot) \in \Delta_D$ by \cref{con:disc},
  and $(o,r_\bot)$ accepts
  $\lb a\,\disc(u) = \disc(w)$. Otherwise, $a \notin \supp(s_\bot)$, hence
  $(o,r_\bot)\trans{\q{a}}(p,s_\bot) \in \Delta_D$
  by \cref{con:disc}, and $(o,r_\bot)$ accepts $\q{a}\,\disc(u) = \disc(w)$.

  $\gamma = a\rb$: Then $(o,r_\bot)\trans{a\rb}(p,s_\bot) \in \Delta$ for some
  $r_\bot$ extended by $r$.
  By the definition of $\Delta_D$ the same transition is in $\Delta_D$. Moreover,
  $a \notin \supp(s_\bot)$ by
  \cref{lem:supptrans}, and $\disc(w) = a\rb\disc(u)$, hence $(o,r_\bot)$ accepts
  $\disc(w)$.

  $\gamma = \q{a}$: Lastly, $(o,r_\bot)\trans{\q{a}}(p,s_\bot) \in \Delta$ for some
  $r_\bot$ extended by $r$,
  and by \cref{con:disc} also $(o,r_\bot)\trans{\q{a}}(p,s_\bot) \in \Delta_D$.
  Since $\disc(\q{a}u) = \q{a}\,\disc(u)$,
  $(o,r_\bot)$ accepts $\disc(w)$.\qedhere
  \end{proofappendix}

  We next consider the classical powerset construction, applied in the
  same manner as for ordinary NFA, where, crucially, we restrict the
  powerset automaton to its reachable part. Disciplined S-restrictions
  of an NDA are specific NFA over~$\hat{S}$; hence the standard
  powerset construction yields a DFA accepting the same literal
  \mbox{language}. A D-NFA is a \emph{deallocation DFA} (\emph{D-DFA})
  if its transition relation is a partial function.  We observe next
  that the conditions on D-NFA as per \cref{def:dnfa} are preserved by
  the powerset construction:

\begin{lemma}\label{lem:pow}
  The powerset construction preserves D-NFA.
\end{lemma}

\begin{proofappendix}{lem:pow}
  Let $A = (Q,\Delta,i,F)$ be a D-NFA.
  The powerset construction applied to $A$ yields the DFA $A' =
  (Q',\delta,\set{i},F_P)$, where $Q'\subseteq \pow(w)$
  is the set of reachable macrostates,
  $\delta\colon Q'\times \shs\to Q'$ is defined by
  \[
    (S, \gamma) \mapsto \set{q'\in Q \mid q\trans{\gamma}q'\in \Delta, q \in S},
  \]
  and we have $F_P = \set{S\in\pow(Q) \mid S \cap F \neq \emptyset}$.

  We prove that the DFA $A'$ is a D-NFA.
  Let $S \in Q'$ be a state in~$A'$. We claim that
  \[
    \textstyle
    \lo_{A'}(S) = \bigcup_{q\in S} \lo_A(q)
    \qquad\text{and}\qquad
    \rc_{A'}(S) = \bigcup_{q\in S} \rc_A(q).
  \]
  
  We have to show that (1) $\lo_{A'}(S) \cap \rc_{A'}(S) = \emptyset$ for every
  $S \in Q'$ and (2) $\lo_{A'}(S)$ and $\rc_{A'}(S)$ above adhere to the
  defining inclusions of D-NFAs.
  
  (1): Assume towards a contradiction that there is
  $a\in \rc_{A'}(S) \cap \lo_{A'}(S)$. Then there are states $q,q' \in S$ such
  that $a \in \rc_A(q)$ and $a \in \lo_A(q')$.
  Assume that $S$ has some incoming run $\sigma$ on $w$ such that
  $a\in\rc(w)$.  If there is none, then $\rc_A(S) = \emptyset$ holds,
  contradicting $a \in \rc_A(S)$ and we are done. By the powerset
  construction, all states $q$ in $S$ then have an incoming run on
  $w$, in particular $q'$, yielding $a \in \rc_A(q')$.  With this,
  $a \in \lo_A(q') \cap \rc_A(q')$ contradicts that $A$ is a D-NFA.

  (2): Let $S \trans{\gamma} S' \in \delta$; We showcase the inclusion for
  the right-closed names and for $\gamma = a$; all other cases are entirely
  analog. We have to show
  $(\bigcup_{q \in S} \rc_A(q))\setminus \set{a} \subseteq
  \bigcup_{q \in S'} \rc_A(q)$. Let $q \in S$ and $b\in \rc_A(q)$.
  We have to show that if $a \neq b$, then $b \in \rc_{A'}(S')$.
  Let $q' \in Q$ such that $q' \trans{a} p \in \Delta$ for $p\in S'$,
  which induces the given $S \trans{a} S' \in \delta$.
  We have that $b \in \rc_A(q)$ implies $b \in \rc_A(w)$
  for some $w$ with an incoming run to $q$. By the powerset construction,
  $q'$ also has an incoming run on $w$, and $b \in \rc_A(q')$. Since
  $A$ is a D-NFA, we have $\rc_A(q') \setminus \set{a} \subseteq \rc_A(p)$,
  thus $b \in \rc_A(p)$ and eventually $b \in \bigcup_{q \in S'} \rc_A(q) =
  \rc_{A'}(S')$.

  Finally, we obtain $L(A) = L(A')$ from the standard powerset construction
  for NFAs.
\end{proofappendix}

Now that we enforced that free and allocating transitions always store their label in the
support, which allows us to close the automaton under equivariance while preserving its alphatic
language, using the closure under equivariance given in \cref{con:nom}.

\begin{lemma}\label{lem:closingDDA}
  The nominalization of a disciplined D-DFA yields a DDA.
\end{lemma}
\begin{proofappendix}{lem:closingDDA}\spnote{fix this proof}
  Let $A = (Q,\Delta, i,F)$ be a disciplined D-DFA, and
  let $A' = (Q',\Delta',i,F')$ be its nominalization.
  In view of \cref{lem:nomNda,lem:pow}, we are left to verify that
  $\Delta'$ is deterministic, which follows from $\Delta$ being so:
  Let $q \in Q'$ and $\gamma\in \ah$,
  $q\trans{\gamma}q' \in \Delta$ is unique.
  We verify that \cref{con:nom} does not introduce additional
  transitions for $q$ on $\gamma$ and distinguish cases:
  \begin{itemize}
    \item $\gamma = a,\q{a},a\rb$:
      First, we have $\iota(q) \trans{\gamma} \iota(q')$ for $\delta:q \trans{\gamma} q'
      \in \Delta$ which is also unique since $\iota$ is injective.
      Then, since we have $\lo_(S\in Q) = \bigcup_{q\in S}\lo(q)$ from \cref{lem:pow}
      and for all $\delta \in \Delta$ we have $\lo(q')\subseteq\lo(q)$,
      hence the closure under equivariance does not produce new transitions
      starting from $q$.

    \item $\gamma = \lb a$:
      This case works similarly to the above case, we now have 
      $\iota(q) \trans{\lb a} \iota(q')$ for $\delta:q \trans{\lb c} q'' \in \Delta$
      where $\braket{a}q' = \braket{c}q''$. By \cref{con:nom}, we have
      $\lo(q'') \subseteq \lo(q) \cup \set{a}$, and in particular
      $a \in \lo(q'')$. Thus, neither the enforcement of left
      $\alpha$-invariance nor equivariance yield new transitions
      with $\lb a$.\qedhere
  \end{itemize}
\end{proofappendix}

\takeout{ 
  Here we apply the powerset construction as for NFA. Disciplined S-restricted NDA
  $A = (Q,\Delta, i, F)$ are specific NFA over $\shs$; hence the powerset construction
  yields a DFA $A' = (Q',\Delta',i,F')$ accepting the same literal language. Furthermore,
  for each transition $q\trans{\gamma}P$, the target $P$ is ufs; if $\gamma = a$ or
  $\gamma = \lb a$, then $a \in \supp(P)$, which allows us to close the automaton under
  equivariance and preserve the literal language. The closure under equivariance of a finite
  automaton $A = (Q,\Delta,i,F)$ is defined as
  $\cleq(A) = (\cleq(Q),\cleq(\Delta),i,\cleq(F))$.
}

Given an NDA $A$, we write $A_\mathsf{DDA}$ for this DDA and refer to it as a
\emph{determinization} of $A$.
\begin{rem}
  Note that this result fails for RNNA. While applying the powerset
  construction to a bar NFA does provide a bar DFA accepting the same
  bar language, the closure under equivariance of that bar DFA does
  not satisfy $\alpha$-invariance, and closing the transition relation
  under $\alpha$-invariance to obtain an RNNA yields a
  non-deterministic transition relation.\spnote{Hier vergleich mit
    DOFA aus Bojanczyk zitieren?}  \smnote{Hierzu müsste ein konkretes
    Beispiel präsentiert werden. Ansonsten ist hier eine unbewiesene
    Behauptung! Was kommt z.B.~bei dem RNNA in Fig.~4 raus (nach den
    einzelnen Schritten)?  Den bar NFA sehe ich, aber den bar DFA
    seine und closure unter equvariance sollte man zeigen sowie was
    dann zum RNNA noch fehlt und warum das wieder
    nichtdet.~wird. Lutz: Das weiß man aus dem einleitenden
    Beispiel}\lsnote{@Simon: Mal in Bollig et al. 2014 nachlesen, ob
    man session automata determinisieren kann.}
\end{rem}

\begin{lemma}[label=lem:suppDisc]
  Let $A_\mathsf{DDA} = (Q,\Delta, i, F)$ be a determinization.
  Then for every state $q\in Q$ and every word $w \in L_0(q)$, $\lo(w) = \supp(q)$.
\end{lemma}

\begin{proofsketch}
  Let $w\in L_0(q)$ be accepted via a run $q \trans{*} f$.  Then the
  inclusion $\lo(w) \subseteq \supp(q)$ is immediate from the support
  lemma (\Cref{lem:supptrans}); we show $\supp(q) \subseteq \lo(w)$.
  So let $a\in\supp(q)$. Now final states in~$A_\mathsf{DDA}$ have
  empty support, so~$a$ must be removed from the support somewhere
  along the run, which by \cref{con:disc} can only happen on an
  $a\rb$-step. Hence,~$a$ occurs in~$w$, and its first occurrence
  cannot be $\lb a$ or $\q{a}$, so it is either $a$ or $a\rb$,
  i.e.~$a\in\lo(w)$.
\end{proofsketch}

\begin{proofappendix}{lem:suppDisc}
  Let $w\in L_0(q)$ be accepted via a run $q \trans{*} f$.  Then the
  inclusion $\lo(w) \subseteq \supp(q)$ is immediate from the support
  lemma (\Cref{lem:supptrans}); we show $\supp(q) \subseteq \lo(w)$.

  We recall that in the nominalization step of the construction of
  $A_\mathsf{DDA}$, states are constructed from a state~$q$ in the the
  determinized disciplined $S$-restriction constructed in the previous
  step in such a way that the support receives size $|\lo(q)|$. Since
  \Cref{con:disc} imposes $\lo(f')=\emptyset$ on final states~$f'$ in
  the disciplined $S$-restriction, final states in $A_\mathsf{DDA}$
  thus have empty support\lsnote{@Simon: Ich hoffe, ich gebe das hier
    korrekt wieder. Da passt aber was nicht -- in der
    Powerset-Konstruktion werden ja durchaus finale Zustände mit
    nicht-finalen zusammengeworfen, und da das Wort ja immer noch
    weitergehen könnte, lässt sich das nicht so ohne weiteres
    umgehen. Ich habe das Gefühl, man muss da durchaus noch
    aufwändiger vorgehen: Man zitiert ein vorher hoffentlich schon mal
    bewiesenes Lemma, demzufolge ein zu~$w$ $\alpha$-äquivalentes Wort
    vom unterliegenden D-NFA akzeptiert wird, und sagt dann, dass
    akzeptierende Runs im Powerset-Automaten immer einen
    akzeptierenden Run im ursprünglichen Automaten induzieren
    (deswegen ist die Powerset-Konstruktion korrekt). Mit dem Run
    argumentiert man dann weiter. Alternativ zeigt man, dass
    Nominalisierung und Powerset-Konstruktion kommutieren. Ah, siehe
    Mattermost: Für erreichbare Makrozustände tritt das Problem nicht
    auf}; in particular, $\supp(f)=\emptyset$.

  
  Now let $a\in\supp(q)$.  Since $a\notin\supp(f)$, there is some
  transition in the given run $q \trans{*} f$ where $a$ is removed
  from the support for the first time (it may be reintroduced and
  removed again any number of times later in the run). By
  \Cref{con:disc} (disciplined $S$-restriction) and \cref{con:nom}
  (nominalization)\lsnote{@Simon: Hier wurde ursprünglich das Support
    Lemma zitiert. Aber das liefert ja keine obere Abschätzung auf die
    in einer Transition entfernten Namen, was schon deswegen ja auch
    nicht geht, weil es für beliebige NDA gilt}, the only possible
  label of such a transition is~$a\rb$, and all previous transitions
  in the run whose label mentions~$a$ must be $a$-transitions
  (since~$\lb a$ and~$\q{a}$ both require that~$a$ is not in the
  support of the prestate of the transition, and~$a\rb$ would have
  removed~$a$ from the support). Thus, $a\in\lo(w)$, as required.
\end{proofappendix}

\begin{lemma}\label{lem:ADiscSameAlpha}
  Let $A_D$ be the disciplined $S$-restriction of an NDA, and let $A_\mathsf{DDA}$ be the
  corresponding determinization (i.e.~the nominalization of the powerset construction of
  $A_D$). 
  \begin{enumerate}
  
    
  \item \label{lem:ADiscSameAlpha:1}The literal language of $A_\mathsf{DDA}$ is closed under $\alpha$-equivalence.
  \item \label{lem:ADiscSameAlpha:2}The alphatic languages of $A_D$ and $A_\mathsf{DDA}$ coincide.
  \item \label{lem:ADiscSameAlpha:3}The literal
  language of $A_\mathsf{DDA}$ is the closure of that of $A_D$ under
  $\alpha$-equivalence:
  $L_0(A_\mathsf{DDA}) = \clal(L_0(A_D))$.
  \end{enumerate}
  
%
\end{lemma}

\begin{proofsketch}
  \ref{lem:ADiscSameAlpha:1}: Immediate from
  \Cref{lem:alphaClosedSupp,lem:suppDisc}.
  \ref{lem:ADiscSameAlpha:2}: By \cref{lem:nomSameAlpha}.
  \ref{lem:ADiscSameAlpha:3}: Immediate from \ref{lem:ADiscSameAlpha:1} and \ref{lem:ADiscSameAlpha:2}.
\end{proofsketch}

\begin{proofappendix}{lem:ADiscSameAlpha}
  \sloppy
  \ref{lem:ADiscSameAlpha:1}: Immediate from
  \Cref{lem:alphaClosedSupp,lem:suppDisc}.

  \ref{lem:ADiscSameAlpha:2}: \lsnote{@Simon: Warum müssen wir hier
    noch etwas beweisen? Die Powerset-Konstruktion erhält sogar die
    Literalsprache, und die Nominalisierung erhält die alphatische
    Sprache. Was fehlt?}Let $A_D = (Q_D,\Delta_D,i,F_D)$ the
  disciplined $S$-restriction of a given NDA, and let
  $A_\mathsf{DDA} =
  (Q_\mathsf{DDA},\Delta_\mathsf{DDA},i,F_\mathsf{DDA})$ be its
  determinization. We prove
  $L_\alpha(A_\mathsf{DDA}) = L_\alpha(A_D)$.

  Since $A_D$ is a subautomaton of $A_\mathsf{DDA}$, we have
  $L_\alpha(A_\mathsf{DDA}) \supseteq L_\alpha(A_D)$ immediately.

  For the reverse inclusion, we show
  $L_\alpha(A_\mathsf{DDA}) \subseteq L_\alpha(A_D)$, and proceed by induction on the length
  of an accepting run in $A_\mathsf{DDA}$, and reinforce the inductive claim: for all
  $q \in Q_\mathsf{DDA}$ and all permutations $\pi$ such that $p := \pi \cdot q \in Q_D$, we
  have $\pi \cdot L_\alpha(q) = L_\alpha(p)$.

  For the base case, if $q \in F_\mathsf{DDA}$ and $p=\pi \cdot q \in Q_D$, then $q \in \cleq(F_D)$ implies
  $p \in F_D$, thus $\epsilon \in L_0(q)$ iff $\epsilon \in L_0(p)$.

  For the inductive step, let $w = \gamma v$ be a word accepted by
  $q \in Q_\mathsf{DDA}$ with a run of length $n$, and let $p=\pi \cdot q \in Q_D$
  be some state in $A_D$ for some permutation $\pi$.
  We show that there is some $w' \aeq \pi \cdot w$ accepted by $p$.
  We have
  $q \trans{\gamma} r \in \Delta_\mathsf{DDA}$ with
  $v$ accepted by $r$, and proceed with a case distinction on $\gamma$, of which most are straightforward:
  \begin{itemize}
    \item $\gamma = a$ or $\gamma = \lb a$: We have that $p \trans{\delta} r' \in \Delta_D$ with $\delta = \pi \cdot \gamma$
     and $r' = \pi \cdot r$. By the induction hypothesis, $r'$ accepts some $v' \aeq \pi \cdot v$,
     hence $p$ accepts $\delta v' \aeq (\pi \cdot \gamma) (\pi \cdot v) = \pi \cdot (\gamma v) = \pi \cdot w$
     by rule 2 of \cref{def:Alpha}.
    \item $\gamma = a \rb$: This case is similar to the first; we additionally verify that $\pi(a) \notin \lo(v')$.
    Since we have $a \notin \lo(v)$ by the definition of $\alpha$-equivalence,
    $\pi(a) \notin \lo(\pi \cdot v)$,
    and thus also $\pi(a) \notin \lo(v')$.
    \item $\gamma = \q{a}$: We have that $p \trans{\delta} r' \in \Delta_\mathsf{DDA}$ with $\delta = \q{\pi(a)}$ and $r' = \pi \cdot r$.
    Although we have $p,r' \in Q_D$, we have $a \notin \supp(p) \cup \supp(r')$ for $\q{a}$-transitions, so choose $b \in S$ such that $b \notin \supp(p) \cup \supp(r')$.
    This $b$ exists by choice of $S$ with $|S| = \dg(A) +1$ and we then have that $(ab)\cdot (p \trans{\delta} r') = p \trans{\q{b}} r' \in \Delta_D$. 
    By the induction hypothesis, $r'$ accepts some $v' \aeq \pi \cdot v$, hence $p$ accepts
    $\q{b}v' \aeq \q{\pi(a)}v' \aeq \pi \cdot (\q{a}v) = \pi \cdot w$
    by the rule for local freshness and transitivity of $\alpha$-equivalence.
  \end{itemize}
  Instantiating the claim with $q=i$ and $\pi=\id$ yields $L_\alpha(A_\mathsf{DDA})=L_\alpha(A_D)$.
  \ref{lem:ADiscSameAlpha:3}: This is immediate from \ref{lem:ADiscSameAlpha:1} and \ref{lem:ADiscSameAlpha:2}.
\end{proofappendix}

\takeout{
The nominalization of a disciplined $S$-restriction closes its literal language under
$\alpha$-equivalence.
\begin{lemma}\label{lem:closingDisc}
  Let $A_D$ be the disciplined $S$-restriction of an NDA, and let
  $A_\mathsf{DDA}$ be the corresponding determinization. The literal
  language of $A_\mathsf{DDA}$ is the closure of that of $A_D$ under
  $\alpha$-equivalence:
  $L_0(A_\mathsf{DDA}) = \clal(L_0(A_D))$.
\end{lemma}
\noindent
Indeed, this follows immediately from \cref{lem:ADiscSameAlpha}.\lsnote{@Simon: Add to \cref{lem:ADiscSameAlpha}}
}


\begin{lemma}\label{lem:l-intersection}
  Let $L$ and $L'$ be literal languages, and let $S \subseteq \names$
  be a set of names such that $\clal(L'\cap \shs) = L'$ and
  $\clal(L \cap \shs) = L$. If
  \[
    L'\cap \shs = \set{\disc(w)\mid w \in L \cap \shs}, 
  \]
  then $\db[L] = \db[L']$.\lsnote{Maybe this can be worded referring
    to alphatic rather than literal languages?}
\end{lemma}

\begin{proofappendix}{lem:l-intersection}
  We have
  \begin{align*}
    \db[L'] & =\textstyle \bigcup_{w'\in L'\cap\shs}\D(\{[w']_\alpha\}) && \by{hypothesis}\\
            & =\textstyle \bigcup_{w\in L\cap\shs}\D(\{[\disc(w)]_\alpha\}) && \by{hypothesis} \\
            & =\textstyle \bigcup_{w\in L\cap\shs}\D(\{[w]_\alpha\}) && \by{\cref{lem:discSameData}}\\
    & = \db[L] && \qedhere
  \end{align*}
  
\end{proofappendix}

\begin{lemma}\label{lem:discSameLoWQ}
  Let $A = (Q, \Delta, i, F)$ be a
  disciplined $S$-restriction.
  Then for each $q \in Q$ and $w \in L(q)$ we have $\lo(q) = \lo(w)$.
\end{lemma}

\begin{proofappendix}{lem:discSameLoWQ}
  Let $A = (Q, \Delta, i, F)$ be a
  disciplined $S$-restriction.
  We show that for each $q \in Q$ and $w \in L(q)$ we have $\lo(q) = \lo(w)$.
  To this end, we show that we have
  \[\lo(q)\supseteq\bigcup_{w \in L(q)} \lo(w)\]
  and
  \[\lo(q) \subseteq \bigcap_{w \in L(q)} \lo(w).\]

  The inclusion $\lo(q) \supseteq \bigcup_{w \in L(q)} \lo(w)$ holds by \cref{lem:cnDnfaCoinc}.
  We show $\lo(q) \subseteq \bigcap_{w \in L(q)} \lo(w)$.
  Let $a \in \lo(q)$ and let $w \in L(q)$; we show $a \in \lo(w)$.
  We proceed by induction on the length of an accepting run for $w$ from $q$.
  If $q \in F$, then $\lo(q) = \emptyset$ by the definition of final states in \cref{con:disc}, so there is nothing to show.
  Otherwise, let $q \trans{\gamma} q' \trans{*} f$ be an accepting run for $w = \gamma w'$ with $f \in F$.
  By the induction hypothesis, $a \in \lo(q')$ implies $a \in \lo(w')$.
  We distinguish cases on $\gamma$:
  \begin{itemize}
  \item $\gamma = b$: By \cref{con:disc}, we have $b \in \lo(q')$.
    By \cref{def:dnfa}, we have $\lo(q) \supseteq \lo(q') \cup \{b\}$.
    If $a = b$, then $a \in \lo(bw') = \lo(w') \cup \{b\}$ by \cref{tab:namedefs}.
    If $a \neq b$, then $a \in \lo(q')$, hence $a \in \lo(w')$ by the induction hypothesis, and thus $a \in \lo(bw') = \lo(w') \cup \{b\}$.
  \item $\gamma = \lb b$: By \cref{con:disc}, we have $b \notin \lo(q)$ and $b \in \lo(q')$.
    By \cref{def:dnfa}, we have $\lo(q) \supseteq \lo(q') \setminus \{b\}$.
    Since $a \in \lo(q)$ and $b \notin \lo(q)$, we have $a \neq b$ and $a \in \lo(q')$.
    By the induction hypothesis, $a \in \lo(w')$, and thus $a \in \lo(\lb b\, w') = \lo(w') \setminus \{b\}$ by \cref{tab:namedefs}.
  \item $\gamma = b\rb$: By \cref{con:disc}, either $q \trans{b\rb} q' \in \Delta$ or $q \trans{b} q' \in \Delta$ with $b \notin \lo(q')$.
    In both cases, \cref{def:dnfa} yields $\lo(q) \supseteq \lo(q') \cup \{b\}$.
    If $a = b$, then $a \in \lo(b\rb w') = \lo(w') \cup \{b\}$ by \cref{tab:namedefs}.
    If $a \neq b$, then $a \in \lo(q')$, hence $a \in \lo(w')$ by the induction hypothesis, and thus $a \in \lo(b\rb w') = \lo(w') \cup \{b\}$.
  \item $\gamma = \q{b}$: By \cref{con:disc}, either $q \trans{\q{b}} q' \in \Delta$ or $q \trans{\lb b} q' \in \Delta$ with $b \notin \lo(q)$ and $b \notin \lo(q')$.
    In both cases, \cref{def:dnfa} yields $\lo(q) \supseteq \lo(q')$.
    Hence $a \in \lo(q')$, so $a \in \lo(w')$ by the induction hypothesis.
    By \cref{tab:namedefs}, we have $\lo(\q{b}w') = \lo(w') \setminus \{b\}$.
    Since $b \notin \lo(q')$ and $a \in \lo(q')$, we have $a \neq b$, and thus $a \in \lo(\q{b}w')$.\qedhere
  \end{itemize}
\end{proofappendix}

\begin{lemma}\label{lem:l0disciplineLiteral}
  Let $A_S$ be the $S$-restriction and $A_D$ the disciplined
  $S$-restriction of the name-dropping modification of an NDA~$A$ as
  per \cref{con:disc}. Then $L_0(A_D)=\disc[L_0(A_S)]$.
\end{lemma}

\begin{proofappendix}{lem:l0disciplineLiteral}
  Recall that $A_S = (Q,\Delta,i,F)$ and $A_D = (Q,\Delta_D,i,F_D)$
  share the same state set~$Q$, and that states $q = (o, r) \in Q$
  hail from a name-dropping modification as in \cref{def:ndm}.

  $L_0(A_D) \subseteq \disc[L_0(A_S)]$: We prove by induction on $w$
  that for every $q \in Q$, if $q$ accepts $w$ in $A_D$, then $q$
  accepts some $w'$ in $A_S$ such that $\disc(w') = w$.
  
  For the base case, let $q \in F_D$. By definition of $F_D$, we have $q \in F$. Thus $q$ accepts $\epsilon$ in $A_S$, and $\disc(\epsilon) = \epsilon$.

  Inductive step: Let $q$ accept $\gamma w$ in $A_D$ via
  $q \trans{\gamma} q' \in \Delta_D$ where~$q'$ accepts~$w$. By the
  induction hypothesis, $q'$ accepts some $w'$ in~$A_S$ such that
  $\disc(w') = w$. We distinguish cases on $\gamma$:
    \begin{itemize}
    \item $\gamma = a$: By \cref{con:disc}, we have $q \trans{a} q'$
      in~$A_S$ and $a \in \lo_{A_S}(q')$. Thus, $q$ accepts $aw'$ in
      $A_S$. Since moreover $\lo(q') = \lo(w) = \lo(w')$ by
      \cref{lem:discSameLoWQ,lem:disc}, we have $a \in \lo(w')$, hence
      $\disc(aw') = a\disc(w') = aw$.
    \item $\gamma = a\rb$: By \cref{con:disc}, we have either $q \trans{a\rb} q' \in \Delta$, or
      $q \trans{a} q' \in \Delta$ with $a \notin \lo(q')$. In the former case,
      $a \notin \lo(w')$ since $w$ is right-non-shadowing, so $\disc(a\rb w') = a\rb\disc(w') = a\rb w$.
      In the latter case, $a \notin \lo(w')$ since $\lo(q') = \lo(w')$, hence
      $\disc(aw') = a\rb\disc(w') = a\rb w$. In both cases, $q$ accepts a word $v$ in $A_S$
      with $\disc(v) = a\rb w$.
    \item $\gamma = \lb a$: By \cref{con:disc}, we have $q \trans{\lb a} q' \in \Delta$ with
      $a \notin \lo(q)$ and $a \in \lo(q')$. Since $a \in \lo(q') = \lo(w')$, we have
      $\disc(\lb a\, w') = \lb a\disc(w') = \lb a\, w$. Thus $q$ accepts $\lb a\, w'$ in $A_S$.
    \item $\gamma = \q{a}$: By \cref{con:disc}, either $q \trans{\q{a}} q' \in \Delta$ or
      $q \trans{\lb a} q' \in \Delta$ with $a \notin \lo(q)$ and $a \notin \lo(q')$.
      In the first case, $\disc(\q{a}w') = \q{a}\disc(w') = \q{a}w$.
      In the second case, $a \notin \lo(w')$ since $\lo(q') = \lo(w')$, hence
      $\disc(\lb a\, w') = \q{a}\disc(w') = \q{a}w$. In both cases, $q$ accepts a word $v$
      in $A_S$ with $\disc(v) = \q{a}w$.\qedhere
    \end{itemize}

    $\disc[L_0(A_S)] \subseteq L_0(A_D)$: We prove by induction on
    words that for every $(o,r) \in Q$, if $(o,r)$ accepts $w$ in
    $A_S$, then the (unique) $(o,r_\bot)\in Q$ such that
    $\supp(r_\bot) = \lo(w)$ holds and $r$ extends $r_\bot$ accepts
    $\disc(w)$ in $A_D$.
  
    Base case: Let $(o,r) \in F$ accept $\epsilon$ in $A_S$. Since
    $\lo(\epsilon) = \emptyset$, we have~$r_\bot$ being undefined
    everywhere, so $\supp(r_\bot) = \emptyset$.  Since
    $\lo(o, r_\bot) = \emptyset$ by \Cref{lem:cnDnfaCoincInv}, we
    have $(o,r_\bot) \in F_D$. Thus $(o,r_\bot)$ accepts
    $\disc(\epsilon) = \epsilon$ in $A_D$.

  Inductive step: Let $(o,r)$ accept $\gamma w$ in $A_S$ via $(o,r) \trans{\gamma} (p,s) \in \Delta$
  with $(p,s)$ accepting $w$. By the induction hypothesis, $(p,s_\bot)$ accepts $\disc(w)$ in $A_D$. 
  We proceed via case distinction on $\gamma$ and showcase $\gamma = a$; all other cases are entirely analogous. By the definition of the name-dropping modification, there is
  $(o,r_\bot) \trans{a} (p,s_\bot)$ in $\Delta$ for some $r_\bot$ extended by $r$.
  If  $a \in \lo((p,s_\bot)) = \lo(\disc(w))$, then by \cref{con:disc},
  $(o,r_\bot) \trans{a} (p,s_\bot) \in \Delta_D$. Thus $(o,r_\bot)$ accepts
  $a\disc(w) = \disc(aw)$ in $A_D$. If instead $a \notin \lo((p,s_\bot))$, then we have $a \notin \lo(\disc(w))$ by \cref{lem:discSameLoWQ}, and $(o,r_\bot)\trans{a\rb}(p,s_\bot) \in \Delta$ by construction, therefore $(o,r_\bot)$ accepts $a\rb\disc(w) = \disc(aw)$ in~$A_D$.

  For the initial state $i$, we have $\lo(i) = \emptyset$ and for each word $w \in L(A_S)$, we have $\lo(w) = \emptyset$. Hence, in particular,~$i$ accepts $w$ in $A_D$.
\end{proofappendix}

We now have all preparations in hand to prove the main result of this
section. In summary, the resulting determinization procedure for NDA
has the following steps: Given an NDA $A = (Q, \Delta, i, F)$, fix a
set $S \subseteq \names$ such that $\supp(i) \subseteq S$ and
$|S| = \deg(A_\bot) + 1$.
\begin{enumerate}
\item Form the name-dropping modification $A_\bot$ of $A$ (\cref{def:ndm});
\item Form the $S$-restriction $A_S$ of $A_\bot$ (\cref{con:sRes});
\item Form the corresponding disciplined $S$-restriction $A_D$ (\cref{con:disc}); 
\item Perform the powerset construction to obtain a DFA that is equivalent to the (D-)NFA
  $A_D$;
\item Form the nominalization $A_\mathsf{DDA}$ of that DFA (\cref{con:nom}).
\end{enumerate}

\begin{theorem}[name=Determinizability of NDA, label=th:DDA]
  For every NDA $A$, there exists a DDA $A_{\mathsf{DDA}}$
  such that $L_D(A_{\mathsf{DDA}}) = \D(\la(A))$.
\end{theorem}

\begin{proofsketch}
  We chase the proof through \cref{fig:DDASketch}: Starting from~$A$,
  first apply the name-dropping modification in order to obtain
  $A_\bot$.  By \cref{th:nameDropModAlphaClosed},
  $L_0(A_\bot)=\clal(L_0(A))$, hence $\D(\la(A))=L_D(A_\bot)$.
  Restrict to a finite set~$S$ of names using the
  $S$-restriction~$A_S$; this preserves the data language over $S$ by
  \cref{lem:restrictNDAL0,cor:restrictNDALD}.  Next, discipline the
  automaton to obtain $A_D$; this does not change the data language by
  \cref{lem:discRestrictionSameData}.  Determinize $A_D$ via the
  powerset construction (preserving the D-NFA property by
  \cref{lem:pow}) and then nominalize; by \cref{lem:closingDDA} the
  result is a DDA.  This closes the literal language under
  $\alpha$-equivalence (\cref{lem:ADiscSameAlpha}) and
  we finally lift equality from the $S$-restriction back to the full
  alphabet via \cref{lem:l-intersection}.
\end{proofsketch}

\begin{proofappendix}{th:DDA}
Let $A$ be an NDA. As sketched in \cref{fig:DDASketch}, we follow the construction
from~$A$ to a DDA $A_{\mathsf{DDA}}$\spnote{funny palindrome: A DDA A-DDA höhö}.
Let $A_{\bot}$ be the name-dropping modification of~$A$ by \cref{def:ndm}.
By \cref{th:nameDropModAlphaClosed},
we have 
  \[L_0(A_{\bot})=\clal(L_0(A))\qquad \text{and hence} \qquad \D(\la(A)) = L_D(A_{\bot}).\]
Fix $S\subseteq\names$ with $\supp(i)\subseteq S$ and $|S|=\deg(A_\bot)+1$, and let $A_S$ be the
$S$-restriction of~$A_{\bot}$ by \cref{con:sRes}. By \cref{lem:restrictNDAL0,cor:restrictNDALD},
we have \[L_0(A_S) = L_0(A_{\bot})\cap \hat S^{*}\qquad \text{and} \qquad L_D(A_S) = S^*\cap L_D(A_{\bot}).\]
Let $A_D$ be the disciplined $S$-restriction by \cref{con:disc}.
By \cref{lem:discRestrictionSameData,lem:l0disciplineLiteral},
we have \[L_D(A_D)=L_D(A_S)\qquad \text{and} \qquad L_0(A_D)=\disc[L_0(A_S)].\]
Let $A_{\mathsf{DDA}}$ be the nominalization of the powerset DFA of~$A_D$
(\cref{lem:closingDDA}). Then
\[
  L_0(A_{\mathsf{DDA}})=\clal(L_0(A_D))
\]
by \cref{lem:ADiscSameAlpha}, and
\[
  L_D(A_{\mathsf{DDA}})
  = \D(\la(A_{\mathsf{DDA}}))
  = \D(\la(A_D))
\]
since $\la(A_{\mathsf{DDA}})=\la(A_D)$ by \cref{lem:ADiscSameAlpha}.
Then, \cref{lem:discSameData} yields
$\D(\la(A_D))=\D(\la(A_S))$, so
\[
  L_D(A_{\mathsf{DDA}})=\D(\la(A_S)).
\]
Intersecting with~$S^*$ gives
\[
  L_D(A_{\mathsf{DDA}})\cap S^*
  = \D(\la(A_S))\cap S^*
  = L_D(A_S)
  = S^*\cap L_D(A_{\bot}).
\]
Set $L:=L_0(A_{\bot})$ and $L':=L_0(A_{\mathsf{DDA}})$.
Then $\clal(L\cap\shs)=L$ and $\clal(L'\cap\shs)=L'$ as shown above.
Moreover, we have $L'\cap\shs=\set{\disc(w)\mid w\in L\cap\shs}$.
Thus, \cref{lem:l-intersection} applies and yields $\db[L]=\db[L']$, hence eventually
we obtain the desired equality
\[
  \D(\la(A))=L_D(A_{\bot})=L_D(A_{\mathsf{DDA}}).\qedhere
\]

\end{proofappendix}

\begin{cor}
  An RNNA whose initial state has empty support is determinizable as NDA regarding data
    languages.
\end{cor}

\takeout{
\section{Relation to Other Formalisms}\label{sec:others}
We briefly relate NDA and regular deallocation transitions
to several established formalisms for
languages over infinite alphabets with binding.

\paragraph*{Regular nondeterministic regular automata (RNNA)}

As highlighted in \cref{sec:intro,fact:rnnaNda,prop:NDA2RNNA}, NDA build upon the
concept of RNNA, which they subsume structurally and to which they are
expressively equivalent. Analogous to \cref{prop:NDA2RNNA},
by the Kleene theorems of RNNA and NDA, we obtain that
regular bar expressions and regular deallocation expressions are
expressively equivalent.

\paragraph*{Dynamic sequences and bracket algebra}
Dynamic sequences~\cite{GabbayEA15} and bracket algebra~\cite{brunet_et_al:LIPIcs:2019:10683}
employ explicit creation and destruction
markers, typically written as $\langle a$ and $\rangle a$ or $\langle_a$ and $\rangle_a$,
respectively, for $a\in\names$. Dynamic sequences and our D-Sequences are syntactically
similar but incommensurable en détail: Dynamic sequences admit constants that are not
subject to renaming, while D-sequences instead admit scoping from both sides for
a single letter.

Bracket algebra consists of words over the extended alphabet $A_\braket{} = \set{a,\langle_a,\rangle_a\mid a\in\names}$
accompanied by renaming of bilaterally bound letters into fresh or bound letters that does not introduce shadowing.
In contrast to dynamic sequeces and D-sequences, in bracket algebra a symbol $\langle_a$ does not indicate
the actual appearance of $a$ in the data language, but only the fact that $a$ is bound (from the left).
With this, there is a pointwise translation $t:\ah\to A_\braket{}$ defined by
\[
  t(a) = a,\quad t(\lb a) = \langle_a a,\quad t(\rb_a) = a \rangle_a,\quad t(\q{a}) = \langle_a a. \rangle_a,
\]

such that}

\section{Conclusions}

We have presented a novel automata model, namely non-deterministic deallocation automata
(NDA), for data languages with binders that extend regular non-deterministic nominal
automata (RNNA) with deallocating and unknown transitions. The new model is built in the
framework of nominal sets. Utilizing a representation as non-deterministic finite automaton,
we have defined regular expressions for NDA and proved a Kleene
theorem for them: regular expressions are equiexpressive with NDA (under the literal
language semantics). Moreover, we have shown that it is possible to determinize NDA
accepting closed languages under preservation of their data language, which is a rare
feature in the realm of nominal automata.

Future work includes investigations on algebraic and coalgebraic semantics for the 
languages and automaton model, respectively, as well as studying the complexity of 
related decision problems, and extending the framework to infinite words.

\label{maintextend} 
\begin{acks}
We thank Thorsten Wißmann for helpful discussions and feedback on an earlier draft of this paper.
\end{acks}
\bibliographystyle{ACM-Reference-Format}
\bibliography{coalgml}

\newpage \appendix

\ifthenelse{\boolean{proofsinappendix}}{%
\section{Omitted Proofs and Further Details}
\label{sec:proofappendix}%
\closeoutputstream{proofstream}
\input{\jobname-proofs.out}
}{}

%
%
%

\end{document}